\documentclass[journal]{IEEEtran}
\usepackage[boxruled]{algorithm2e}
\usepackage{cite}
\usepackage{graphicx}
\usepackage{psfrag}
\usepackage{subfigure}
\usepackage{url}
\usepackage{amsmath}
\usepackage{array}
\usepackage{amssymb}
\usepackage{amsfonts}
\usepackage{graphicx}
\usepackage{algorithmic}

\newtheorem{lemma}{Lemma}

\newtheorem{corollary}{Corollary}
\newtheorem{remark}{Remark}
\newtheorem{definition}{Definition}
\newtheorem{theorem}{Theorem}

\newtheorem{example}{Example}
\newtheorem{note}{Note}
\usepackage{setspace}
\usepackage{amscd}
\usepackage{mathrsfs}
\usepackage{epsfig}
\usepackage{color}
\usepackage{textcomp}
\usepackage{multirow}
\usepackage{threeparttable}
\title{Linear Network Coding, Linear Index Coding and Representable Discrete Polymatroids}
\begin{document}

\author{Vijayvaradharaj T. Muralidharan and B. Sundar Rajan,~\IEEEmembership{Fellow,~IEEE}
\thanks{V. T. Muralidharan is currently with Qualcomm India Private Limited, Hyderabad-500081, India (e-mail: vmuralid@qti.qualcomm.com). B. S. Rajan is with the Dept. of Electrical Communication Engineering, Indian Institute of Science, Bangalore-560012, India (e-mail: bsrajan@ece.iisc.ernet.in).}
\thanks{This work was carried out when V. T. Muralidharan was with the Dept. of Electrical Communication Engineering, Indian Institute of Science, Bangalore-560012, India.}
}
\maketitle
%%%%%%%%
\begin{abstract}
Discrete polymatroids are the multi-set analogue of matroids. In this paper, we explore the connections among linear network coding, linear index coding and representable discrete polymatroids. We consider vector linear solutions of networks over a field $\mathbb{F}_q,$  with possibly different message and edge vector dimensions, which are referred to as linear fractional solutions. It is well known that a scalar linear solution over $\mathbb{F}_q$ exists for a network  if and only if the network is \textit{matroidal} with respect to a \textit{matroid} representable over $\mathbb{F}_q.$ We define a \textit{discrete polymatroidal} network and show that a linear fractional solution over a field $\mathbb{F}_q,$ exists for a network if and only if the network is discrete polymatroidal with respect to a discrete polymatroid representable over $\mathbb{F}_q.$ An algorithm to construct networks starting from certain class of discrete polymatroids is provided. Every representation over $\mathbb{F}_q$ for the discrete polymatroid, results in a linear fractional solution over $\mathbb{F}_q$ for the constructed network.\\ %Using the construction procedure, we obtain some new examples of networks which admit linear solution, but no scalar linear solution over $\mathbb{F}_q.$\\ 
\indent Next, we consider the index coding problem, which involves a sender which generates a set of messages $X=\{x_1,x_2,\dotso x_k\}$ and a set of receivers $\mathcal{R}$ which demand messages. A receiver $R \in \mathcal{R}$ is specified by the tuple $(x,H)$ where $x \in X$ is the message demanded by $R$ and $H \subseteq X \setminus \{x\}$ is the side information possessed by $R.$ We first show that a linear solution to an index coding problem exists if and only if there exists a representable discrete polymatroid satisfying certain conditions which are determined by the index coding problem considered. El Rouayheb et. al. showed that the problem of finding a multi-linear representation for a matroid can be reduced to finding a \textit{perfect linear index coding solution} for an index coding problem obtained from that matroid. Multi-linear representation of a matroid can be viewed as a special case of representation of an appropriate discrete polymatroid. We generalize the result of El Rouayheb et. al. by showing that the problem of finding a representation for a discrete polymatroid can be reduced to finding a perfect linear index coding solution for an index coding problem obtained from that discrete polymatroid.
\end{abstract}
\section{Background and Related Work}
\label{I}
The concept of network coding, originally introduced by Ahlswede et. al. in \cite{Ah}, helps towards providing more throughput in a communication network than what pure routing solutions provide. For solvable multicast networks, it was shown in \cite{Li} that linear solutions exist for sufficiently large field size. An algebraic framework for finding linear solutions in networks was introduced in \cite{KoMe}.

The connection between matroids and network coding was studied by Dougherty et. al. in \cite{DoFrZe}. In \cite{DoFrZe}, the notion of \textit{matroidal network} was introduced and it was shown that if a scalar linear solution over $\mathbb{F}_q$ exists for a network, then the network is matroidal with respect to a representable matroid. The converse that a scalar linear solution exists for a network if the network is matroidal with respect to a representable matroid was shown in \cite{KiMe}. 

A construction procedure was given in \cite{DoFrZe} to obtain networks from matroids, in which the resulting network admits a scalar linear solution over $\mathbb{F}_q,$ if the matroid is representable over $\mathbb{F}_q.$ Using the networks constructed with the construction procedure given in \cite{DoFrZe}, it was shown in \cite{DoFrZe_In} that there exists networks which do not admit any scalar and vector linear solution, but admit a non-linear solution. In \cite{LiSu}, optimal solutions for cycilc networks were constructed from associated acyclic networks, motivated by results from matroid duality theory. Linear network codes over cyclic networks were characterized using matroids in \cite{SuLiCh}.

Extending the notion of matroidal network to networks which admit error correction, it was shown in \cite{PrRa_ISIT} that a network admits a scalar linear
error correcting network code if and only if it is a matroidal error correcting network associated with a representable matroid. Constructions of networks from matroids with error correction capability were provided in \cite{PrRa_ISIT,PrRa_ISITA}.

It was shown in \cite{MeEfHoKa} that it is possible for a non-scalar linear solvable network to admit a vector linear solution, in which the edges carry vectors over $\mathbb{F}_q$ whose dimensions are same as that of the message vectors. Throughout this paper, by a vector network coding solution, we refer to a solution for which all the dimensions of the message vectors are equal to the edge vector dimension. It is possible that a network does not admit any scalar or vector solution, but admits a solution if all the dimensions of the message vectors are not equal to the edge vector dimension. Such network coding solutions called Fractional Network Coding (FNC) solutions, for which all the dimensions of the message vectors are not necessarily equal to the edge vector dimension, have been considered in \cite{CaDoFrZe,KiMe_FNC,DoFrZe_FNC}. The work in \cite{CaDoFrZe} primarily focusses on fractional routing, which is a special case of FNC. 
In \cite{KiMe_FNC}, algorithms were provided to compute the capacity region for a network, which was defined to be the closure of all rates achievable using FNC. In \cite{DoFrZe_FNC}, achievable rate regions for certain specific networks were found and it was shown that achievable rate regions using linear FNC need not be convex.

An index coding problem $\mathcal{I}(X,\mathcal{R}),$ which is a special case of the general network coding problem, involves a sender which generates a set of messages $X=\{x_1,x_2,\dotso x_k\}$ and a set of receivers $\mathcal{R}$ which demand messages \cite{BaBiJaKo,ChSp,RoSpGe}. A receiver $R \in \mathcal{R}$ is specified by the tuple $(x,H)$ where $x \in X$ is the message demanded by $R$ and $H \subseteq X \setminus \{x\}$ is the side information possessed by $R.$
In \cite{RoSpGe}, El Rouayheb, Sprinston and Georghiades analyzed the connection among network coding, index coding and multi-linear representations of matroids. In \cite{RoSpGe}, it was shown that the problem of finding a linear solution for a network coding problem can be reduced to the problem of finding a \textit{perfect linear index code} (for a formal definition see Section \ref{III B}) for an index coding problem, which was obtained from the network considered. Also, it was shown in \cite{RoSpGe} that the problem of finding a multi-linear representation for a matroid can be reduced to finding a perfect linear index code for an index coding problem obtained from that matroid.

Discrete polymatroids are the multi-set analogue of matroids\cite{HeHi,Vl,Fu}{\footnote{The term discrete polymatroid was first introduced by Herzog and Hibi in \cite{HeHi}, while the concept was earlier treated in the first edition of \cite{Fu} with the underlying additive group being the set of integers.}}. Linear and multi-linear representations of matroids can be viewed equivalently as representations of appropriate discrete polymatroids. Representable discrete polymatroids have been used in the context of secret sharing in cryptography \cite{FaFaPa,HsCh,FaPa,FaMaPa}.
In this paper, for the first time to the best of our knowledge, we explore the connections among linear network coding, linear index coding and representable discrete polymatroids. 

The organization of this paper is as follows: An overview of matroids and discrete polymatroids is presented in Section II. Section III deals with the preliminaries related to network coding and index coding. Section IV deals with the connection between linear FNC and representable discrete polymatroids. The connection between linear index coding and representable discrete polymatroids is explored in Section V. In Section VI, we discuss about other possible connections between network/index coding and discrete polymatroids, obtained using the results in this paper and the one in \cite{RoSpGe}.   

The main contributions of this paper are as follows:
\begin{itemize}
\item
Discrete polymatroids can be viewed as the  generalization of matroids. It is known that the vectors which belong to a discrete polymatroid, can be viewed as the generalization of the notion of independent sets of a matroid and the basis vectors of a discrete polymatroid can be viewed as the generalization of the notion of basis sets of a matroid (Section \ref{II C}). To the best of our knowledge, the notion of circuits of matroids has not been generalized to discrete polymatroids. In Section \ref{II D}, we introduce the notion of \textit{minimal excluded vector} for a discrete polymatroid, which can be viewed as the generalization of the notion of circuits of a matroid. %We show that a discrete polymatroid can be defined alternatively in terms of certain axioms involving the  set of minimal excluded vectors.
In the later sections, this notion of minimal excluded vector is extensively used to construct networks from discrete polymatroids, which admit linear FNC solutions, as well as to construct index coding problems which admit perfect linear index coding solutions. 
\item
In \cite{KiMe}, Kim and Medard made the following comment:
``\textit{\dotso Unfortunately, the results presented in this paper do not seem to generalize to vector-linear network coding or more general network coding schemes. The difficulty is that the matroid structure requires that a subset of the ground set of a matroid is either independent or dependent, but what this corresponds to in vector-linear codes, for instance, is not clear.\dotso}'' 
In this paper, we establish that there is a fundamental connection between discrete polymatroids and linear FNC. Towards, establishing that connection,  the notion of \textit{discrete polymatroidal network} is introduced, which can be viewed as a generalization of the notion of matroidal network introduced in \cite{DoFrZe}. In Section \ref{IV A}, it is shown that a linear FNC solution exists for a network over a field $\mathbb{F}_q$ if and only if the network is discrete polymatroidal with respect to a discrete polymatroid representable over $\mathbb{F}_q.$
\item
A construction algorithm to obtain networks from a class of discrete polymatroids is provided in Section \ref{IV B}. Starting from a discrete polymatroid which is representable over $\mathbb{F}_q,$ the resulting networks admit a linear FNC solution over $\mathbb{F}_q.$ 
\item
In Section \ref{VI A}, it is shown that a linear solution to an index coding problem exists if and only if there exists a representable discrete polymatroid satisfying certain conditions which are determined by the index coding problem considered. 
In Section \ref{VI B}, we provide a construction of an index coding problem, starting from a discrete polymatroid. It is shown that a perfect linear index coding solution exists for this index coding problem, if and only if the discrete polymatroid from which the index coding problem was constructed is representable. In this way, the problem of finding a representation for a discrete ploymatroid reduces to the problem of finding a perfect linear solution for an index coding problem constructed from the discrete polymatroid.
\end{itemize}

%Fig. \ref{lnc_lic_rdp} shows a pictorial depiction of the main results in this paper. 
%\begin{figure*}[htbp]
%\centering
%\includegraphics[totalheight=1.2 in,width=7in]{LNC_LIC_RDP.eps}
%\caption{A pictorial representation of the main results obtained in this paper}
%\label{lnc_lic_rdp}
%\end{figure*}

The main differences between the work in this paper and the related work in \cite{RoSpGe} are as follows:
\begin{itemize}
\item
The work in \cite{RoSpGe} considers multi-linear representations of matroids. In this paper, we consider discrete polymatroids, which are more general than matroids. With every matroid we can associate a unique discrete polymatroid, but not vice versa. All multi-linear representations of matroids can be viewed equivalently as representations of appropriate discrete polymatroids, but the converse is not true. There exists discrete polymatroids whose representation cannot be viewed equivalently as the multi-linear representation of any matroid. For more details on this, see Section \ref{II C}. We show that not all linear FNC solutions can be characterized using multi-linear representations of matroids, whereas they can be characterized using representations of discrete polymatroids.
\item
The relationship among multi-linear representation of matroids, linear index coding and linear network coding established in \cite{RoSpGe} is as follows: Starting from a matroid, an index coding problem was constructed and it was shown that a perfect-linear index coding solution exists for the index coding problem if and only if the associated matroid has a multi-linear representation. Also, a network coding problem was obtained from the constructed index coding problem, which has a vector linear solution if and only if the associated matroid has a multi-linear representation. This relationship between matroid multi-linear representation and network (index) coding established in \cite{RoSpGe} is restricted to the network (index) coding problem constructed from a matroid and not for an arbitrary network (index) coding problem. 
The connections established in this paper between discrete polymatroids and linear FNC in Section \ref{IV A}, and  between discrete polymatroids and linear index coding in Section \ref{VI A}, are valid for arbitrary networks and index coding problems respectively. 
%In Section III of this paper, it is shown that for any arbitrary index coding problem, a linear solution exists if and only if there exists a representable discrete polymatroid satisfying certain properties which are determined by the index coding problem. 
\item
The construction of networks and index coding problems presented in \cite{RoSpGe} are from matroids, where as the constructions provided in this paper are from discrete polymatroids, which are more general than matroids. The construction of index coding problem from discrete polymatroids provided in Section \ref{VI B} in this paper is a generalization of the construction from matroids in \cite{RoSpGe}.%, in the sense that all the results provided Section IV in this paper reduce to the results provided in \cite{RoSpGe}, when specialized to a discrete polymatroid associated with a matroid.
\end{itemize}
\textbf{\textit{Notations:}}
The set $\lbrace 1,2,\dotso,r \rbrace$ is denoted as  $\lceil r \rfloor.$ $\mathbb{Z}_{\geq 0} $ and $\mathbb{R}_{\geq 0}$ denote the set of non-negative integers and real numbers respectively. For a vector $v$ of length $r$ and $A \subseteq \lceil r \rfloor,$ $v(A)$ is the vector obtained by taking only the components of $v$ indexed by the elements of $A.$ The vector of length $r$ whose $i^{\text{th}}$ component is one and all other components are zeros is denoted as $\epsilon_{i,r}.$  For $u,v \in \mathbb{Z}_{\geq 0}^r,$ $u \leq v$ if all the components of $v-u$ are non-negative and,  $u < v$ if $u \leq v$ and $u \neq v.$ For $u,v \in \mathbb{Z}_{\geq 0}^r,$ $u \vee v$ is the vector whose $i^{\text{th}}$ component is the maximum of the $i^{\text{th}}$ components of $u$ and $v.$ A vector $u \in \mathbb{Z}_{\geq 0}^r$ is called an integral sub-vector of $v \in \mathbb{Z}_{\geq 0}^r$ if $u \leq v.$ For a set $A,$ $\vert A \vert$ denotes its cardinality and for a vector $v \in  \mathbb{Z}_{\geq 0}^r,$ $\vert v \vert$ denotes the sum of the components of $v.$ For a vector $u \in \mathbb{Z}_{\geq 0}^{r},$ $(u)_{>0}$ denotes the set of indices corresponding to the non-zero components of $u.$

\section{Matroids and Discrete Polymatroids}
In Section \ref{II A} and Section \ref{II B}, the basic definitions and notations related to matroids and discrete polymatroids are provided. In Section \ref{II C}, how a matroid can be viewed as a special case of a discrete polymatroid is explained. In Section \ref{II D}, the notion of minimal excluded vectors for a discrete polymatroid is introduced, which when specialized reduces to the well known notion of circuits for matroids. 
\subsection{Matroids}
\label{II A}
In this subsection, a brief overview of matroids is presented. For a comprehensive treatment, the readers are referred to \cite{We,Ox}.
\begin{definition}[\cite{We}]
Consider a function $\Upsilon:2^{\lceil r \rfloor} \rightarrow \mathbb{Z}_{\geq 0}$ on ground set $\lceil r \rfloor$ which satisfies the following conditions $\forall A \subseteq \lceil r \rfloor$:
\begin{itemize}
\item[(R1)] 
 $\Upsilon(A)\leq \vert A \vert.$
\item[(R2)]
$\Upsilon(A)\leq \Upsilon(B),A \subseteq B.$
\item[(R3)]
$\Upsilon(A\cup B)+\Upsilon(A \cap B)\leq \Upsilon(A)+\Upsilon(B).$ 
\end{itemize}
A matroid with rank function $\Upsilon$ is the pair $(\lceil r \rfloor,\mathcal{I}),$ where the set $\mathcal
{I}$ called the set of independent sets is defined as $\mathcal{I}=\{X\subseteq \lceil r \rfloor:\Upsilon(X)=\vert X \vert\}.$
\end{definition}

The sets which do not belong to $\mathcal{I}$ are called the dependent sets. A maximal independent set is a basis set and a minimal dependent set is called a circuit. 
The rank of the matroid $\mathbb{M},$ denoted by $rank(\mathbb{M})$ is equal to $\Upsilon(\lceil r \rfloor).$ 
A matroid can be equivalently defined in terms of the set of independent sets, basis sets and the set of circuits.  

%\begin{example}
%\label{ex1}
% Consider the matroid on the ground set $\lceil 4 \rfloor$ with the set of independent sets given by  {\small$\{\emptyset, \{1\},\{2\},\{3\},\{4\},\{1,2\},\{1,3\},\{1,4\},\{2,3\},\{2,4\},\{3,4\}\}.$} This matroid is referred to as the uniform matroid $U_{2,4}.$ The rank function for this matroid is given by, $\Upsilon(X)=\min\{\vert X \vert,2\}, X \subseteq \lceil 4 \rfloor$ and the rank of this matroid is equal to 2. 
%\end{example}

A matroid $\mathbb{M}$ is said to be representable over $\mathbb{F}_q$ if there exist one-dimensional vector subspaces $V_1,V_2, \dotso V_r$ of a vector space $E$ such that $\dim(\sum_{i \in X} V_i)=\Upsilon(X), \forall X \subseteq \lceil r \rfloor$ and the set of vector subspaces $V_i, i \in \lceil r  \rfloor,$ is said to form a representation of $\mathbb{M}.$ The one-dimensional vector subspaces $V_i, i \in \lceil r \rfloor,$ can be described by a matrix $A$ over $\mathbb{F}_q$ whose $i^{\text{th}}$ column spans $V_i.$ 

%\begin{example}
%\label{ex2}
%Continuing with Example \ref{ex1}, let $A=\begin{bmatrix} 1 & 0 & 1 & 1\\0 &1 &1 &2 \end{bmatrix}$ be a matrix over $\mathbb{F}_3.$ Let $V_i, i \in \lceil 4 \rfloor,$ denote the span of the $i^{\text{th}}$ column of $A$ over $\mathbb{F}_3.$ It can be verified that the vector subspaces $V_1,V_2,V_3$ and $V_4$ form a representation of $U_{2,4}$ over $\mathbb{F}_3.$
%%\endproof
%\end{example}

The notion of multi-linear representation of matroids was introduced in \cite{SiAs,Ma}.
A matroid $\mathbb{M}$ on the ground set $\lceil r \rfloor$ is said to be multi-linearly representable of dimension $n$ over $\mathbb{F}_q$ if there exist vector subspaces $V_1,V_2,\dotso,V_r$ of a vector space $E$ over $\mathbb{F}_q$ such that $dim(\sum_{i \in X} V_i)=n\Upsilon(X), \forall X \subseteq \lceil r \rfloor.$ The vector subspaces $V_i,i\in\lceil r \rfloor,$ are said to form a multi-linear representation of dimension $n$ over $\mathbb{F}_q$ for the matroid $\mathbb{M}.$ For $n=1,$ the notion of multi-linear representation reduces to the notion of representation of matroids. 

\subsection{Discrete Polymatroids}
\label{II B}
In this subsection, an overview of discrete polymatroids is presented. For more details and examples on discrete polymatroids, interested readers are referred to \cite{HeHi,Vl,Fu}.

A discrete polymatroid $\mathbb{D}$ is defined as follows:
\begin{definition}[\cite{FaFaPa}]
Consider a function $\rho:2^{\lceil r \rfloor} \rightarrow \mathbb{Z}_{\geq 0}$ on the ground set $\lceil r \rfloor$ with $\rho(\phi)=0,$ which satisfies (R2) and (R3) in Definition 1, but not necessarily (R1). 
An integer polymatroid $\mathbb{I}$ with rank function $\rho$ is the region defined as $\{x \in \mathbb{R}_{\geq 0}^r:\vert x(A)\vert\leq \rho(A), \forall A \subseteq \lceil r \rfloor\}$\cite{Ed}.
A discrete polymatroid $\mathbb{D}$ with rank function $\rho$ is the set of vectors in $\mathbb{I}$ whose components take only integral values. In other words, a discrete polymatroid $\mathbb{D}$ with rank function $\rho$ is defined as $\mathbb{D}=\{x \in \mathbb{Z}_{\geq 0}^r:\vert x(A)\vert\leq \rho(A), \forall A \subseteq \lceil r \rfloor\}.$
\end{definition}

\begin{note}
A function $\rho:2^{\lceil r \rfloor} \rightarrow \mathbb{Z}_{\geq 0}$ for which $\rho(X)=0, \forall X \subseteq \lceil r \rfloor$ is  the rank function of a trivial discrete polymatroid which contains only the all-zero vector. In this paper, we only consider non-trivial discrete polyamtroids. 
\end{note}
%For a discrete polymatroid $\mathbb{D},$ the rank function $\rho^{\mathbb{D}}: 2^{\lceil r \rfloor} \rightarrow \mathbb{Z}_{\geq 0}$ is defined as $\rho^{\mathbb{D}}(A)=\max \{ \vert u(A) \vert , u \in \mathbb{D}\},$ where $\emptyset \neq A  \subseteq \lceil r \rfloor$ and $\rho^{\mathbb{D}}(\emptyset)=0.$ Alternatively, a discrete polymatroid $\mathbb{D}$ can be written in terms of its rank function as $\mathbb{D}=\lbrace x \in \mathbb{Z}_{\geq 0}^r: \vert x(A) \vert \leq \rho^{\mathbb{D}}(A), \forall A \subseteq \lceil r \rfloor \rbrace.$ In the rest of the paper, for simplicity, the superscript $\mathbb{D}$ in $\rho^{\mathbb{D}}$ is dropped.

%A function $\rho: 2^{\lceil r \rfloor} \rightarrow \mathbb{Z}_{\geq 0}$ is the rank function of a discrete polymatroid if and only if it satisfies the following conditions \cite{FaFaPa}:
%\begin{description}
%\item [(D1)]
%For $A \subseteq B \subseteq \lceil r \rfloor,$ $\rho(A)\leq \rho(B).$
%\item [(D2)]
% $\forall A,B \subseteq \lceil r \rfloor,$  $\rho(A \cup B) + \rho(A\cap B)\leq\rho(A)+\rho(B).$
%\item [(D3)]
%$\rho(\emptyset)=0.$
%\end{description}
%% 
 
A vector $u \in  \mathbb{D}$ for which there does not exist $v \in \mathbb{D}$ such that $u<v,$ is called a basis vector of $\mathbb{D}.$  Let $\mathcal{B}(\mathbb{D})$ denote the set of basis vectors of $\mathbb{D}.$  The sum of the components of a basis vector of $\mathbb{D}$ is referred to as the rank of $\mathbb{D},$ denoted by $rank(\mathbb{D}).$ Note that for all the basis vectors, sum of the components will be equal \cite{Vl}. A discrete polymatroid is nothing but the set of all integral subvectors of its basis vectors.

%\begin{example}
%\label{ex5}
%\begin{figure}[htbp]
%\centering
%\includegraphics[totalheight=2.5 in,width=3.5in]{example_2d_ind.eps}
%\caption{An example of a discrete polymatroid}
%\label{example_2d_ind}
%\end{figure}
%Let $\rho: 2^{\lceil 2 \rfloor} \rightarrow \mathbb{Z}_{\geq 0}$ be defined as follows:
%$\rho (\{1\})=3$ and $\rho(\{2\})=\rho(\{1,2\})=5.$ It can be seen that $\rho$ satisfies (D1)--(D3) and hence  $\rho$ is the rank function of a discrete polymatroid. The vectors which belong to this discrete polymatroid are the points marked by `$\circ$' in Fig. \ref{example_2d_ind}. The set of basis vectors for this discrete polymatroid is given by $\{(0,5),(1,4),(2,3),(3,2)\}.$
%\end{example}
%\begin{example}[\cite{Vl}]
%\label{ex6}
%Let \mbox{$\rho: 2^{\lceil 3 \rfloor} \rightarrow \mathbb{Z}_{\geq 0}$} be defined as follows: $\rho(\emptyset)=0,$ $\rho(\{1\})=1,$ $\rho(\{2\})=\rho(\{3\})=\rho(\{1,3\})=2,$ $\rho(\{1,2 \})=3,$ $\rho(\{2,3\})=\rho(\{1,2,3\})=4.$ It can be verified that $\rho$ satisfies (D1)--(D3) and hence is the rank function of the discrete polymatroid given by,

%{\footnotesize
%\begin{align*}
%\{&(0, 0, 0), (1, 0, 0), (0, 1, 0), (0, 0, 1), (1, 1, 0), (0, 1, 1), (1, 0, 1), (0, 2, 0),\\
% &\hspace{7 cm}(0, 0, 2),(0, 1, 2), (0, 2, 1), (1, 1, 1), (1, 2, 0), (0, 2, 2), (1, 2, 1)\}.
%\end{align*}
%} 
%\noindent The set of basis vectors of this discrete polymatroid is $\mathcal{B}(\mathbb{D})=\{(0,2,2),(1,2,1)\}.$
%\endproof
%\end{example}

\begin{example}
\label{ex7}
Consider the discrete polymatroid $\mathbb{D}$ on the ground set $\lceil 3 \rfloor$ with the rank function $\rho$ given by $\rho(\{1\})=\rho(\{2\})=\rho(\{2,3\})=2,\rho(\{3\})=1$ and $\rho(\{1,2\})=\rho(\{1,3\})=\rho(\{1,2,3\})=3.$ Note that the function $\rho$ satisfies the conditions (R2) and (R3). The set of basis vectors for this discrete polymatroid is given by $\mathcal{B}(\mathbb{D})=\{(1,1,1),(1,2,0),(2,0,1),(2,1,0)\}.$
\end{example}

 %A vector $u \in \mathbb{D}$ is a basis vector of $\mathbb{D},$ if $u <v$ for no $v \neq u \in \mathbb{D}.$ The set of basis vectors of $\mathbb{D}$ is denoted as $\mathcal{B}(\mathbb{D}).$ For all $u \in \mathcal{B}(\mathbb{D}),$ $\vert u \vert$ is equal \cite{Vl}, which is called the rank of $\mathbb{D},$ denoted by $rank(\mathbb{D}).$  
 
 Let $V_1,V_2,\dotso, V_r$ be vector subspaces of a finite dimensional vector space $E.$ Define the mapping $\rho: 2^{\lceil r \rfloor} \rightarrow \mathbb{Z}_{\geq 0}$ as $\rho(X)=dim(\sum_{i \in X} V_i),$ $X \subseteq \lceil r \rfloor.$ The mapping $\rho$ satisfies (R2) and (R3), and is the rank function of a discrete polymatroid, which we denote by $\mathbb{D}(V_1,V_2,\dotso,V_r).$ 
Note that $\rho$ remains the same even if we replace the vector space $E$ by the sum of the vector subspaces $V_1,V_2,\dotso,V_r.$ In the rest of the paper, the vector subspace $E$ is taken to be the sum of the  vector subspaces $V_1,V_2,\dotso,V_r$ considered. 
The vector subspaces $V_1,V_2,\dotso,V_r$ can be described by a matrix $A=[A_1 \; A_2 \; \dotso A_r ],$ where $A_i, i \in \lceil r \rfloor,$ is a matrix whose columns span $V_i.$
\begin{definition}[\cite{FaFaPa}]
A discrete polymatroid $\mathbb{D}$ is said to be representable over $\mathbb{F}_q$ if there exist vector subspaces $V_1,V_2,\dotso,V_r$ of a vector space $E$ over $\mathbb{F}_q$ such that $dim(\sum_{i \in X} V_i)=\rho(X),$ $\forall X \subseteq \lceil r \rfloor.$ The set of vector subspaces $V_i,i\in\lceil r \rfloor,$ is said to form a representation of $\mathbb{D}.$ 
A discrete polymatroid is said to be representable if it is representable over some field.
\end{definition}
\begin{example}
\label{ex8}
Let $A_1=\begin{bmatrix}1 & 0\\0 & 1\\0 & 0\end{bmatrix},$ $A_2=\begin{bmatrix} 0 & 1\\0 & 1 \\1 &1\end{bmatrix}$ and $A_3=\begin{bmatrix}0 \\ 0 \\ 1\end{bmatrix}$ be matrices over $\mathbb{F}_2.$ Let $V_i, i\in \lceil 3 \rfloor,$ denote the column span of $A_i.$ It can be verified that the vector subspaces $V_1,V_2$ and $V_3$ form a representation over $\mathbb{F}_2$ of the discrete polymatroid given in Example \ref{ex7}.
\end{example}
\begin{example}
\label{ex9}
Let $A= \underbrace{\hspace{-.3 cm}\left[\begin{matrix}\hspace{.2 cm} 1 \\ \hspace{.2 cm}0 \\\hspace{.2 cm}0 \end{matrix}\right.}_{A_1} \; \underbrace{\begin{matrix} 0\\ 1\\0 \end{matrix}}_{A_2} \; \underbrace{\begin{matrix} 0\\ 0\\1 \end{matrix}}_{A_3}  \; \underbrace{\left.\begin{matrix} 1&0\\ 0&1\\0&1 \end{matrix}\right]}_{A_4}$ be a matrix over $\mathbb{F}_q.$ Let $V_i$ denote the column span of $A_i,$ \mbox{$i \in \lceil 4 \rfloor.$} The rank function $\rho$ of the discrete polymatroid $\mathbb{D}(V_1,V_2,V_3,V_4)$ is as follows: \mbox{$\rho(X)=1, \text{\;if\;} X \in \left\{\{1\},\{2\},\{3\}\right\};$} \mbox{$\rho(X)=2, \text{\;if\;} X \in \left\{\{1,2\},\{1,3\},\{1,4\},\{2,3\},\{4\}\right\}$} and \mbox{$\rho(X)=3 \text{\; otherwise}.$} The set of basis vectors for this discrete polymatroid is given by, 
{ $$\left\{(0,0,1,2),(0,1,0,2),(0,1,1,1),(1,0,1,1),(1,1,0,1),\right.$$ $$\left.\hspace{6.75 cm}(1,1,1,0)\right\}.$$}
\end{example}

Next, an example of a discrete polymatroid which is not representable is provided.
\begin{example}
Let $\rho:2^{\lceil 4 \rfloor}\rightarrow \mathbb{Z}_{\geq 0}$ be a function given by $\rho(\{1\})=\rho(\{2\})=\rho(\{3\})=\rho(\{4\})=2,$ $\rho(\{1,2\})=\rho(\{1,3\})=\rho(\{1,4\})=\rho(\{2,3\})=\rho(\{2,4\})=3$ and $\rho(\{3,4\})=\rho(\{1,2,3\})= \rho(\{1,2,4\})=\rho(\{1,3,4\})=\rho(\{2,3,4\})=\rho(\{1,2,3,4\})=4.$ It can be verified that $\rho$ satisfies the conditions (R2) and (R3), and hence it is the rank function of a discrete polymatroid. Note that $\rho$ does not satisfy the Ingleton inequality \cite{Ing}, which is a necessary condition for a discrete polymatroid to be representable. Hence, this discrete polymatroid is not representable. The set of basis vectors for this discrete polymatroid is given by,
$$\left\{(0,0,2,2),(2,1,1,0),(2,1,0,1),(2,0,1,1),(0,2,1,1),\right.$$ $$\left.(1,2,0,1),(1,2,1,0),(1,1,2,0),(1,0,2,1),(1,0,2,1),\right.$$ $$\left.\hspace{2 cm}(1,1,0,2),(1,0,1,2),(0,1,1,2),(1,1,1,1)\right\}.$$
\end{example}
%\begin{example}
%\label{ex10}
%Let $A= \underbrace{\hspace{-.0 cm}\left[\begin{matrix} 1 & 0\\ 0 & 1\\0&0\\0&0 \end{matrix}\right.}_{A_1} \; \underbrace{\begin{matrix} 0\\ 0\\1\\0 \end{matrix}}_{A_2} \; \underbrace{\begin{matrix} 0\\ 0\\0\\1 \end{matrix}}_{A_3}  \; \underbrace{\left.\begin{matrix} 1&1\\ 1&0\\1&1\\1&0 \end{matrix}\right.}_{A_4}\; \underbrace{\left.\begin{matrix} 0&0\\ 0&1\\0&1\\1&0 \end{matrix}\right]}_{A_5}$ be a matrix over $\mathbb{F}_q.$ Let $V_i$ denote the column span of $A_i,$ \mbox{$i \in \lceil 5 \rfloor.$} Then the rank function $\rho$ of the discrete polymatroid $\mathbb{D}(V_1,V_2,V_3,V_4,V_5)$ is as follows: \mbox{$\rho(X)=1, \text{\;if\;} X \in \left\{\{2\},\{3\}\right\};$} \\\mbox{$\rho(X)=2, \text{\;if\;} X \in \left\{\{1\},\{4\},\{5\},\{2,3\},\{3,5\}\right\};$}\\ $\rho(X)=3, \text{\;if\;} X \in \left\{\{1,2\},\{1,3\},\{2,4\},\{2,5\},\{3,4\},\{2,3,5\}\right\}$ and  \mbox{$\rho(X)=4, \text{\;otherwise.}$} The set of basis vectors for this discrete polymatroid is given by,
%{
%\begin{align*}
%&\hspace{-.2 cm}\left\{(0,0,0,2,2),(0,0,1,2,1),(0,1,0,2,1),(0,1,1,1,1),(0,1,1,2,0),\right.\\&(1,0,0,2,1),(1,0,1,1,1),(1,0,1,2,0),(1,1,0,0,2),(1,1,0,1,1),\\&(1,1,0,2,0),(1,1,1,0,1),(1,1,1,1,0),(2,0,0,0,2),(2,0,0,1,1),\\&(2,0,0,2,0),(2,0,1,0,1),(2,0,1,1,0),(2,1,0,0,1),(2,1,0,1,0),\left.(2,1,1,0,0)\right\}.
%\end{align*}
%}
%\end{example}

%Discrete polymatroids are closely related to \textit{integral polymatroids}, which are a special class of polymatroids \cite{HeHi,Vl}.
\subsection{Matroids viewed as a special case of Discrete Polymatroids}
\label{II C} 
Discrete polymatroids can be viewed as a generalization of matroids \cite{HeHi,Vl}. It is well known that there is a one-to-one correspondence between the independent sets of a matroid and the vectors which form an associated discrete polymatroid. Similarly, it is known that there is a one-to-one correspondence between the basis sets of a matroid and the basis vectors of an associated discrete polymatroid. In this subsection, a brief discussion about this connection between matroids and discrete polymatroids is presented.
 
 %A function $\Upsilon: 2^{\lceil r \rfloor} \rightarrow \mathbb{Z}_{\geq 0}$ is the rank function of a matroid if and only if it satisfies the conditions (D1)--(D3) and the additional condition that $\Upsilon(X)\leq \vert X \vert, \forall X \subseteq \lceil r \rfloor$ (follows from Theorem 3 in Chapter 1.2 in \cite{We}). 
 
 Since the rank function $\Upsilon$ of a matroid $\mathbb{M}$ satisfies (R2)and (R3), it is also the rank function of a discrete polymatroid denoted as $\mathbb{D}(\mathbb{M}).$ Note that the rank function $\Upsilon$ of $\mathbb{D}(\mathbb{M})$ satisfies $\Upsilon(X)\leq \vert X \vert, \forall X \subseteq \lceil r \rfloor,$ in addition to (R2) and (R3). 
There is a one-to-one correspondence between the matroid $\mathbb{M}$ and the discrete polymatroid $\mathbb{D}(\mathbb{M}).$ For every independent set $I$ of the matroid $\mathbb{M},$ there exists a unique vector belonging to $\mathbb{D}(\mathbb{M})$ whose components indexed by the elements of $I$ take the value one and all other components are zeros.
In other words, in terms of the set of independent sets $\mathcal{I}$ of $\mathbb{M},$  the discrete polymatroid $\mathbb{D}(\mathbb{M})$ can be written as $\mathbb{D}(\mathbb{M})=\{\sum_{i \in I} \epsilon_{i,r} :I \in \mathcal{I}\}.$ Conversely, the set of independent sets $\mathcal{I}$ of $\mathbb{M}$ is given by $\mathcal{I}=\{(u)_{>0} : u \in \mathbb{D}(\mathbb{M})\}.$ 

Similarly, for a basis set $B$ of a matroid $\mathbb{M},$ the vector $\sum_{i \in B}\epsilon_{i,r}$ is a basis vector of $\mathbb{D}(\mathbb{M})$ and conversely, for a basis vector $b$ of $\mathbb{D}(\mathbb{M}),$ the set $(b)_{>0}$ is a basis set of $\mathbb{M}.$
\begin{example}
\label{ex11}
Consider the matroid on the ground set $\lceil 4 \rfloor$ with the set of independent sets given by  {\small$\{\emptyset, \{1\},\{2\},\{3\},\{4\},\{1,2\},\{1,3\},\{1,4\},\{2,3\},\{2,4\},\{3,4\}\}.$} This matroid is referred to as the uniform matroid $U_{2,4}.$ The rank function for this matroid is given by, $\Upsilon(X)=\min\{\vert X \vert,2\}, X \subseteq \lceil 4 \rfloor.$
For the matroid $U_{2,4}$ , the discrete polymatroid $\mathbb{D}(U_{2,4})$ is given by $$\mathbb{D}(U_{2,4})=\{(0,0,0,0),(1,0,0,0),(0,1,0,0),(0,0,1,0),$$ $$\hspace{1.8cm}(0,0,0,1),(1,1,0,0),(1,0,1,0),(1,0,0,1),$$ $$\hspace{3.2cm}(0,1,1,0),(0,1,0,1),(0,0,1,1)\}.$$ For every independent set $I$ of $\mathbb{M},$ $\mathbb{D}(\mathbb{M})$ contains a vector whose components indexed by the elements of $I$ are ones and all other components are zeros. 
\end{example}

A set of vector subspaces $V_i, i \in \lceil r  \rfloor,$ forms a representation of $\mathbb{M}$ if and only if it forms a representation of $\mathbb{D}(\mathbb{M}).$ In this way, the representability of a matroid $\mathbb{M}$ over $\mathbb{F}_q$ can be viewed equivalently as the representability of the discrete polymatroid $\mathbb{D}(\mathbb{M})$ over $\mathbb{F}_q.$

%\begin{example}
%\label{ex12}
%The vector subspaces $V_i,i\in \lceil 4 \rfloor,$  provided  in Example \ref{ex2}, which form a representation over $\mathbb{F}_3$ for $U_{2,4}$, form a representation over $\mathbb{F}_2$ for the discrete polymatroid $\mathbb{D}(U_{2,4})$ provided in Example \ref{ex11}. 
%\end{example}

For a discrete polymatroid $\mathbb{D}$ with rank function $\rho,$ let $n\mathbb{D}$ denote the discrete polymatroid whose rank function $\rho'(X)= n\rho(X), \forall X \subseteq \lceil r \rfloor.$ Note that the function $\rho'$ satisfies the conditions (R2) and (R3). 

\begin{example}
\label{ex13}
For the uniform matroid $U_{2,4},$ the discrete polymatroid $2\mathbb{D}(U_{2,4})$ has the rank function $ \rho'$ given by $\rho'(X)=min\{2\vert X \vert,4\}, X \subseteq \lceil 4 \rfloor.$ The set of basis vectors for this discrete polymatroid is given by,
{
\begin{align*}
&\left \lbrace (0,0,2,2),(0,1,1,2),(0,1,2,1),(0,2,0,2),(0,2,1,1),\right.\\
&\hspace{0.2 cm}(0,2,2,0),(1,0,1,2),(1,0,2,1),(1,1,0,2),(1,1,1,1),\\
&\hspace{0.2 cm}(1,1,2,0),(1,2,0,1),(1,2,1,0),(2,0,0,2),(2,0,1,1),\\
&\left.\hspace{1.6 cm}(2,0,2,0),(2,1,0,1),(2,1,1,0),(2,2,0,0)\right\}.
\end{align*}
}
\end{example}

It is straightforward to see that a matroid has a multi-linear representation of dimension $n$ over $\mathbb{F}_q$ if and only if the discrete polymatroid $n\mathbb{D}(\mathbb{M})$ is representable over $\mathbb{F}_q.$  In this way, the notion of multi-linear representation of dimension $n$ of a matroid $\mathbb{M}$ can be viewed equivalently in terms of the notion of representation of the discrete polymatroid $ n\mathbb{D}(\mathbb{M}).$ 

%\begin{example}
%\label{ex14}
%The vector subspaces $V_i, \in \lceil 4 \rfloor,$ provided in Example \ref{ex3}, which form a multi-linear representation of dimension 2 over $\mathbb{F}_2$ for the matroid $U_{2,4},$ form a representation for the discrete polymatroid $2\mathbb{D}(U_{2,4})$ over $\mathbb{F}_2.$ 
%\end{example}

While the multi-linear representation of any matroid can be viewed equivalently in terms of the representation of an appropriate discrete polymatroid, the converse is not true. For example, consider the representable discrete polymatroid $\mathbb{D}$ given in Example \ref{ex7}. The vector subspaces $V_1, V_2$ and $V_3$ in Example \ref{ex8} which form a representation for $\mathbb{D}$ cannot form a multi-linear representation for any matroid. The reason for this is as follows: For vector subspaces $V_1, V_2$ and $V_3$ to form a multi-linear representation of a matroid, $dim(\sum_{i \in X}V_i)$ should be  a multiple of $n,$ for some integer $n,$ for all $X \subseteq \lceil 3 \rfloor.$ Since $dim(V_3)=1,$ the only possibility for $n$ is 1. In that case, the matroid for which $V_1,V_2$ and $V_3$ form a multi-linear representation of dimension 1 should have a rank function $\Upsilon$ which satisfies $\Upsilon(\{1\})=2,$ which is not possible since $\Upsilon(\{1\})\leq 1.$
\subsection{Excluded and Minimal Excluded Vectors for a Discrete Polymatroid}
\label{II D}
As explained in the previous subsection, the vectors which belong to a discrete polymatroid can be viewed as the generalization of independent sets of matroid and the basis vectors of a discrete polymatroid can be viewed as the generalization of basis sets of a matroid. To the best of our knowledge, the notions of dependent sets and circuits of a matroid have not been generalized to discrete polymatroids. In this subsection, we introduce the notions of excluded and minimal excluded vectors for discrete polymatroids, which when specialized to a matroid reduce to the well known notions of dependent sets and circuits respectively. These notions are useful towards constructing networks and index coding problems from discrete polymatroids in Section \ref{IV B} and Section \ref{VI B}.   

We define an excluded vector for a discrete polymatroid $\mathbb{D}$ as follows:
\begin{definition}
For a discrete polymatroid $\mathbb{D}$ on the ground set $\lceil r \rfloor,$ a vector $u \in \mathbb{Z}^r_{\geq 0}$ is said to be an excluded vector if the $i^{\text{th}}$ component of $u$ is less than or equal to $\rho(\lbrace i \rbrace), \forall i \in \lceil r \rfloor,$ and $u \notin \mathbb{D}.$
\end{definition}

Let $\mathcal{D}(\mathbb{D})$ denote the set of excluded vectors for the discrete polymatroid $\mathbb{D}.$

%\begin{example}
%\label{ex15}
%For the discrete polymatroid considered in Example \ref{ex5}, the excluded vectors are the points indicated by `x' in Fig. \ref{example_2d_dep}.
%\begin{figure}[htbp]
%\centering
%\includegraphics[totalheight=2.7 in,width=3.6in]{example_2d_dep.eps}
%\caption{Diagram showing the excluded vectors for the discrete polymatroid defined in Example \ref{ex5}}
%\label{example_2d_dep}
%\end{figure}
%\end{example}

\begin{example}
\label{ex16}
For the discrete polymatroid considered in Example \ref{ex7}, the set of excluded vectors is given by $\{(0,2,1),(1,2,1),(2,1,1),(2,2,0),(2,2,1)\}.$
\end{example}

The notion of excluded vectors for discrete polymatroids can be viewed as the generalization of the notion of dependent sets for matroids. For a matroid $\mathbb{M},$ the set of excluded vectors for $\mathbb{D}(\mathbb{M})$ uniquely determines the set of dependent sets for $\mathbb{M}.$ The set of dependent sets for $\mathbb{M}$ is given by $\{(u)_{>0}: u \in \mathcal{D}(\mathbb{D}(\mathbb{M}))\}.$ Conversely, for a dependent set $D$ for $\mathbb{M},$ the vector $\sum_{i \in D} \epsilon_{i,r}$ is an excluded vector for $\mathbb{D}(\mathbb{M}).$
\begin{example}
\label{ex17}
For the uniform matroid $U_{2,4}$ considered in Example \ref{ex11}, the set of dependent sets is given by $\{\{1,2,3\},\{1,2,4\},\{1,3,4\},\{2,3,4\},\{1,2,3,4\}\}.$ The set of excluded vectors for $\mathbb{D}(U_{2,4})$ is given by $\{(1,1,1,0),(1,1,0,1),(1,0,1,1),(0,1,1,1),(1,1,1,1)\}.$
\end{example}

We define a minimal excluded vector for a discrete polymatroid $\mathbb{D}$ as follows:
\begin{definition}
An excluded vector $u \in \mathcal{D}(\mathbb{D})$ is said to be a minimal excluded vector, if there does not exist $v \in \mathcal{D}(\mathbb{D})$ for which $v<u.$
\end{definition}

 Let $\mathcal{C}(\mathbb{D})$ denote the set of minimal excluded vectors for the discrete polymatroid $\mathbb{D}.$ 
%\begin{example}
%\label{ex18}
%For the discrete polymatroid considered in Example \ref{ex5}, the minimal excluded vectors are $(1,5),$ $(2,4)$ and $(3,3).$
%\end{example}

\begin{example}
\label{ex199}
For the discrete polymatroid considered in Example \ref{ex7}, the set of minimal excluded vectors is given by $\{(0,2,1),(2,1,1),(2,2,0)\}.$
\end{example}

The notion of minimal excluded vectors for discrete polymatroids can be viewed as the generalization of the notion of circuits for matroids.
The set of minimal excluded vectors for the discrete polymatroid $\mathbb{D}(\mathbb{M})$ uniquely determines the set of circuits for the matroid $\mathbb{M}.$ The set of circuits of $\mathbb{M}$ is given by $\{(u)_{>0}: u \in \mathcal{C}(\mathbb{D}(\mathbb{M}))\}.$ Conversely, for a circuit $C$ for $\mathbb{M},$ the vector $\sum_{i \in C} \epsilon_{i,r}$ is a minimal excluded vector for $\mathbb{D}(\mathbb{M}).$

\begin{example}
\label{ex200}
For the uniform matroid $U_{2,4}$ considered in Example \ref{ex11}, the set of circuits is given by $\{\{1,2,3\},\{1,2,4\},\{1,3,4\},\{2,3,4\}\}.$ The set of minimal excluded vectors for the discrete polymatroid $\mathbb{D}(U_{2,4})$ is $\{(1,1,1,0),(1,1,0,1),(1,0,1,1),(0,1,1,1)\}.$
\end{example}

\section{Network Coding and Index Coding - Preliminaries}
In Section \ref{III A}, the basic definitions and notations related to networks and their solvability are defined. In Section \ref{III B}, the preliminaries related to the index coding problem are provided.
\subsection{Network Coding}
\label{III A}
A communication network consists of a directed acyclic graph without self-loops, with the set of vertices denoted by $\mathcal{V}$ and the set of edges denoted by $\mathcal{E}.$ For an edge $e$ directed from $x$ to $y,$ $x$ is called the head vertex of $e$ denoted by $head(e)$ and $y$ is called the tail vertex of $e$ denoted by $tail(e).$ The in-degree of an edge $e$ is the in-degree of its head vertex and out-degree of $e$ is the out-degree of its tail vertex. The messages in the network are generated at edges with in-degree zero, which are called the input edges of the network and let $\mathcal{S}  \subset \mathcal{E}$ denote the set of input edges with $\vert \mathcal{S} \vert =m.$ Let $x_i, i \in \lceil m \rfloor,$ denote the row vector of length $k_i$ generated at the $i^{\text{th}}$ input edge of the network. Let $x=[ x_1,x_2,\dotso,x_m ]$ denote the row vector obtained by the concatenation of the $m$ message vectors. An edge which is not an input edge is referred to as an intermediate edge. All the intermediate edges in the network are assumed to carry a vector of dimension $n$ over $\mathbb{F}_q.$ A vertex $v \in \mathcal{V}$ demands the set of messages generated at the input edges given by $\delta(v) \subseteq \mathcal{S},$ where $\delta$ is called the demand function of the network. $In(v)$ denotes the set of incoming edges of a vertex $v $ ($In(v)$ includes the intermediate edges as well as the input edges which are incoming edges at node $v$) and $Out(v)$ denotes the union of the set of intermediate edges originating from $v$ and $\delta(v).$

A $(k_1,k_2,\dotso,k_m;n)$-FNC solution over $\mathbb{F}_q$ is a collection of functions $\lbrace \psi_e : \mathbb{F}_q^{\sum_{i=1}^{m}k_i} \rightarrow \mathbb{F}_q^{k_i} , e \in \mathcal{S}\rbrace \cup \lbrace \psi_e : \mathbb{F}_q^{\sum_{i=1}^{m}k_i} \rightarrow \mathbb{F}_q^{n} , e \in \mathcal{E}\setminus\mathcal{S}\rbrace,$ where the function $\psi_e$ is called the global encoding function associated with edge $e.$  The global encoding functions satisfy the following conditions:
\begin{description}
\item [(N1):]
$\psi_i(x)=[x_i], \forall i \in \mathcal{S},$
\item [(N2):]
For every $v \in \mathcal{V},$ for all $j \in \delta(v),$ there exists a function $\chi_{v,j} : \mathbb{F}_{q}^{n\vert In(v)\vert} \rightarrow \mathbb{F}_q^{k_j}$ called the decoding function for message $j$ at node $v$ which satisfies \\$\chi_{v,j}(\psi_{i_1}(x),\psi_{i_2}(x), \dotso, \psi_{i_t(x)})=x_{j},$  where $In(v)=\{i_1,i_2, \dotso i_t \} .$
\item [(N3):]
For all \mbox{$i \in \mathcal{E} \setminus \mathcal{S},$} there exists \mbox{$\phi_i: \mathbb{F}_q^{n \vert In(head(i))\vert} \rightarrow \mathbb{F}_q^n$} such that \\ \mbox{$\psi_i(x)=\phi_i(\psi_{i_1}(x),\psi_{i_2}(x), \dotso, \psi_{i_r}(x)),$}  where $In(head(i))=\{i_1,i_2, \dotso i_r \}.$ The function $\phi_i$ is called the local encoding function associated with edge $i.$ 
\end{description}
%

%Note that the dimension of the $i^{th}$ message vector $k_i$ need not necessarily be lesser than the edge vector dimension $n.$ For example, for the network considered in Fig. \ref{frac_routing}, in which nodes 4 and 5 respectively demand messages $x_1$ and $x_2,$ a (2,2;1)-FNC solution exits which is in fact a fractional routing solution.

%\begin{figure}[t]
%\centering
%\includegraphics[totalheight=2.5 in,width=1.5in]{frac_routing.eps}
%\caption{A network which admits a $(2,2;1)$-FNC solution}
%\label{frac_routing}
%\end{figure}   

An FNC solution with \mbox{$k_1=k_2=\dotso=k_m=n=1$} is called a scalar solution and an FNC solution for which \mbox{$k_1=k_2=\dotso=k_m=n=k$} is called a vector solution of dimension $k.$ 
A solution for which all the local encoding functions and hence the global encoding functions are linear is said to be a linear solution. For a linear $(k_1,k_2,\dotso,k_m;n)$-FNC solution, the global encoding function $\psi_i, i \in \mathcal{E}\setminus \mathcal{S},$ is of the form $\psi_i(x)= x M_i,$ where $M_i$ is an $\sum_{i=1}^{m}k_i \times n$ matrix over $\mathbb{F}_q$ called the global encoding matrix associated with edge $i.$

If a network admits a $(k_1,k_2,\dotso,k_m;n)$-FNC solution, then $(k_1/n,k_2/n,\dotso,k_m/n)$ is said to be an achievable rate vector and the scalar $\frac{1}{m}\sum_{i=1}^m \frac{k_i}{n}$ is said to be an achievable average rate \cite{DoFrZe_FNC}. The closure of the set of all achievable rate vectors is said to be the achievable rate region of the network and the supremum of all achievable average rates is said to be the average coding capacity of the network \cite{DoFrZe_FNC}. %A $(k_1,k_2,\dotso,k_m;n)$-FNC solution is said to be average rate optimal if $\frac{1}{m}\sum_{i=1}^m \frac{k_i}{n}$ is equal to the average coding capacity of the network.
 A $(k,k,\dotso,k;n)$-FNC solution is said to be a uniform FNC solution and the scalar $k/n$ is called a uniform achievable rate. The supremum of all uniform achievable rates is defined to be the uniform coding capacity of the network \cite{DoFrZe_FNC}.
\subsection{Index Coding}
\label{III B}
Most of the definitions and notations in this subsection have been adapted from \cite{RoSpGe}.

An index coding problem $\mathcal{I}(X,\mathcal{R})$ includes
\begin{itemize}
\item
a set of messages $X=\{x_1,x_2,\dotso,x_m\}$ and 
\item
a set of receiver nodes $\mathcal{R}\subseteq \{(x,H);x \in X, H \subseteq X \setminus \{x\}\}.$
\end{itemize}

For a receiver node $R=(x,H)\in \mathcal{R},$ $x$ denotes the message demanded by $R$ and $H$ denotes the side information possessed by $R.$ 
Each one of the messages $x_i, i \in \{1,2,\dotso,m\},$ is assumed to be row vectors of length $n,$ over an alphabet set, which in this paper is assumed to be a finite field $\mathbb{F}_q$ of size $q.$   Let $y=[x_1\;x_2\dotso x_m]$ denote the row vector of length $nm$ obtained by the concatenation of the $m$ message vectors.

An index coding solution (also referred to as an index code) over $\mathbb{F}_q$ of length $c$ and dimension $n$ for the index coding problem $\mathcal{I}(X,\mathcal{R})$ is a function $f: \mathbb{F}_q^{nm} \rightarrow \mathbb{F}_q^{c},$ $c$ an integer, which satisfies the following condition: For every $R=(x,H) \in \mathcal{R},$ there exists a function $\psi_R: \mathbb{F}_q^{n \vert H \vert +c} \rightarrow \mathbb{F}_q^{n}$ such that $\displaystyle{\psi_R((x_i)_{i \in H},f(y))=x,\forall y \in \mathbb{F}^{nm}_q.}$
The function $\psi_R$ is referred to as the decoding function at receiver $R.$

An index coding solution for which $n=1$ is called a scalar solution; otherwise it is called a vector solution. An index coding solution is said to be linear if the functions $f$ and $\psi_R$ are linear.

For an index coding problem $\mathcal{I}(X,\mathcal{R}),$ define 
{\small$\displaystyle{\mathcal{M}(\mathcal{I}(X,\mathcal{R}))=\max_{Y \subseteq X}\vert \{R =(x,H)\in \mathcal{R}: H=Y\}\vert.}$} The length $c$ and dimension $n$ of an index coding solution for the index coding problem $\mathcal{I}(X,\mathcal{R})$ satisfy the condition $c/n \geq \mathcal{M}(\mathcal{I}(X,\mathcal{R}))$ \cite{RoSpGe}. 
\begin{definition}[\cite{RoSpGe}]
An index coding solution for which $c/n = \mathcal{M}(\mathcal{I}(X,\mathcal{R}))$ is said to be a perfect index coding solution.
\begin{example}
Consider the index coding problem with the message set$X=\{x_1,x_2,x_3,x_4\}$ and 
the set of receiver nodes $$\mathcal{R}=\{(x_3,\{x_1,x_2\}),(x_4,\{x_1,x_2\})(x_1,\{x_2,x_3,x_4\}),$$ $$\hspace{6 cm}(x_2,\{x_1,x_3,x_4\})\}.$$ For this index coding problem, in order to satisfy the demands of receiver nodes $(x_3,\{x_1,x_2\})$ and $(x_4,\{x_1,x_2\})$ which contain the same side information, we need to have $c/n \geq 2.$ In other words, for this index coding problem, we have, $\mathcal{M}(\mathcal{I}(X,\mathcal{R}))=2.$ A scalar perfect linear index coding solution over $\mathbb{F}_q$ with $c=2$ exists for this index coding problem and is given by $f(X)=[x_1+x_2+x_3; x_3+x_4].$
\end{example}
\end{definition}
\section{Linear Fractional Network Coding and Representable Discrete Polymatroids}
In this section, we obtain results on the connection between representable discrete polymatroids and linear FNC. In Section \ref{IV A}, the notion of a \textit{$(k_1,k_2,\dotso,k_m;n)$-discrete polymatroidal network} is introduced and it is shown that a linear $(k_1,k_2,\dotso,k_m;n)$-FNC solution exists for a network  if and only if the network is $(k_1,k_2,\dotso,k_m;n)$-discrete polymatroidal with respect to a discrete polymatroid representable over $\mathbb{F}_q.$ In Section \ref{IV B}, an algorithm to construct networks from a class of discrete polymatroids is provided. If the discrete polymatroid from which the network is constructed is representable over $\mathbb{F}_q,$ then the constructed network admits a linear FNC solution over $\mathbb{F}_q.$

%In this section, the focus is on the vector linear solvability of networks. In Section \ref{IV A}, the notion of a discrete polymatroidal network is defined and it is shown that a network is vector linearly solvable over $\mathbb{F}_q$ if and only if it is discrete polymatroidal with respect to a discrete polymatroid representable over $\mathbb{F}_q.$ In Section \ref{IV B}, an algorithm to construct networks from a class of discrete polymatroids is provided. If the discrete polymatroid from which the network is constructed is representable over $\mathbb{F}_q,$ then the constructed network has a vector linear solution over $\mathbb{F}_q.$ The results presented in this section for vector linear solutions are generalized to FNC solutions in Section V.
\subsection{Linear Fractional Solvability of Networks and Representation of Discrete Polymatroids}
\label{IV A}
The notion of a matroidal network was introduced by Dougherty et. al. in \cite{DoFrZe}. In \cite{DoFrZe}, it was shown that if a scalar linear solution exists for a network, then the network is matroidal with respect to a representable matroid. The converse that a scalar linear solution exists for a network if the network is matroidal with respect to a representable matroid was shown in \cite{KiMe}. In this section, we generalize this result to networks which admit linear FNC solutions.

%For a discrete polymatroid $\mathbb{D},$ let $\rho_{max}(\mathbb{D})=\max_{i \in \lceil r \rfloor} \rho(\lbrace i \rbrace).$ 

We define a $(k_1,k_2,\dotso,k_m;n)$-discrete polymatroidal network as follows:
\begin{definition}
\label{defn_DPMN}
A network is said to be $(k_1,k_2,\dotso,k_m;n)$-discrete polymatroidal with respect to a discrete polymatroid $\mathbb{D},$ if there exists a map $f: \mathcal{E} \rightarrow \lceil r \rfloor$ which satisfies the following conditions:
\begin{description}
\item [(DN1):]
~~~$f$ is one-to-one on the elements of $\mathcal{S}.$
\item [(DN2):]
~~~\mbox{\small$\sum_{i \in f(\mathcal{S})} k_i \epsilon_{i,r}\in \mathbb{D}.$}
\item [(DN3):]
~~~\mbox{\small$\forall i \in f(\mathcal{S}),$} \mbox{\small$\rho(\{i\})=k_i$} and \mbox{\small$\displaystyle{\max_{i\in \mathcal{E}\setminus \mathcal{S}} \rho(f(\{i\}))=n}.$}
\item [(DN4):]
~~~\mbox{\small$\rho(f(In(x)))=\rho(f(In(x) \cup Out(x))), \forall x \in \mathcal{V}.$}
\end{description}
\end{definition}

 It can be verified that a network is matroidal with respect to a matroid $\mathbb{M}$ if and only if it is $(1,1,\dotso,1;1)$-discrete polymatroidal with respect to $\mathbb{D}(\mathbb{M}).$ In this way, for a discrete polymatroid $\mathbb{D}(\mathbb{M}),$ the notion of a discrete polymatroidal network with respect to $\mathbb{D}(\mathbb{M})$ is equivalent to the notion of a matroidal network with respect to $\mathbb{M}.$  

%From the results in \cite{DoFrZe} and \cite{KiMe}, it follows that a network has scalar linear solution over $\mathbb{F}_q$ if and only if the network is matroidal with respect to a matroid representable over $\mathbb{F}_q.$ In the following theorem, we provide a generalization of this result for vector linear solvable networks, in terms of the representability of discrete polymatroids. 

The connection between the linear fractional solvablity over $\mathbb{F}_q$ for a network and the network being discrete polymatroidal with respect to a discrete polymatroid representable over $\mathbb{F}_q$ is established in the following theorem. 

\begin{theorem}
\label{thm3}
A network has a linear $(k_1,k_2,\dotso,k_m;n)$-FNC solution over $\mathbb{F}_q,$ if and only if it is $(k_1,k_2,\dotso,k_m;n)$-discrete polymatroidal with respect to a discrete polymatroid  $\mathbb{D}$ representable over $\mathbb{F}_q.$
\end{theorem}
\begin{proof}
Let the edge set $\mathcal{E}$ of the network be $\lceil l \rfloor$ and let the message set $\mathcal{S}$ be $\lceil m \rfloor.$ The edges are assumed to be arranged in the ancestral ordering which exists since the networks considered are acyclic and the set of intermediate edges in the network is assumed to be $\{m+1,m+2,\dotso l\}.$ We first prove the `if' part of the theorem. Assume that the network considered is $(k_1,k_2,\dotso,k_m;n)$-discrete polymatroidal with respect to a representable discrete polymatroid $\mathbb{D}(V_1,V_2,\dotso,V_r)$ on the ground set $\lceil r \rfloor.$ For brevity, the discrete polymatroid $\mathbb{D}(V_1,V_2,\dotso, V_r)$ is denoted as $\mathbb{D}.$ Let $f$ denote the network-discrete polymatroid mapping. Since, $f$ is one-to-one on the elements of $\mathcal{S},$ assume $f(\mathcal{S})$ to be $\lceil m \rfloor.$

It is claimed that without loss of generality, the set $\lceil r \rfloor$ can be taken to be the image of the map $f.$ Otherwise, if the image of the map $f$ is $\{i_1,i_2,\dotso i_t\},$ then we show that the network is $(k_1,k_2,\dotso,k_m;n)$-discrete polymatroidal with respect to the discrete polymatroid $\mathbb{D}(V_{i_1},V_{i_2},\dotso,V_{i_t}),$ with the same network-discrete polymatroid mapping $f.$ (DN1), (DN3) and (DN4) follow from the fact that the network is discrete polymatroidal with respect to $\mathbb{D}(V_1,V_2,\dotso,V_r).$ To show that (DN2) is satisfied, it needs to be shown that the vector $u = \sum_{i \in \lceil m \rfloor} k_i \epsilon_{i,t} \in \mathbb{D}(V_{i_1},V_{i_2},\dotso,V_{i_t}).$ Let $v$ denote the vector defined as $\sum_{i \in \lceil m \rfloor} k_i \epsilon_{i,r}.$ Since, the network is discrete polymatroidal with respect to $\mathbb{D}(V_1,V_2,\dotso,V_r),$ from (DN2), we have,
\begin{equation}
\label{eqn_thm1_proof}
\vert v(A) \vert \leq dim\left(\sum_{j \in A} V_j\right), \forall A \subseteq \lceil r \rfloor.
\end{equation}
To show that $u \in \mathbb{D}(V_{i_1},V_{i_2},\dotso,V_{i_t}),$ it needs to be shown that $\vert u (A)\vert \leq dim(\sum_{j \in A} V_{i_j}), \forall A \subseteq \lceil t \rfloor$ which follows from \eqref{eqn_thm1_proof} and from the fact that any subset of $\{i_1,i_2,\dotso,i_t\}$ is also a subset of $\lceil r \rfloor.$ 

Next it will be shown that $dim(\sum_{i=1}^r V_i)=\sum_{i=1}^m k_i.$ Define $s_0=\lceil m \rfloor.$ Let $s_1=s_0 \cup \lbrace f(m+1)\rbrace.$ Since the edges in the set $\{m+1,m+2, \dotso, l \}$ are arranged in ancestral ordering, $In(head(m+1))$ is contained in $s_0.$ Hence, from (DN3) we have $\rho(s_1)=dim(\sum_{i \in s_0} V_i + V_{f(m+1)})=  dim(\sum_{i \in s_0}V_i)=\rho(s_0).$ Iteratively, defining $s_{i+1}=s_{i} \cup \lbrace f(m+i+1)\rbrace,$ using a similar argument, we have $\rho(s_{i+1})=\rho(s_0).$ Hence, we have $\rho(s_{l-m})=\rho(s_0)=\rho(\lceil m \rfloor).$ But $s_{l-m}=\lceil r \rfloor,$ since the image of $f$ is $\lceil r \rfloor.$ Hence, we have, $\rho(\lceil r \rfloor)=\rho(\lceil m \rfloor).$ From (DN2), we have $\sum_{i \in \lceil m\rfloor}k_i \epsilon_{i,r} \in \mathbb{D}.$ Hence from the definition of a discrete polymatroid, we have $\sum_{i=1}^m k_i \leq \rho(\lceil m \rfloor).$ From (D2), we have $\rho(\lceil m \rfloor)\leq \rho(\{1\})+\rho(\{2,3,\dotso,m\})\dotso \leq \sum_{i=1}^m \rho(\{i\}).$ We have $\rho(\lceil m \rfloor)\leq \sum_{i=1}^m k_i,$ since from (DN3) $\rho(\{i\})=k_i,$ for $i \in f(\mathcal{S}).$ As a result $dim(\sum_{i=1}^r V_i)=\rho(\lceil r \rfloor)=\rho(\lceil m\rfloor)=\sum_{i=1}^m k_i.$ 

The vector subspace $V_i, i \in \lceil r \rfloor, i \notin \lceil m \rfloor$ can be described by a matrix $A_i$ of size $\sum_{i=1}^m k_i \times n$ whose columns span $A_i.$ For $i \in \lceil m \rfloor,$ the vector subspace $V_i$ can be written as the column span of a matrix $A_i$ of size $\sum_{i=1}^m k_i \times k_i.$ Let $B=[A_1 A_2 \dotso A_m].$ Since $dim({\sum_{i=1}^m}V_i)=\sum_{i=1}^m k_i,$ $B$ is invertible and can be taken to be the $\sum_{i=1}^m k_i \times \sum_{i=1}^m k_i$ identity matrix (Otherwise, it is possible to define $A'_i=B^{-1} A_i$ and $V'_i$ to be the column span of $A'_i$ so that $D(V'_1,V'_2,\dotso,V'_r)=D(V_1,V_2,\dotso,V_r)$).

The claim is that taking the global encoding matrix of edge $i$ to be $A_{f(i)}$ forms a \\$(k_1,k_2,\dotso,k_m;n)$-FNC solution for the network. The proof of the claim is as follows: Since $B$ is an identity matrix, $A_i x=x_i$ for $i \in \lceil m \rfloor$ and hence (N1) is satisfied. For any node $v$ in the network, from (DN4) it follows that $dim(\sum_{i \in In(v) \cup Out(v)}V_{f(i)})= dim(\sum_{i \in In(v)}V_{f(i)}).$ Hence, $\forall j \in Out(v),$ $A_{f(j)}$ can be written as $\sum_{i \in In(V)} W_i A_{f(i)}.$ Hence, (N2) and (N3) are satisfied. This completes the `if' part of the proof.

For the `only if' part of the proof, assume that the network considered admits a $(k_1.k_2,\dotso,k_m;n)$-FNC solution, with $A_i, i \in \lceil l \rfloor,$ being the global encoding matrix associated with edge $i.$ Consider the discrete polymatroid $D(V_1,V_2,\dotso,V_l),$ where $V_i$ denotes the column span of $A_i.$ Let $f(i)=i, i\in \lceil l \rfloor$ be the mapping from the edge set of the network to the ground set of the discrete polymatroid. It can be verified that the network is $(k_1,k_2,\dotso;n)$-discrete polymatroidal with respect to $\mathbb{D}(V_1,V_2,\dotso,V_l).$ 
\end{proof}

Specializing for $k_i=n=1,i \in \lceil m \rfloor,$ from Theorem \ref{thm3}, we obtain the following corollary:

\begin{corollary}
\label{cor_scalar}
A scalar linear solution exists for a network over $\mathbb{F}_q$ if and only if the network is matroidal with respect to a matroid representable over $\mathbb{F}_q.$ 
\end{corollary}

Note that the statement in Corollary \ref{cor_scalar} is more general than the statement of Theorem 13 in \cite{KiMe} stated as follows: ``\textit{A network is scalar-linearly solvable over a finite field of characteristic p if and only if the network is a matroidal network associated with a representable matroid over a finite field of characteristic p.}'' For a network which is matroidal with respect to a matroid representable over a field $\mathbb{F}_q,$ Theorem 13 in \cite{KiMe} implies that a scalar linear solution exists for the network over a sufficiently large field whose characteristic is the same as that of $\mathbb{F}_q.$ In contrast, the result in Corollary \ref{cor_scalar} above implies that such a scalar linear solution exists over the field $\mathbb{F}_q$ itself, and there is no need to look for solutions over larger fields. 

It is important to note that the discrete polymatroid $\mathbb{D}$ in Theorem \ref{thm3} needs not be unique.  A network can admit more than one linear FNC solution over $\mathbb{F}_q$ and from these solutions it may be possible to obtain multiple discrete polymatroids with respect to which the network is discrete polymatroidal, as illustrated in Example \ref{ex18} below. 

Also, note that Theorem \ref{thm3} characterizes the linear fractional solvability of a network in terms of discrete polymatroid representation. As mentioned earlier in Section \ref{II C}, not all representations of discrete polymatroids can be viewed as the multi-linear representations of matroids. Vector linear solvability of networks cannot be characterized using multi-linear representations of matroids, whereas they can be characterized using representations of discrete polymatroids. This fact is also illustrated in Example \ref{ex18} below.

In Example \ref{ex18} below, we consider the popular example of M-network introduced in \cite{MeEfHoKa}, which was shown to have a 2 dimensional vector linear solution, which is in fact a vector routing solution, but does not admit scalar linear solution over any field. It was shown in \cite{DoFrZe} that the M-network is not matroidal with respect to any representable matroid. But since the M-network admits a vector linear solution of dimension 2, from Theorem 1, it follows that the M-network is $(2,2,2,2;2)$-discrete polymatroidal with respect to a representable discrete polymatroid, as discussed in the following example.
\begin{example}
\label{ex18}
\begin{figure*}[t]
\tiny
\begin{align}
\label{matrix_M_network_1}
&\hspace{3 cm}\underbrace{\hspace{-.01 cm}\left[\begin{matrix} 1&\hspace{-.2 cm}0\\ 0&\hspace{-.2 cm}1\\0&\hspace{-.2 cm}0\\0&\hspace{-.2 cm}0\\0&\hspace{-.2 cm}0\\0&\hspace{-.2 cm}0\\0&\hspace{-.2 cm}0 \\0&\hspace{-.2 cm}0\end{matrix}\right.}_{A_1} \quad\quad \underbrace{\begin{matrix} 0&\hspace{-.2 cm}0\\ 0&\hspace{-.2 cm}0\\1&\hspace{-.2 cm}0\\0&\hspace{-.2 cm}1\\0&\hspace{-.2 cm}0\\0&\hspace{-.2 cm}0\\0&\hspace{-.2 cm}0\\0&\hspace{-.2 cm}0 \end{matrix}}_{A_2} \quad\quad \underbrace{\begin{matrix} 0&\hspace{-.2 cm}0\\ 0&\hspace{-.2 cm}0\\0&\hspace{-.2 cm}0\\0&\hspace{-.2 cm}0\\1&\hspace{-.2 cm}0\\0&\hspace{-.2 cm}1\\0&\hspace{-.2 cm}0\\0&\hspace{-.2 cm}0 \end{matrix}}_{A_3} \quad\quad \underbrace{\begin{matrix} 0&\hspace{-.2 cm}0\\ 0&\hspace{-.2 cm}0\\0&\hspace{-.2 cm}0\\0&\hspace{-.2 cm}0\\0&\hspace{-.2 cm}0\\0&\hspace{-.2 cm}0\\1&\hspace{-.2 cm}0\\0&\hspace{-.2 cm}1 \end{matrix}}_{A_4} \quad\quad \underbrace{\begin{matrix} 1&\hspace{-.2 cm}0\\ 0&\hspace{-.2 cm}0\\0&\hspace{-.2 cm}0\\0&\hspace{-.2 cm}1\\0&\hspace{-.2 cm}0\\0&\hspace{-.2 cm}0\\0&\hspace{-.2 cm}0\\0&\hspace{-.2 cm}0 \end{matrix}}_{A_5} \quad\quad \underbrace{\begin{matrix} 0&\hspace{-.2 cm}0\\ 1&\hspace{-.2 cm}0\\0&\hspace{-.2 cm}1\\0&\hspace{-.2 cm}0\\0&\hspace{-.2 cm}0\\0&\hspace{-.2 cm}0\\0&\hspace{-.2 cm}0\\0&\hspace{-.2 cm}0 \end{matrix}}_{A_6} \quad\quad \underbrace{\begin{matrix} 0&\hspace{-.2 cm}0\\ 0&\hspace{-.2 cm}0\\0&\hspace{-.2 cm}0\\0&\hspace{-.2 cm}0\\1&\hspace{-.2 cm}0\\0&\hspace{-.2 cm}0 \\0&\hspace{-.2 cm}0\\0&\hspace{-.2 cm}1\end{matrix}}_{A_7} \quad\quad \underbrace{\begin{matrix} 0&\hspace{-.2 cm}0\\ 0&\hspace{-.2 cm}0\\0&\hspace{-.2 cm}0\\0&\hspace{-.2 cm}0\\0&\hspace{-.2 cm}0\\1&\hspace{-.2 cm}0\\0&\hspace{-.2 cm}1\\0&\hspace{-.2 cm}0 \end{matrix}}_{A_8}\quad\quad \underbrace{\begin{matrix} 0&\hspace{-.2 cm}0\\ 1&\hspace{-.2 cm}0\\0&\hspace{-.2 cm}0\\0&\hspace{-.2 cm}0\\0&\hspace{-.2 cm}1\\0&\hspace{-.2 cm}0\\0&\hspace{-.2 cm}0\\0&\hspace{-.2 cm}0 \end{matrix}}_{A_9}\quad\quad \underbrace{\begin{matrix} 0&\hspace{-.2 cm}0\\ 1&\hspace{-.2 cm}0\\0&\hspace{-.2 cm}0\\0&\hspace{-.2 cm}0\\0&\hspace{-.2 cm}0\\0&\hspace{-.2 cm}0\\0&\hspace{-.2 cm}0\\0&\hspace{-.2 cm}1 \end{matrix}}_{A_{10}}\quad\quad \underbrace{\begin{matrix} 0&\hspace{-.2 cm}0\\ 0&\hspace{-.2 cm}0\\1&\hspace{-.2 cm}0\\0&\hspace{-.2 cm}0\\0&\hspace{-.2 cm}1\\0&\hspace{-.2 cm}0\\0&\hspace{-.2 cm}0\\0&\hspace{-.2 cm}0 \end{matrix}}_{A_{11}}\quad\quad \underbrace{\left.\begin{matrix} 0&\hspace{-.2 cm}0\\ 0&\hspace{-.2 cm}0\\1&\hspace{-.2 cm}0\\0&\hspace{-.2 cm}0\\0&\hspace{-.2 cm}0\\0&\hspace{-.2 cm}0\\0&\hspace{-.2 cm}0\\0&\hspace{-.2 cm}1 \end{matrix}\right]}_{A_{12}}
\\
\nonumber
\\
 \hline
 \nonumber
 \\
\nonumber
&\underbrace{\hspace{-.01 cm}\left[\begin{matrix} 1&\hspace{-.2 cm}0\\ 0&\hspace{-.2 cm}1\\0&\hspace{-.2 cm}0\\0&\hspace{-.2 cm}0\\0&\hspace{-.2 cm}0\\0&\hspace{-.2 cm}0\\0&\hspace{-.2 cm}0 \\0&\hspace{-.2 cm}0\end{matrix}\right.}_{A'_1} \quad\quad \underbrace{\begin{matrix} 0&\hspace{-.2 cm}0\\ 0&\hspace{-.2 cm}0\\1&\hspace{-.2 cm}0\\0&\hspace{-.2 cm}1\\0&\hspace{-.2 cm}0\\0&\hspace{-.2 cm}0\\0&\hspace{-.2 cm}0\\0&\hspace{-.2 cm}0 \end{matrix}}_{A'_2} \quad\quad \underbrace{\begin{matrix} 0&\hspace{-.2 cm}0\\ 0&\hspace{-.2 cm}0\\0&\hspace{-.2 cm}0\\0&\hspace{-.2 cm}0\\1&\hspace{-.2 cm}0\\0&\hspace{-.2 cm}1\\0&\hspace{-.2 cm}0\\0&\hspace{-.2 cm}0 \end{matrix}}_{A'_3} \quad\quad \underbrace{\begin{matrix} 0&\hspace{-.2 cm}0\\ 0&\hspace{-.2 cm}0\\0&\hspace{-.2 cm}0\\0&\hspace{-.2 cm}0\\0&\hspace{-.2 cm}0\\0&\hspace{-.2 cm}0\\1&\hspace{-.2 cm}0\\0&\hspace{-.2 cm}1 \end{matrix}}_{A'_4} \quad\quad \underbrace{\begin{matrix} 1&\hspace{-.2 cm}0\\ 0&\hspace{-.2 cm}0\\0&\hspace{-.2 cm}0\\0&\hspace{-.2 cm}1\\0&\hspace{-.2 cm}0\\0&\hspace{-.2 cm}0\\0&\hspace{-.2 cm}0\\0&\hspace{-.2 cm}0 \end{matrix}}_{A'_5} \quad\quad \underbrace{\begin{matrix} 0&\hspace{-.2 cm}0\\ 1&\hspace{-.2 cm}0\\0&\hspace{-.2 cm}1\\0&\hspace{-.2 cm}0\\0&\hspace{-.2 cm}0\\0&\hspace{-.2 cm}0\\0&\hspace{-.2 cm}0\\0&\hspace{-.2 cm}0 \end{matrix}}_{A'_6} \quad\quad \underbrace{\begin{matrix} 0&\hspace{-.2 cm}0\\ 0&\hspace{-.2 cm}0\\0&\hspace{-.2 cm}0\\0&\hspace{-.2 cm}0\\1&\hspace{-.2 cm}0\\0&\hspace{-.2 cm}0 \\0&\hspace{-.2 cm}0\\0&\hspace{-.2 cm}1\end{matrix}}_{A'_7} \quad\quad \underbrace{\begin{matrix} 0&\hspace{-.2 cm}0\\ 0&\hspace{-.2 cm}0\\0&\hspace{-.2 cm}0\\0&\hspace{-.2 cm}0\\0&\hspace{-.2 cm}0\\1&\hspace{-.2 cm}0\\0&\hspace{-.2 cm}1\\0&\hspace{-.2 cm}0 \end{matrix}}_{A'_8}\quad\quad \underbrace{\begin{matrix} 0&\hspace{-.2 cm}0\\ 1&\hspace{-.2 cm}0\\0&\hspace{-.2 cm}0\\0&\hspace{-.2 cm}0\\0&\hspace{-.2 cm}1\\0&\hspace{-.2 cm}0\\0&\hspace{-.2 cm}0\\0&\hspace{-.2 cm}0 \end{matrix}}_{A'_9}\quad\quad \underbrace{\begin{matrix} 0&\hspace{-.2 cm}0\\ 1&\hspace{-.2 cm}0\\0&\hspace{-.2 cm}0\\0&\hspace{-.2 cm}0\\0&\hspace{-.2 cm}0\\0&\hspace{-.2 cm}0\\0&\hspace{-.2 cm}0\\0&\hspace{-.2 cm}1 \end{matrix}}_{A'_{10}}\quad\quad \underbrace{\begin{matrix} 0&\hspace{-.2 cm}0\\ 0&\hspace{-.2 cm}0\\1&\hspace{-.2 cm}0\\0&\hspace{-.2 cm}0\\0&\hspace{-.2 cm}1\\0&\hspace{-.2 cm}0\\0&\hspace{-.2 cm}0\\0&\hspace{-.2 cm}0 \end{matrix}}_{A'_{11}}\quad\quad \underbrace{\begin{matrix} 0&\hspace{-.2 cm}0\\ 0&\hspace{-.2 cm}0\\1&\hspace{-.2 cm}0\\0&\hspace{-.2 cm}0\\0&\hspace{-.2 cm}0\\0&\hspace{-.2 cm}0\\0&\hspace{-.2 cm}0\\0&\hspace{-.2 cm}1 \end{matrix}}_{A'_{12}}\quad\quad
%\hline 
%\label{matrix_M_network_2}
 \underbrace{\begin{matrix} 1&\hspace{-.2 cm}0\\ 0&\hspace{-.2 cm}0\\0&\hspace{-.2 cm}0\\0&\hspace{-.2 cm}0\\0&\hspace{-.2 cm}0\\0&\hspace{-.2 cm}0\\0&\hspace{-.2 cm}0 \\0&\hspace{-.2 cm}0\end{matrix}}_{A'_{13}} \quad\quad \underbrace{\begin{matrix} 1&\hspace{-.2 cm}0\\ 0&\hspace{-.2 cm}0\\0&\hspace{-.2 cm}0\\0&\hspace{-.2 cm}0\\0&\hspace{-.2 cm}0\\0&\hspace{-.2 cm}0\\0&\hspace{-.2 cm}0\\0&\hspace{-.2 cm}0 \end{matrix}}_{A'_{14}} \quad\quad \underbrace{\begin{matrix} 0&\hspace{-.2 cm}0\\ 0&\hspace{-.2 cm}0\\0&\hspace{-.2 cm}0\\1&\hspace{-.2 cm}0\\0&\hspace{-.2 cm}0\\0&\hspace{-.2 cm}0\\0&\hspace{-.2 cm}0\\0&\hspace{-.2 cm}0 \end{matrix}}_{A'_{15}} \quad\quad \underbrace{\begin{matrix} 0&\hspace{-.2 cm}0\\ 0&\hspace{-.2 cm}0\\0&\hspace{-.2 cm}0\\1&\hspace{-.2 cm}0\\0&\hspace{-.2 cm}0\\0&\hspace{-.2 cm}0\\0&\hspace{-.2 cm}0\\0&\hspace{-.2 cm}0 \end{matrix}}_{A'_{16}} \quad\quad \underbrace{\begin{matrix} 0&\hspace{-.2 cm}0\\ 0&\hspace{-.2 cm}0\\0&\hspace{-.2 cm}0\\0&\hspace{-.2 cm}0\\0&\hspace{-.2 cm}0\\1&\hspace{-.2 cm}0\\0&\hspace{-.2 cm}0\\0&\hspace{-.2 cm}0 \end{matrix}}_{A'_{17}} \quad\quad \underbrace{\begin{matrix} 0&\hspace{-.2 cm}0\\ 0&\hspace{-.2 cm}0\\0&\hspace{-.2 cm}0\\0&\hspace{-.2 cm}0\\0&\hspace{-.2 cm}0\\0&\hspace{-.2 cm}0\\1&\hspace{-.2 cm}0\\0&\hspace{-.2 cm}0 \end{matrix}}_{A'_{18}}\\ 
 \label{matrix_M_network_2}
&\hspace{14 cm}\quad\quad \underbrace{\begin{matrix} 0&\hspace{-.2 cm}0\\ 0&\hspace{-.2 cm}0\\0&\hspace{-.2 cm}0\\0&\hspace{-.2 cm}0\\0&\hspace{-.2 cm}0\\1&\hspace{-.2 cm}0 \\0&\hspace{-.2 cm}0\\0&\hspace{-.2 cm}0\end{matrix}}_{A'_{19}} \quad\quad \underbrace{\left.\begin{matrix} 0&\hspace{-.2 cm}0\\ 0&\hspace{-.2 cm}0\\0&\hspace{-.2 cm}0\\0&\hspace{-.2 cm}0\\0&\hspace{-.2 cm}0\\0&\hspace{-.2 cm}0\\1&\hspace{-.2 cm}0\\0&\hspace{-.2 cm}0 \end{matrix}\right]}_{A'_{20}}
\end{align}
\hrule
\end{figure*}
\begin{figure}[htbp]
\centering
\subfigure[Solution 1]{
\includegraphics[totalheight=3in,width=3in]{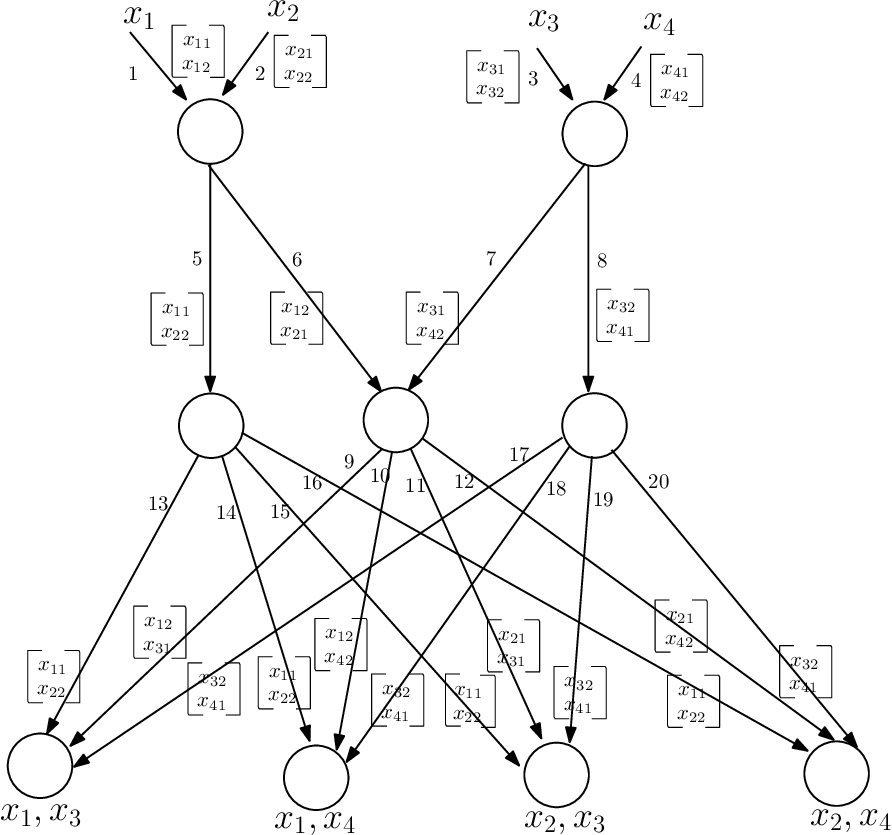}
\label{fig:sol1}
}
\subfigure[Solution 2]{
\includegraphics[totalheight=3in,width=3in]{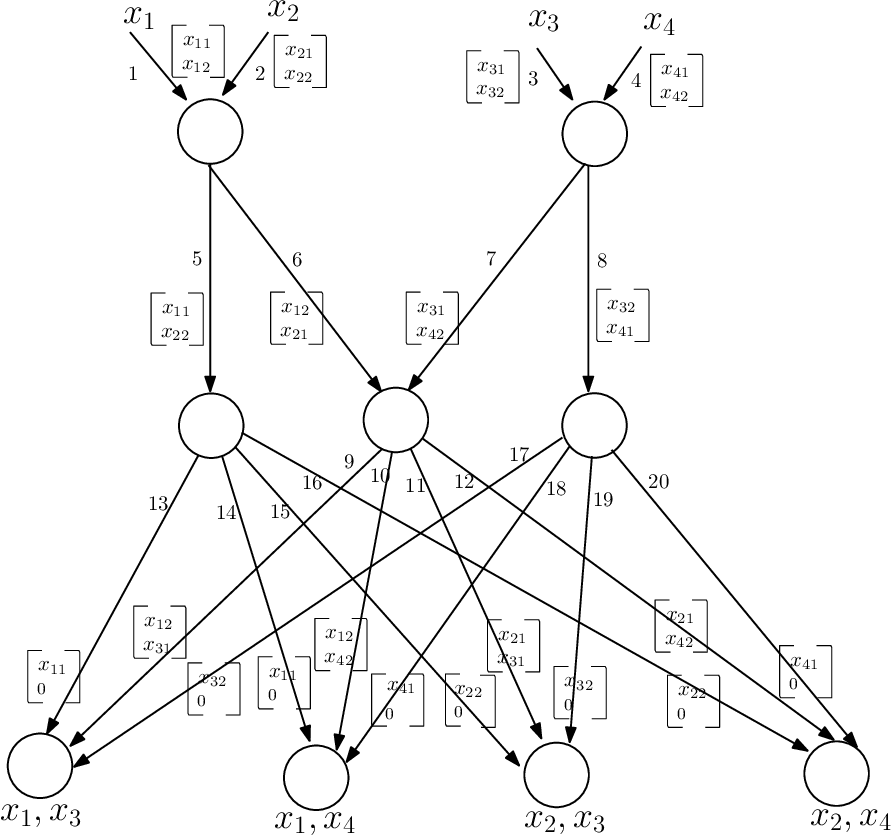}
\label{fig:sol2}
}
\caption{The M-network}
\label{fig:M_network}
\end{figure}
Consider the M-network shown in Fig. \ref{fig:M_network}. We consider two possible solutions for the M-network, from which it is possible to obtain two different discrete polymatroids with respect to which the M-network is $(2,2,2,2;2)$-discrete polymatroidal.\\ 
\underline{\textit{Solution 1:}}
 Assume the global encoding matrix of edge $i,i\in \lceil 12 \rfloor,$ to be the matrix $A_i$ given in \eqref{matrix_M_network_1} at the top of this page. Take $A_5$ to be the global encoding matrix of the  edges $13,14,15,16$ and  $A_8$ to be that of $17,18,19,20.$ The solution thus obtained for the M-network is as shown in Fig. \ref{fig:sol1}. Let the network-discrete polymatroid mapping $f_1$ be defined as follows: 
\begin{displaymath}
   f_1(i)= \left\{
     \begin{array}{ll}
       &i  :i \in \{1,2,\dotso, 12\}\\
       &5  :i \in \{\:13,14,15,16\}\\
       &8 : i \in \{\:17,18,19,20\}
     \end{array}.
   \right.
\end{displaymath}  
Define $V_i$ to be the column span of $A_i,i\in\lceil 12 \rfloor.$ It can be verified that the M-network is $(2,2,2,2;2)$-discrete polymatroidal with respect to $\mathbb{D}(V_1,V_2,\dotso V_{12}),$ with $f_1$ being the network-discrete polymatroid mapping. 

 From the definition of multi-linear representation, it follows that the vector subspaces (excluding the trivial zero vector subspaces) which form a multi-linear representation of dimension $k$ for a matroid should be $k$-dimensional. Note that the vector subspaces $V_i,i \in \lceil 12 \rfloor,$ have dimension 2 and they form a representation for the discrete polymatroid $\mathbb{D}(V_1,V_2,\dotso V_{12}).$ Despite having their dimensions to be equal, the vector subspaces $V_i,i \in \lceil 12 \rfloor,$ cannot form a multi-linear representation of dimension 2 for any matroid. The reason for this is that $dim(V_1+V_5)=3$ which is not a multiple of 2. \\
\underline{\textit{Solution 2:}}
 Assume the global encoding matrices of edge $i,i\in \lceil 20 \rfloor,$ to be the matrix $A'_i$ (defined in \eqref{matrix_M_network_2} at the top of the previous page). The solution thus obtained for the M-network is as shown in Fig. \ref{fig:sol2}. Let the network-discrete polymatroid mapping $f_2(i)=i, i \in \lceil 20 \rfloor.$
Define $V'_i$ to be the column span of $A'_i,i\in\lceil 20 \rfloor.$ It can be verified that the M-network is $(2,2,2,2;2)$-discrete polymatroidal with respect to $\mathbb{D}(V'_1,V'_2,\dotso V'_{20}),$ with $f_2$ being the network-discrete polymatroid mapping.

 Note that all the vector subspaces $V_i, i\in \lceil 12 \rfloor,$ in Solution 1 have the same dimension 2. In contrast, in Solution 2, the vector subspaces $V'_1,V'_2,\dotso, V'_{12}$ have dimension 2, while the vector subspaces $V'_{13},V'_{14},\dotso, V'_{20}$ have dimension 1. The M-network is $(2,2,2,2;2)$-discrete polymatroidal with respect to two different discrete polymatroids $\mathbb{D}(V_1,V_2,\dotso,V_{12})$ and $\mathbb{D}(V'_1,V'_2,\dotso,V'_{20}).$ %\endproof
 
 As shown in Theorem \ref{thm3} and illustrated in the previous example, there is a fundamental connection between linear fractional solvability of networks and representations of discrete polymatroids, whereas such a connection does not exist with multi-linear representations of matroids. %In the next subsection, some more examples of networks obtained from discrete polymatroids, which admit a vector linear solution over $\mathbb{F}_q$ but no scalar linear solution over $\mathbb{F}_q$ are provided.

\subsection{Construction of Linear Fractional Solvable Networks from Discrete Polymatroids}
\label{IV B}
In this section, an algorithm to construct networks from a class of discrete polymatroids is provided. The network constructed admits a linear FNC solution over $\mathbb{F}_q,$ if the discrete polymatroid from which it was constructed is representable over $\mathbb{F}_q.$

Note that a representable discrete polymatroid $ \mathbb{D}$ which arises in connection with linear FNC in Theorem 1 is not an arbitrary discrete polymatroid. It satisfies certain conditions which are obtained as follows: From (DN2), it follows that there exists a vector $\sum_{i \in f(\mathcal{S})}k_i \epsilon_{i,r}$ in $\mathbb{D}.$ From the proof of Theorem 1, it follows that $rank(\mathbb{D})=\sum_{i \in f(\mathcal{S})}k_i.$ Hence, $b=\sum_{i \in f(\mathcal{S})}k_i \epsilon_{i,r}$ is a basis vector for $\mathbb{D}.$ From (DN3) it follows that for this basis vector $b,$ $\forall i \in (b)_{>0},$ $\rho(\{i\})=k_i.$ 

Hence, every discrete polymatroid $\mathbb{D}$ which arises in connection with linear FNC in Theorem 1 satisfies the following condition: $\mathbb{D}$ contains a basis vector $b=\sum_{i \in (b)_{>0}}k_i \epsilon_{i,r} \in \mathcal{B}(\mathbb{D})$ for which  $\forall i \in (b)_{>0},$ $\rho(\{i\})=k_i.$ 

In this subsection, we restrict ourselves to only the class of discrete polymatroids which satisfy the above condition and provide an algorithm to construct networks from discrete polymatroids which belong to this class.

Before providing the construction algorithm, we provide some useful definitions.

Let $\mathcal{C}'_i(\mathbb{D})$ denote the set of minimal excluded vectors for $\mathbb{D}$ for which the $i^{th}$ component is one. The elements of $\mathcal{C}'_i(\mathbb{D})$ are referred to as \textit{$i$-unit minimal excluded vectors}. Let $\mathcal{C}_i(\mathbb{D})$ denote the set of vectors $u$ in  $\mathcal{C}'_i(\mathbb{D})$ which satisfy the condition that there does not exist $ v  \in \mathcal{C}'_i(\mathbb{D}),$ $v \neq u,$ for which  $(v)_{>0} \subset (u)_{>0}.$ The elements of $\mathcal{C}_i(\mathbb{D})$ are referred to as \textit{reduced $i$-unit minimal excluded vectors}.

\begin{example}
\label{ex19}
For the discrete polymatroid $2\mathbb{D}(U_{2,4})$ provided in Example \ref{ex13}, the set of minimal excluded vectors $\mathcal{C}(2\mathbb{D}(U_{2,4}))$ is given by 
{
\begin{align*}
&\{(0,1,2,2),(0,2,1,2),(0,2,2,1),(1,0,2,2),(1,1,1,2),\\
&(1,1,2,1),(1,2,0,2),(1,2,1,1),(1,2,2,0),(2,0,1,2),\\&(2,0,2,1),(2,1,0,2),(2,1,1,1),(2,1,2,0),(2,2,0,1),\\&\hspace{6.5 cm}(2,2,1,0)\}.
\end{align*}}

The set of $i$-unit minimal excluded vectors, $i \in \lceil 4 \rfloor$ is given by,
{
\begin{align*}
&\mathcal{C}'_{1}(2\mathbb{D}(U_{2,4}))=\{(1,0,2,2),(1,1,1,2),(1,1,2,1),(1,2,0,2),\\
&\hspace{5.5 cm}(1,2,1,1),(1,2,2,0)\},\\
&\mathcal{C}'_{2}(2\mathbb{D}(U_{2,4}))=\{(0,1,2,2),(1,1,1,2),(1,1,2,1),(2,1,0,2),\\
&\hspace{5.5 cm}(2,1,1,1),(2,1,2,0)\},\\
&\mathcal{C}'_{3}(2\mathbb{D}(U_{2,4}))=\{(0,2,1,2),(1,1,1,2),(1,2,1,1),(2,0,1,2),\\
&\hspace{5.5cm}(2,1,1,1),(2,2,1,0)\},\\
&\mathcal{C}'_{4}(2\mathbb{D}(U_{2,4}))=\{(0,2,2,1),(1,1,2,1),(1,2,1,1),(2,0,2,1),\\
&\hspace{5.5cm}(2,1,1,1),(2,2,1,0)\}.
\end{align*}} 

The set of reduced $i$-unit minimal excluded vectors, $i \in \lceil 4 \rfloor$ is given by,
{
\begin{align*}
&\mathcal{C}_{1}(2\mathbb{D}(U_{2,4}))=\{(1,0,2,2),(1,2,0,2),(1,2,2,0)\},\\
&\mathcal{C}_{2}(2\mathbb{D}(U_{2,4}))=\{(0,1,2,2),(2,1,0,2),(2,1,2,0)\},\\
&\mathcal{C}_{3}(2\mathbb{D}(U_{2,4}))=\{(0,2,1,2),(2,0,1,2),(2,2,1,0)\},\\
&\mathcal{C}_{4}(2\mathbb{D}(U_{2,4}))=\{(0,2,2,1),(2,0,2,1),(2,2,0,1)\}.
\end{align*}} 
%\endproof
\end{example}

Now we proceed to give the construction algorithm.\\

\underline{\textbf{{ALGORITHM 1}}}\\
%\begin{description}
\underline{\emph{Step 1:}}
 Choose a basis vector $b \in \mathcal{B}(\mathbb{D})$ given by $\sum_{i \in (b)_{>0}} k_i \epsilon_{i,r}$ which satisfies the condition that $\rho(\{i\})=k_i, \forall i \in (b)_{>0}.$ For every $i \in (b)_{>0},$ add a node $i$ to the network with an input edge $e_i$ which generates the message $x_i.$ Let $f(e_i)=i.$ Define $M=T=(b)_{>0}.$\\
\underline{\emph{Step 2:}}
For $i \in \lceil r \rfloor \notin T,$ find a vector $u \in \mathcal{C}_i(\mathbb{D}),$ for which \mbox{$(u-\epsilon_{i,r})_{>0} \subseteq T.$} Add a new node $i'$ to the network with incoming edges from all the nodes which belong to $(u-\epsilon_{i,r})_{>0}.$ Also, add a node $i$ with a single incoming edge from $i',$ denoted as $e_{i',i}.$ Define $f(e)=head(e), \forall e \in In(i)$ and $f(e_{i',i})=i.$ Let $T \leftarrow T \cup \lbrace i\rbrace.$
Repeat step 2 until it is no longer possible.\\
\underline{\emph{Step 3:}} 
For $i \in M,$ choose a vector $u$ from $\mathcal{C}_i(\mathbb{D})$ for which $(u)_{>0} \subseteq T.$ Add a new node $h$ to the network which demands message $x_i$ and which has connections from the nodes in $(u-\epsilon_{i,r})_{>0}.$ Define $f(e)=head(e), \forall e \in In(h).$  Repeat this step as many number of times as desired.\\
\underline{\emph{Step 4:}} 
For a basis vector $b \in \mathcal{B}(\mathbb{D}),$ add a node $j$ which has incoming edges from the nodes which belong to $(b)_{>0}$ and demands all the messages. Define $f(e)=head(e), \forall e \in In(j).$ Repeat this step as many number of times as desired.  
%
%The network constructed using ALGORITHM 1, is discrete polymatroidal with respect to $\mathbb{D}$ with the network discrete polymatroid mapping $f$ defined in the algorithm. Hence, if $\mathbb{D}$ is representable over $\mathbb{F}_q,$ then the constructed network admits a vector linear solution over $\mathbb{F}_q,$ as shown in the following theorem.

For a discrete polymatroid $\mathbb{D},$ let $\rho_{max}(\mathbb{D})=\max_{i \in \lceil r \rfloor} \rho(\lbrace i \rbrace).$ 

Theorem \ref{thm4} below establishes the connection between the network constructed using ALGORITHM 1 and the discrete polymatroid from which the network was constructed, for a discrete polymatroid representable over $\mathbb{F}_q.$

\begin{theorem}
\label{thm4}
A network constructed using ALGORITHM 1 from a discrete polymatroid $\mathbb{D}$  which is representable over $\mathbb{F}_q,$ with the basis vector $b$ given by $\sum_{i \in (b)_{>0}} k_i \epsilon_{i,r}$ chosen in Step 1,  admits a linear $(k_1,k_2,\dotso,k_m;n)$-FNC solution over $\mathbb{F}_q,$ where $n=\rho_{max}(\mathbb{D}).$
\begin{proof}
The proof of the theorem is given by showing that the constructed network is $(k_1,k_2,\dotso,k_m;n)$-discrete polymatroidal with respect to the representable discrete polymatroid $\mathbb{D}$ from which it is constructed. The satisfaction of (DN1) is ensured by step 1 of the construction procedure. Since the vector $b=\sum_{i\in \mathcal{S}} k_i \epsilon_{i,r}$ belongs to $\mathcal{B}(\mathbb{D}),$ it belongs to $\mathbb{D}$ as well and hence (DN2) is satisfied. Also, since $\rho(\{i\})=k_i, \forall i \in (b)_{>0}$ and $n=\rho_{max}(\mathbb{D}),$ (DN3) is satisfied.

The nodes in the network constructed using ALGORITHM 1 are of five kinds (i) node $i,$ $i \in M,$ which are added in Step 1, (ii) node $i', i \in \lceil r \rfloor \setminus M,$ added in Step 2, (iii) node $i, i \in \lceil r \rfloor \setminus M,$ added in Step 2, (iv) nodes added in Step 3 which demand messages and (v) nodes added in Step 4 which demand messages. For a node $x$ of kind (i) or of kind (iii), since the in-degree is one and all the outgoing edges are mapped by $f$ to the same element in $\lceil r \rfloor,$ $f(In(x))=f(In(x) \cup Out(x))$ and hence $\rho(f(In(x)))=\rho(f(In(x) \cup Out(x))).$ Hence (DN4) is satisfied for nodes of kind (i) and (iii).

Consider a node $i' \in \lceil r \rfloor$ of kind (ii). Let $e_{i',i}$ denote the edge connecting $i'$ and $i.$ Let $u^i \in \mathcal{C}_i(\mathbb{D})$ denote the vector which was used in Step 2 while adding the node $i$ and $i'$ to the network. Since $f(e_{i',i})=i,$ we need to show that $\rho(f(In(i')))= \rho(f(In(i'))\cup \{i\}).$ Since $f(In(i'))=(u^i-\epsilon_{i,r})_{>0}$ and $(u^i-\epsilon_{i,r})_{>0} \cup \{i\}=(u^i)_{>0},$ it needs to be shown that $\rho\left(\left(u^i-\epsilon_{i,r}\right)_{>0}\right)=\rho\left((u^i)_{>0}\right),$ i.e., $dim\left(\sum_{j \in (u^i)_{>0}}V_j\right)=dim\left(\sum_{j \in \left(u^i-\epsilon_{i,r}\right)_{>0}}V_j\right).$ Let $a^i=(u^i-\epsilon_{i,r}).$ Since $u^i$ is a minimal excluded vector, $a^i \in \mathbb{D}$ and hence for all $A \subseteq \lceil r \rfloor,$ we have, 
\begin{equation}
\label{eqn2_thm}
\vert a^i(A)\vert\leq dim\left(\sum_{j \in A} V_j \right).
\end{equation}
Since $u^i \notin \mathbb{D},$ we have,
\begin{equation}
\label{eqn3_thm}
dim \left(\sum_{j \in A'} V_j\right)<\vert u^i(A') \vert,
\end{equation}
for some $A' \subseteq \lceil r \rfloor.$ Clearly $A'$ should contain $i,$ otherwise $\vert a^i(A')\vert=\vert u^i(A') \vert$ and, \eqref{eqn2_thm} and \eqref{eqn3_thm} cannot be simultaneously satisfied. Since $A'$ contains $i,$ we have $\vert u^i(A') \vert=\vert a^i(A')\vert+1.$ Hence, from \eqref{eqn2_thm} and \eqref{eqn3_thm} we get $dim\left(\sum_{j \in A'} V_j \right)= \vert a^i(A')\vert.$

\begin{figure}[htbp]
\centering
\includegraphics[totalheight=1.8 in,width=2in]{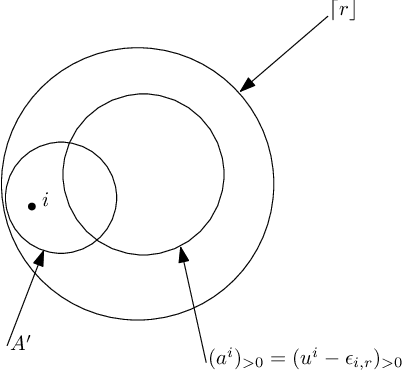}
\caption{Pictorial depiction of the sets $\lceil r \rfloor,$ $(a^i)_{>0}$ and $A'$ used in the proof of Theorem 2.}
\label{fig:figure_proof_theorem2}
\end{figure}
The sets $\lceil r \rfloor,$ $(a^i)_{>0}$ and the set $A'$ containing $i$ are pictorially depicted in Fig. \ref{fig:figure_proof_theorem2}. We have,
$\displaystyle{dim \left(\sum_{j \in (a^i)_{>0} \cap A'} V_j \right) \leq dim \left(\sum_{j \in A'} V_j \right)=\vert a^i(A')\vert.}$ Since \mbox{$a^i \in \mathbb{D},$} we have, $$\displaystyle{\vert a^i(A')\vert = \left\vert a^i\left(\left(a_i\right)_{>0} \cap A'\right)\right\vert \leq dim \left ( \sum_{j \in \left(a_i\right)_{>0}\cap A'} V_j \right).}$$ Hence, $dim \left ( \sum_{j \in \left(a_i\right)_{>0}\cap A'} V_j \right)=dim \left(\sum_{j \in A'} V_j \right).$ Since $i \in A',$ it follows that $$dim \left( \sum_{j \in \left(a_i\right)_{>0}\cap A'} V_j +V_i\right)=dim\left(\sum_{j \in \left(a_i\right)_{>0}\cap A'} V_j \right).$$ As a result, we have, 
{
\begin{align*}
&dim \left( \sum_{j \in \left(a_i\right)_{>0}\cap A'} V_j +V_i +\sum_{j \in \left(a_i\right)_{>0}\setminus A'} V_j\right)=\\
&\hspace{1.5cm}dim\left(\sum_{j \in \left(a_i\right)_{>0}\cap A'} V_j +\sum_{j \in \left(a_i\right)_{>0}\setminus A'} V_j\right),
\end{align*}} 

\noindent i.e., $dim\left(\sum_{j \in (u^i)_{>0}}V_j\right)=dim\left(\sum_{j \in \left(u^i-\epsilon_{i,r}\right)_{>0}}V_j\right).$

Following a procedure exactly similar to the one used for a node kind (ii), it can be shown that $\rho(f(In(x)))=\rho(f(In(x) \cup Out(x)))$ for a node $x$ of kind (iv).  

To show that (DN3) is satisfied for a node of kind (v), it needs to be shown that for $b \in \mathcal{B}(\mathbb{D}),$ $\rho\left(\left(b\right)_{>0} \cup f\left(\left\{i\right\}\right)\right)=\rho((b)_{>0}),$ $\forall i \in M.$ It needs to be shown that $dim(\sum_{j \in (b)_{>0}} V_j +V_{f(\{i\})})=dim(\sum_{j \in (b)_{>0}}V_j),$ $i \in M,$ which is true since $b$ is a basis vector for $\mathbb{D}.$  
This completes the proof of Theorem 2.
%Since $u^i-\epsilon_i \in \mathbb{D},$ we have $\vert  u^i-\epsilon_i \vert \leq \rho(f(In(i'))).$ But since the $i^{\text{th}}$ component of $u^i$ is 1 and $u^i \notin \mathbb{D}\mathbb{D},$ $\vert u^i \vert = \vert u^i-\epsilon_i \vert +1>\rho(f(In(i') \cup f(i)).$ Hence, $\rho(f(In(i') \cup f(i)) \leq \vert u^i-\epsilon_i \vert.$ As a result $\vert u^i-\epsilon_i \vert \geq \rho(f(In(i') \cup f(i)) \geq  \rho(f(In(i')) \geq \vert  u^i-\epsilon_i \vert.$ Hence, we have $\rho(f(In(i'))\cup f(e_{i',i}))=\rho(f(In(i'))).$
\end{proof}
\end{theorem}

The construction procedure provided in ALGORITHM 1 is illustrated using the following examples. 
\begin{example}
\label{ex20}
Continuing with Example \ref{ex19}, for simplicity, let $\mathbb{D}$ denote the discrete polymatroid $2\mathbb{D}(U_{2,4}).$
 The construction procedure for the discrete polymatroid considered in Example \ref{ex19} is summarized in Table \ref{table1}. The steps involved in the construction are illustrated in Fig. \ref{fig:network1}.
\end{example}
\begin{table}[h]
\centering
\includegraphics[totalheight=3 in,width=3.5 in]{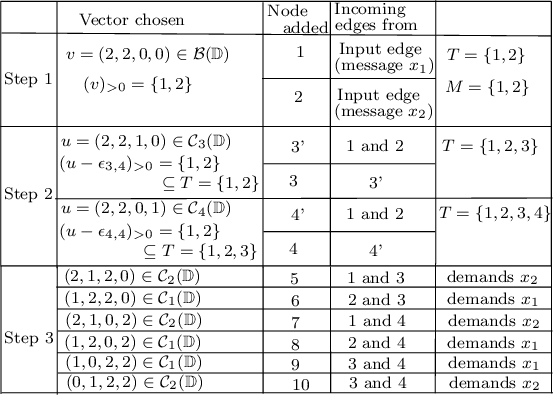}
\caption{Steps involved in the construction of a network from the discrete polymatroid in Example \ref{ex19}}
\label{table1}
\end{table} 

\begin{table*}[t]
\centering
\includegraphics[totalheight=6.25 in,width=5.55 in]{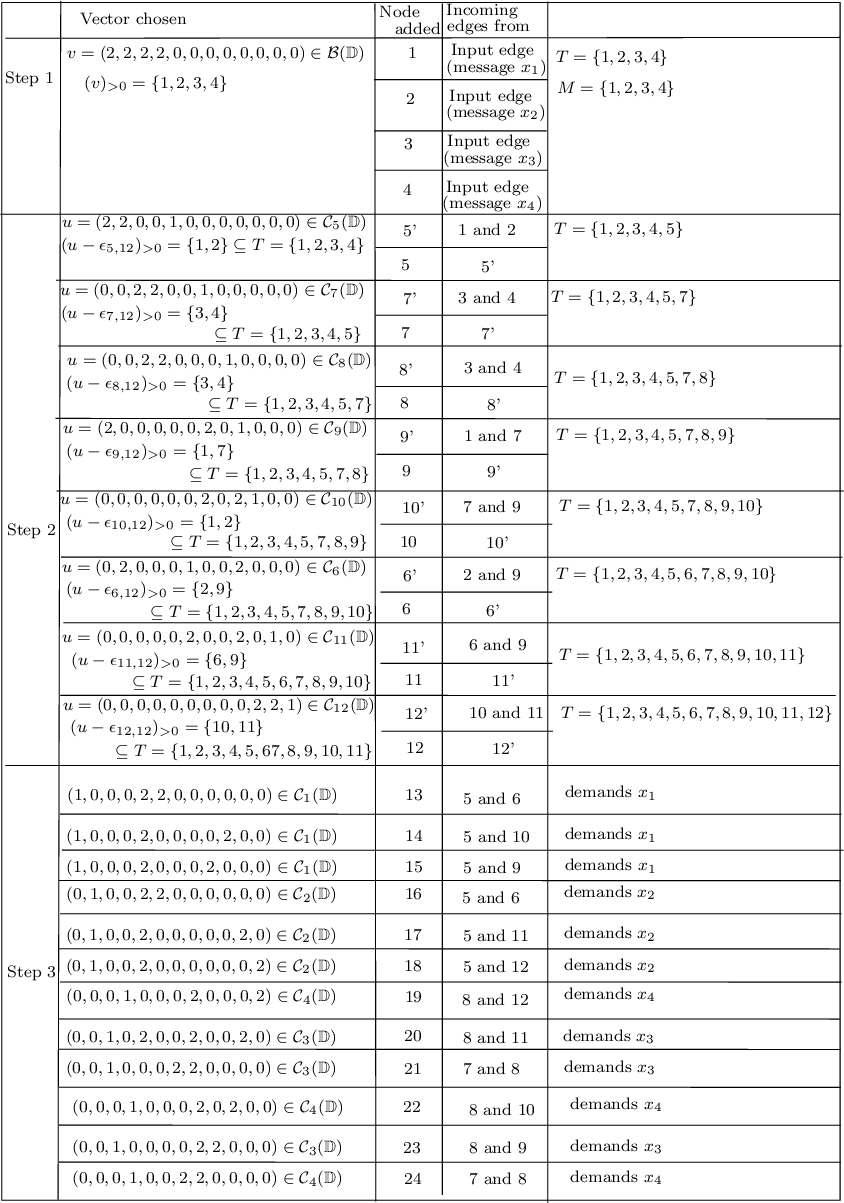}
\caption{Steps involved in the construction of the network in Example \ref{ex21}}
\label{table2}
\end{table*}
\begin{figure}[h]
\centering
\includegraphics[totalheight=3.75 in,width=3.5in]{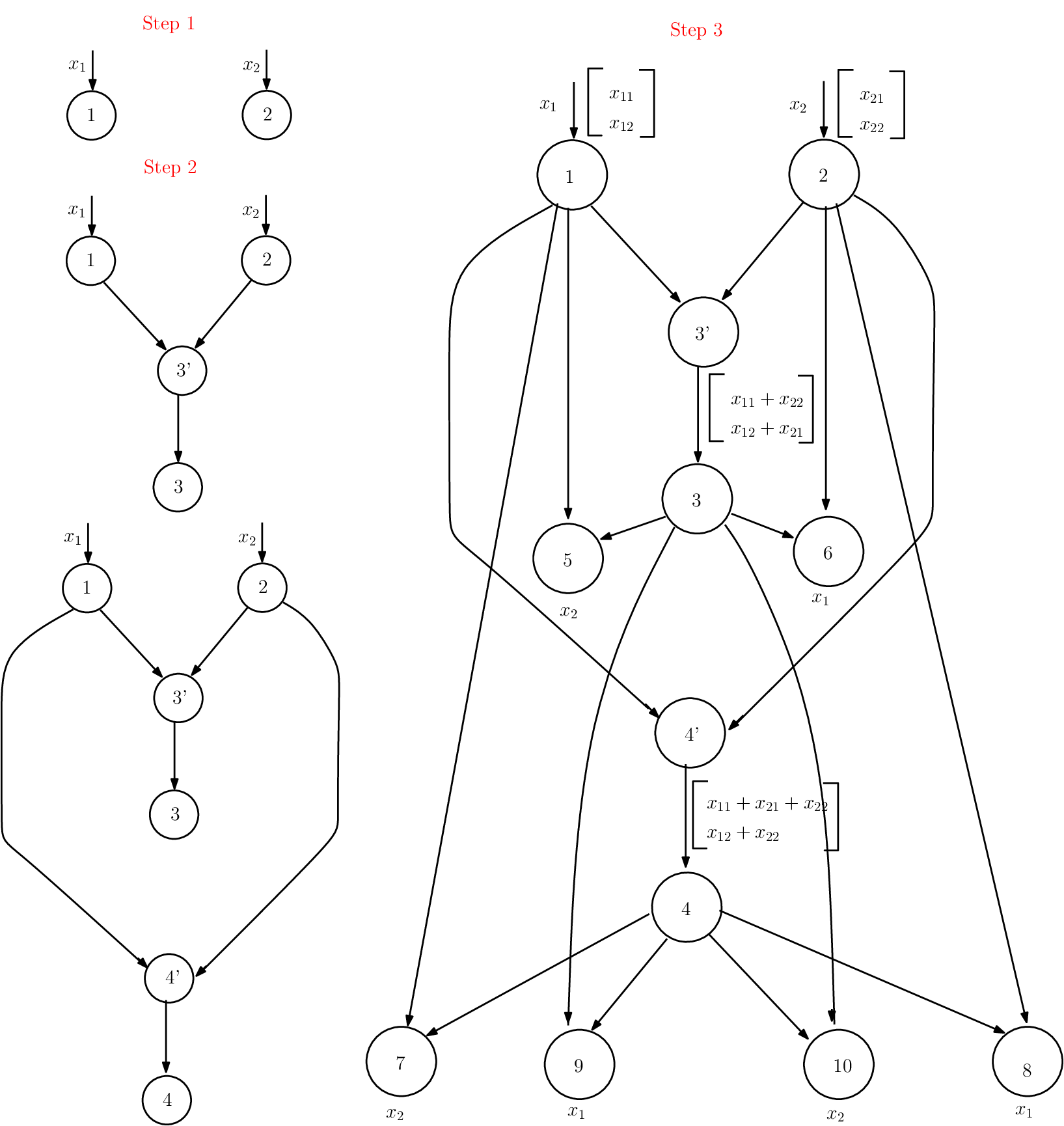}
\caption{The network constructed from the discrete polymatroid $2\mathbb{D}(U_{2,4})$}
\label{fig:network1}
\end{figure}

Let $A= \underbrace{\hspace{-.2 cm}\left[\begin{matrix} 1&0\\ 0&1\\0&0\\0&0 \end{matrix}\right.}_{A_1} \quad \underbrace{\begin{matrix} 0&0\\ 0&0\\1&0\\0&1 \end{matrix}}_{A_2} \quad \underbrace{\begin{matrix} 1&0\\ 0&1\\0&1\\1&0 \end{matrix}}_{A_3}  \quad \underbrace{\left.\begin{matrix} 1&0\\ 0&1\\1&0\\1&1 \end{matrix}\right]}_{A_4}$ be a matrix over $\mathbb{F}_2.$
Let $V_i$ denote the column span of $A_i, i \in \lceil 4 \rfloor.$ It can be verified that the vector subspaces $V_1,$ $V_2,$ $V_3$ and $V_4$ form a representation for $2\mathbb{D}(U_{2,4})$ over $\mathbb{F}_2.$
A vector linear solution of dimension 2 over $\mathbb{F}_2$ shown in Fig. \ref{fig:network1} is obtained by taking the global encoding matrices for the edges $3' \rightarrow 3$ and $4' \rightarrow 4$ to be the matrices $A_3$ and $A_4.$ All the outgoing edges of a node which has in-degree one carry the same vector as that of the incoming edge. The network in Fig. \ref{fig:network1} does not admit a scalar linear solution over $\mathbb{F}_2$ as shown in the following lemma.
\begin{lemma}
The network given in Fig. \ref{fig:network1} does not admit a scalar linear solution over $\mathbb{F}_2.$
\begin{proof}
Observe that node 5 demands $x_2$ and the only path from 2 to 5 is via the edge $3' \rightarrow 3.$ Also, node 6 demands $x_1$ and the only path from 1 to 6 is via the edge $3' \rightarrow 3.$ To satisfy these demands, the edge $3' \rightarrow 3$ needs to carry $x_1+x_2.$ By a similar reasoning, to satisfy the demands of nodes 7 and 8, the edge $4' \rightarrow 4$ needs to carry $x_1+x_2.$ But if the edges $3' \rightarrow 3$ and $4' \rightarrow 4$ carry $x_1+x_2,$ the demands of nodes 9 and 10 cannot be satisfied.
\end{proof}
\end{lemma}

While the network in Fig. \ref{fig:network1} does not admit a scalar linear solution over $\mathbb{F}_2,$ it has a scalar linear solution over all fields of size greater than two, as shown in the following lemma.
\begin{lemma}
The network in Fig. \ref{fig:network1} admits a scalar linear solution over all fields of size greater than two.
\begin{proof}
It can be verified that the network shown in Fig. \ref{fig:network1} is matroidal with respect to the uniform matroid $U_{2,4}$ with the mapping $f$ from the edge set to the ground set $\lceil 4 \rfloor$ of the matroid defined as follows: for $i \in \lceil 4 \rfloor,$ all the elements of $In(i')$ are mapped to $head(i'),$  the elements of $out(i)$ and the edge joining $i'$ and $i$ are mapped to $i.$ Since $U_{2,4}$ is representable over all fields of size greater than or equal to three (follows from Proposition 6.5.2, Page 203, \cite{Ox}), the network in Fig. \ref{fig:network1} admits a scalar linear solution over all fields of size greater than two.
\end{proof}
\end{lemma} 

The network constructed in the previous example turned out to be matroidal with respect to a matroid representable over all fields other than $\mathbb{F}_2$ and as a result it admitted scalar linear solutions over all $\mathbb{F}_q$ other than $\mathbb{F}_2.$ In the following example, the constructed network is discrete polymatroidal with respect to a representable discrete polymatroid whereas it cannot be matroidal with respect to any representable matroid. Hence it is not scalar linearly solvable over any field, but is vector linear solvable.

\begin{example}
\label{ex21}
Let $V_i, i \in \lceil 12 \rfloor,$ denote the column span of the matrix $A_i$ shown in \eqref{matrix_M_network_1}. Let $\mathbb{D}$ denote the discrete polymatroid $\mathbb{D}(V_1,V_2,\dotso,V_{12}).$
The steps involved in the construction of a network from this discrete polymatroid is summarized in Table \ref{table2}. 
The network thus constructed is shown in Fig. \ref{fig:network2}. The vector linear solution of dimension 2, which is in fact a vector routing solution, is obtained by choosing the global encoding matrix of the edge $i' \rightarrow i,i \in \lceil 12 \rfloor,$ to be $A_i$, as shown in Fig. \ref{fig:network2}. All the outgoing edges of a node which has in-degree one carry the same vector as that of the incoming edge.

 \begin{figure*}[t]
\centering
\includegraphics[totalheight=5 in,width=5in]{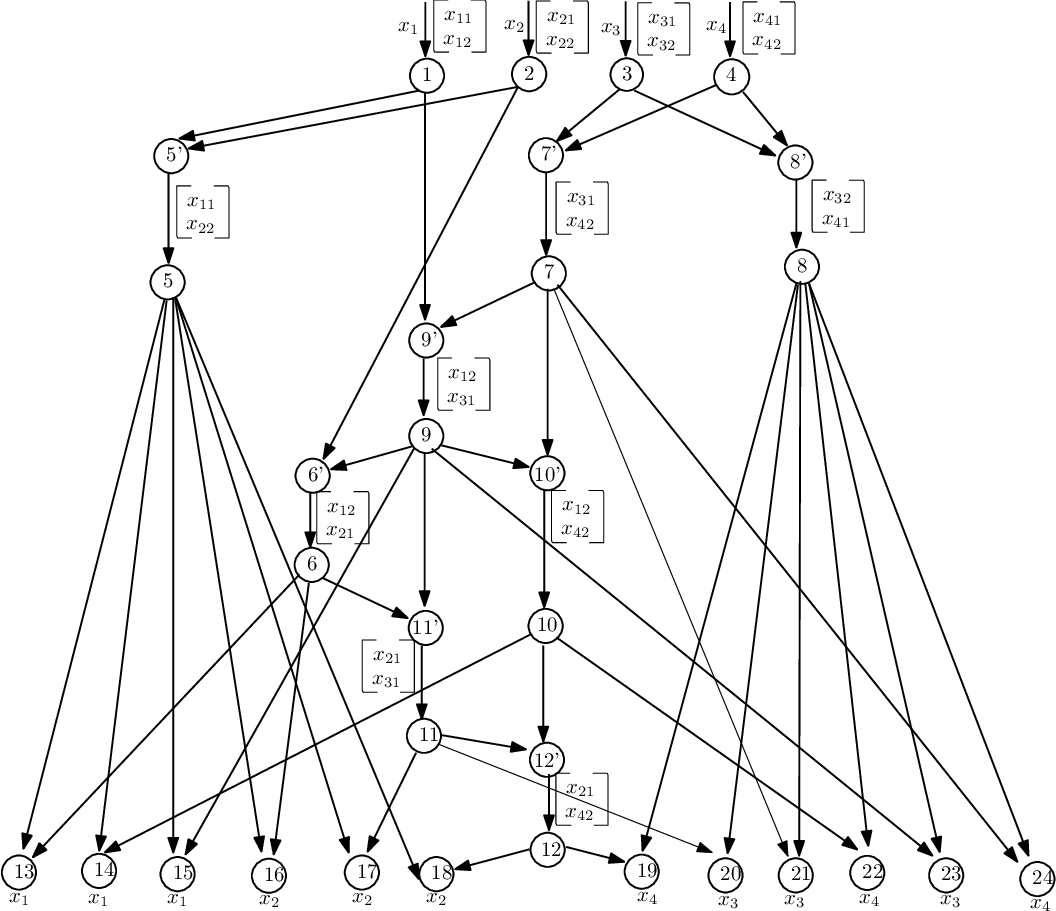}
\caption{A network which is vector linearly solvable but not scalar linearly solvable}
\label{fig:network2}
\end{figure*}
\end{example}

The following lemma shows that the network in Fig. \ref{fig:network2} is not scalar linearly solvable.
\begin{lemma}
The network in Fig. \ref{fig:network2} is not scalar linearly solvable.
\begin{proof}
To prove the lemma, it is shown that the network cannot be matroidal with respect to a representable matroid. The ideas used in the proof are similar to the ones used in the proof of Theorem V.8 in \cite{DoFrZe}.

On the contrary, assume that the network is matroidal with respect to a representable matroid $\mathbb{M}$ on the ground set $\lceil r \rfloor$ and let $f$ be the network-matroid mapping. Let the set of one dimensional vector spaces $V_i, i \in \lceil r \rfloor$ form a representation of $\mathbb{M}.$ All the outgoing edges of a node which has in-degree one carry the same vector as that of the incoming edge. For simplicity, let $i$ denote the incoming edge of node $i,$ where $i \in \lceil 12 \rfloor.$ Let $\rho$ denote the rank function of $\mathbb{D}(\mathbb{M}).$ Let $g(x)=\rho(f(x)), x \subseteq \lceil 12 \rfloor.$ 

We have $g(\{1,2\})\leq 2.$ From (DN2), it follows that $\sum_{i \in \lceil 4 \rfloor}\epsilon_{f(i),12} \in \mathbb{D}(\mathbb{M}).$ Hence we have $\sum_{i \in \lceil 2 \rfloor}\epsilon_{f(i),12} \in \mathbb{D}(\mathbb{M}),$ from which it follows that $2 \leq g(\{1,2\}).$ Hence, we have $g(\{1,2\})=2.$ Similarly, we also have $g(\{3,4\})=2.$

It is claimed that $g(\{5\})=1.$ Otherwise, $g(\{5\})$ has to be 0. In that case, since the nodes 13 and 16 demand $x_1$ and $x_2$ respectively, from (DN3) it follows that $dim(V_{f(1)}+V_{f(6)})=dim(V_{f(6)})$ and $dim(V_{f(2)}+V_{f(6)})=dim(V_{f(6)}).$ This will force $V_{f(1)}=V_{f(2)}$ which is not possible. Hence $g(\{5\})$ has to be 1. Similarly, it can be shown that $g(\{8\})=1.$

We have,
\begin{align}
\label{eqn_first}
g(\lbrace 3,8 \rbrace)+g(\{4,8\}) &\geq g(\lbrace 8 \}))+g(\{3,4,8\})\\
\label{eqn1}
 &\geq 1+ g(\{3,4\})=3,
\end{align}
where \eqref{eqn_first} holds since $g(\{3,4,8\})=g(\{4,8 \})$ (follows from (DN3)) and \eqref{eqn1} follows from the facts that $g(\{8\})=1$ and $g(\{3,4\})=2.$
Similarly, it can be shown that  
\begin{align}
\label{eqn4}
g(\lbrace 1,5 \rbrace)+g(\{2,5\})\geq 3.
\end{align}

Also, we have,
\begin{align}
\label{eqn_third}
g(\{2,5\})+g(\{3,8\})&=g(\{2,5,3,8\})\\
\nonumber
&\leq g(\{2,5,3,8,11 \})\\
\label{eqn_fourth}
&\hspace{-1 cm}\leq g(\{2,5,11\})+g(\{3,8,11\})-g(\{11\})\\
\label{eqn2}
&\hspace{-1 cm}=g(\{5,11\})+g(\{8,11\})-1\leq 3,
\end{align}
where \eqref{eqn_third} follows from the fact that $$dim\left(V_{f(2)}+V_{f(5)}\right)+dim\left(V_{f(3)}+V_{f(8)}\right)=$$ $$\hspace{2 cm}dim\left(V_{f(2)}+V_{f(5)}+V_{f(3)}+V_{f(8)}\right),$$ which in turn follows from the facts that $dim((V_{f(1)}+V_{f(2)}) \cap (V_{f(3)}+V_{f(4)}))=0$ and $V_{f(5)}$ and $V_{f(8)}$ are respectively vector subspaces of $V_{f(1)}+V_{f(2)}$ and $V_{f(3)}+V_{f(4)}.$ Equation \eqref{eqn_fourth} follows from (D2).
Similarly, it can be shown that 
\begin{align}
\label{eqn3}
g(\{2,5\})+g(\{4,8\})\leq 3.
\end{align}
From \eqref{eqn1}, \eqref{eqn2} and \eqref{eqn3}, we get $g(\{2,5\})\leq 1.5.$ Similarly, it can be shown that $g(\{1,5\}) \leq 1.5.$ Hence, from \eqref{eqn4}, we get $g(\{1,5\})=g(\{2,5\})=1.5$ which is not an integer, resulting in a contradiction. Hence, the network in Fig. \ref{fig:network2} cannot be matroidal with respect to any representable matroid.
\end{proof}
\end{lemma}
%\begin{figure}[t]
%\centering
%\includegraphics[totalheight=4 in,width=2.5in]{fnc_example.eps}
%\caption{A network for which scalar and vector solutions do not exist but an FNC solution exists}
%\label{fnc_example}
%\end{figure}
In the following two examples, we provide examples of networks constructed using Algorithm 1 which admit FNC solutions.
\begin{example}
\label{ex22}
For the discrete polymatroid considered in Example \ref{ex9}, the set of reduced $i$-unit minimal excluded vectors \mbox{$\mathcal{C}_i(\mathbb{D}),$} \mbox{$i \in \lceil 4 \rfloor,$} is given by \mbox{$\mathcal{C}_1(\mathbb{D})=\lbrace (1,0,0,2)\rbrace,$}
\mbox{$\mathcal{C}_2(\mathbb{D})=\lbrace (0,1,1,2)\rbrace,$}
\mbox{$\mathcal{C}_3(\mathbb{D})=\lbrace (0,1,1,2)\rbrace$} and
\mbox{$\mathcal{C}_4(\mathbb{D})=\lbrace (1,1,1,1)\rbrace.$} The construction procedure for the discrete polymatroid considered in Example \ref{ex9} is summarized in Table \ref{table3}. The different steps involved in the construction are depicted in Fig. \ref{example1_construction}. Since, in Step 1, the basis vector $b=(1,1,1,0)$ is used and $\rho_{max}(\mathbb{D})=\rho(\{4\})=2,$ the constructed network admits a linear $(1,1,1;2)$-FNC solution. The linear $(1,1,1;2)$-FNC solution shown in Fig. \ref{example1_construction} is obtained by taking the global encoding matrix of the edge joining 4' and 4 to be the matrix $A_4$ given in Example \ref{ex9}. 
\begin{table}[h]
\centering
\includegraphics[totalheight=2.25 in,width=3.5 in]{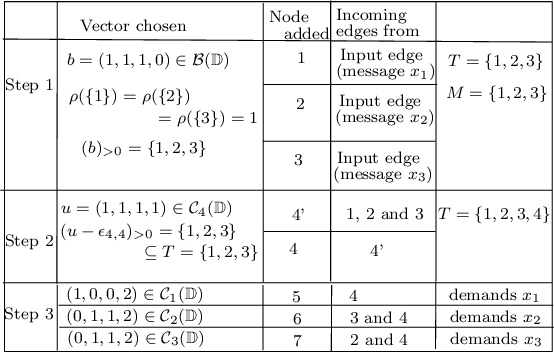}
\caption{Steps involved in the construction of a network from the discrete polymatroid in Example \ref{ex9}}
\label{table3}
\end{table} 
\begin{figure}[h]
\centering
\includegraphics[totalheight=3 in,width=3.5 in]{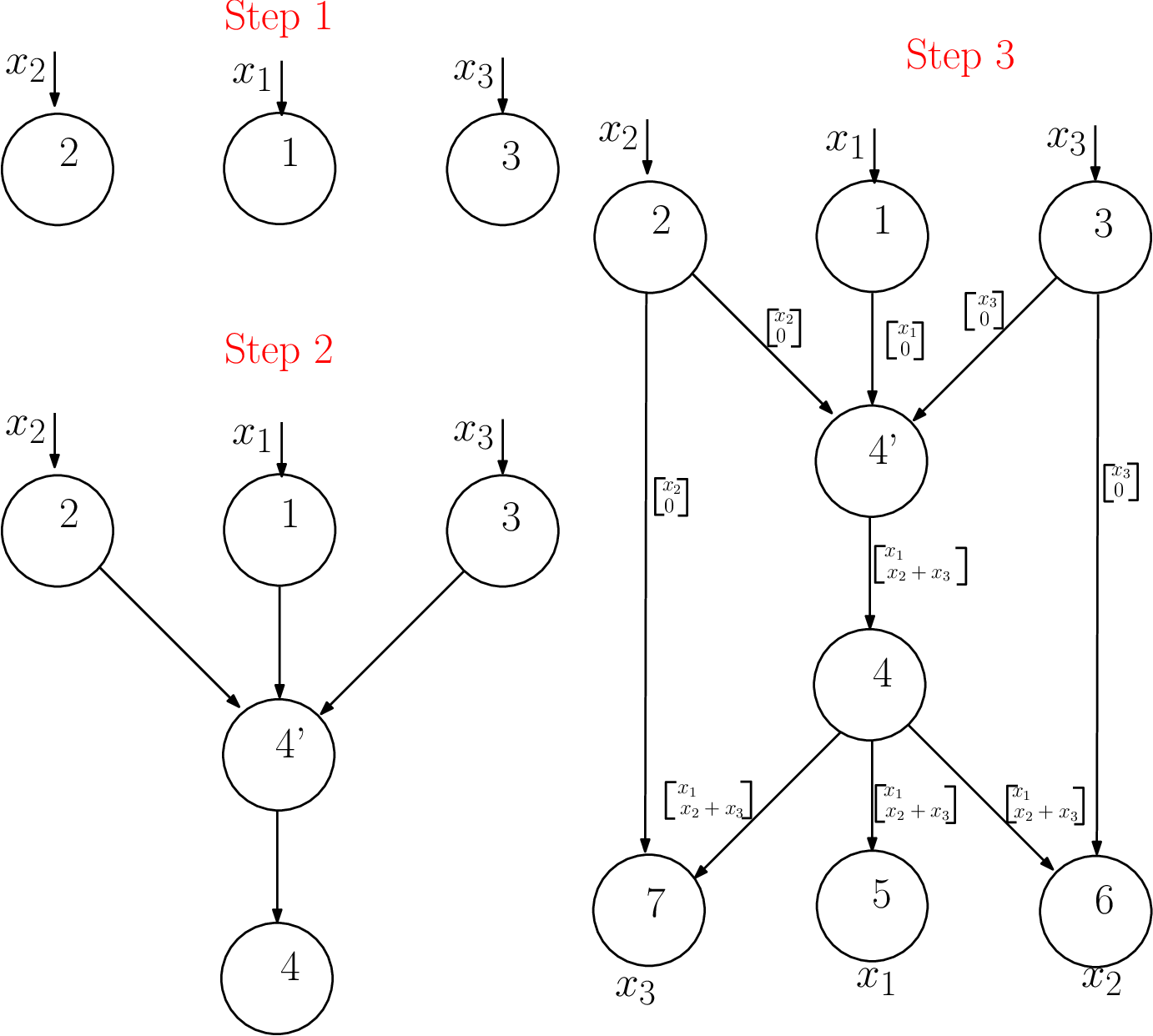}
\caption{Diagram showing the steps involved in the construction of a network from the discrete polymatroid in Example \ref{ex9}}
\label{example1_construction}
%%\label{fig:network1_table}
\end{figure} 
\end{example}

The network shown in Fig. \ref{example1_construction} has the properties listed in the following lemma.

\begin{lemma}
The network shown in Fig. \ref{example1_construction} has the following properties:
%The network shown in Fig. \ref{fnc_example}
\begin{enumerate}
\item
The network shown in Fig. \ref{example1_construction} does not admit any scalar or vector solution.
\item
For the network in Fig. \ref{example1_construction}, there does not exist an achievable rate tuple $(k_1/n,k_2/n,k_3/n)$ for which $(k_1/n,k_2/n,k_3/n)>(1/2,1/2,1/2).$ Note that the rate tuple achieved by the (2,1,1;2)-FNC solution provided in Fig. \ref{example1_construction} is $(1/2,1/2,1/2).$
\item
The uniform coding capacity of the network shown in Fig. \ref{example1_construction} is equal to 1/2. Hence, the (1,1,1;2)-FNC solution provided in Fig. \ref{example1_construction}, which is a uniform FNC solution, achieves the uniform coding capacity.
\end{enumerate}
\end{lemma}
\begin{proof}
1) To satisfy the demand of node 5, the edge from 4' to 4 has to carry $x_1,$ which would mean that the demands of the nodes 6 and 7 cannot be met. Hence, the network shown in Fig. \ref{example1_construction} does not admit any scalar and vector solution. 

2) Consider a $(k_1,k_2,k_3;n)$-FNC solution for which $k_i/n\geq 1/2, \forall i \in \lceil 3 \rfloor.$ To satisfy the demand at node 5, $k_1$ out of $n$ dimensions of the edge joining 4' and 4 should carry $x_1.$ Hence, to satisfy the demands of node 6 and 7, the conditions $(n-k_1)\geq k_2$ and $(n-k_1)\geq k_3$ needs to be satisfied. Since, $k_i/n\geq 1/2, \forall i \in \lceil 3 \rfloor,$ we have $k_1+k_2=k_1+k_3=n,$ from which it follows that $k_i/n=1/2, \forall i \in \lceil 3 \rfloor.$

3) Every $(k,k,k;n)$-FNC solution for this network should satisfy the condition that $\frac{k}{n}\leq \frac{1}{2}.$ The reason for this is as follows: $k$ out of $n$ dimensions of the vector flowing in the edge joining 4' and 4 should carry $x_1$ to satisfy the demand of node 5. The demands of node 6 and node 7 should be met by what is carried in the remaining $n-k$ dimensions. Hence, $n-k$ should be at least $k$ to be able to satisfy the demands of nodes 6 and 7. 
\end{proof} 
\end{example}

%In Example \ref{ex22}, a uniform FNC solution was provided. 
In Example \ref{ex22}, a uniform FNC solution was provided. 
In Example \ref{ex23}, we provide a network with a non-uniform FNC solution and for which the average rate achieved by the FNC solution provided is greater than the uniform coding capacity.
%\end{example}
\begin{example}
\label{ex23}
Let $A= \underbrace{\hspace{-.0 cm}\left[\begin{matrix} 1 & 0\\ 0 & 1\\0&0\\0&0 \end{matrix}\right.}_{A_1} \; \underbrace{\begin{matrix} 0\\ 0\\1\\0 \end{matrix}}_{A_2} \; \underbrace{\begin{matrix} 0\\ 0\\0\\1 \end{matrix}}_{A_3}  \; \underbrace{\left.\begin{matrix} 1&1\\ 1&0\\1&1\\1&0 \end{matrix}\right.}_{A_4}\; \underbrace{\left.\begin{matrix} 0&0\\ 0&1\\0&1\\1&0 \end{matrix}\right]}_{A_5}$ be a matrix over $\mathbb{F}_q.$ Let $V_i$ denote the column span of $A_i,$ \mbox{$i \in \lceil 5 \rfloor.$} The set of basis vectors for the discrete polymatroid $\mathbb{D}(V_1,V_2,V_3,V_4,V_5)$ is given by,
{
\begin{align*}
&\hspace{-.2 cm}\left\{(0,0,0,2,2),(0,0,1,2,1),(0,1,0,2,1),(0,1,1,1,1),\right.\\
&(0,1,1,2,0),(1,0,0,2,1),(1,0,1,1,1),(1,0,1,2,0),\\&(1,1,0,0,2),(1,1,0,1,1),(1,1,0,2,0),(1,1,1,0,1),\\&(1,1,1,1,0),(2,0,0,0,2),(2,0,0,1,1),(2,0,0,2,0),\\&(2,0,1,0,1),(2,0,1,1,0),(2,1,0,0,1),(2,1,0,1,0),\\
& \left.\hspace{5.75 cm}(2,1,1,0,0)\right\}.
\end{align*}
}
For this discrete polymatroid, it can be verified that the sets of reduced $i$-unit minimal excluded vectors $\mathcal{C}_i(\mathbb{D}),i \in \lceil 5 \rfloor$ are given by,
 \mbox{\small$\mathcal{C}_1(\mathbb{D})=\lbrace (1,0,0,2,2),(1,1,1,2,0)\rbrace,$}\\
\mbox{\small$\mathcal{C}_2(\mathbb{D})=\lbrace (0,1,0,2,2),(2,1,0,0,2),(2,1,0,2,0)\rbrace,$}
\mbox{\small$\mathcal{C}_3(\mathbb{D})=\lbrace (2,0,1,2,0),(0,0,1,0,2)\rbrace,$}\\
\mbox{\small$\mathcal{C}_4(\mathbb{D})=\lbrace (0,0,1,1,2),(2,1,1,1,0)\rbrace$ and}
\mbox{\small$\mathcal{C}_5(\mathbb{D})=\lbrace (0,0,1,2,1),(2,0,0,2,1),(2,1,1,0,1)\rbrace.$} \\The construction procedure for the discrete polymatroid considered is summarized in Table \ref{table4}. The different steps involved in the construction are depicted in Fig. \ref{example2_construction}. Since, in Step 1, the basis vector $b=(2,1,1,0,0)$ is used and $\rho_{max}(\mathbb{D})=\rho(\{5\})=2,$ the constructed network admits a linear $(2,1,1;2)$-FNC solution. The linear $(2,1,1;2)$-FNC solution shown in Fig. \ref{example2_construction} is obtained by taking the global encoding matrix of the edge joining 4' and 4 to be the matrix $A_4$ and that of the edge joining 5' and 5 to be the matrix $A_5.$ 
\begin{table}[h]
\centering
\includegraphics[totalheight=2.5 in,width=3.5 in]{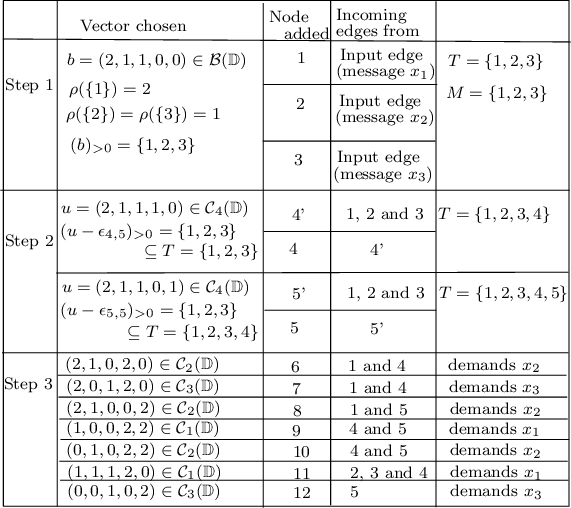}
\caption{Steps involved in the construction of the network in Example \ref{ex23}}
\label{table4}
\end{table} 
\begin{figure}[h]
\centering
\includegraphics[totalheight=3.5 in,width=3.5 in]{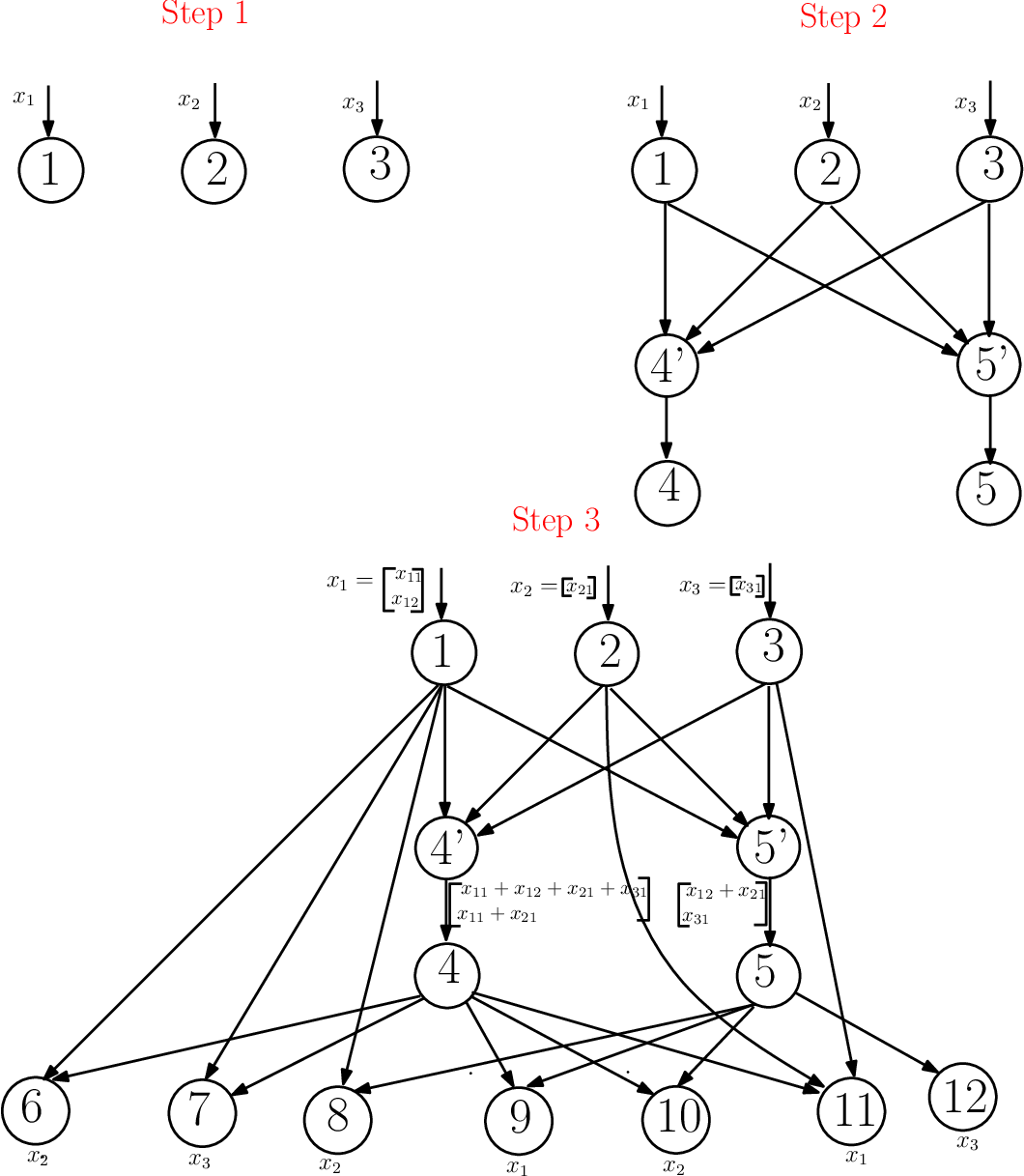}
\caption{Diagram showing the steps involved in the construction of the network in Example \ref{ex23}}
\label{example2_construction}
\end{figure} 
\end{example}

%In Example \ref{ex23}, we provide a network with a non-uniform FNC solution and for which the average rate achieved by the FNC solution provided is greater than the uniform coding capacity.
%\begin{example}
%\label{ex23}
%\begin{figure}[t]
%\centering
%\includegraphics[totalheight=3 in,width=4in]{mfnc_example.eps}
%\caption{A network for which scalar and vector solutions do not exist but an FNC solution exists}
%\label{mfnc_example}
%\end{figure}
%Consider the network given in Fig. \ref{mfnc_example}. A linear (2,1,1;2)-FNC solution for this network is shown in Fig. \ref{mfnc_example}. All the outgoing edges of a node which has only one incoming edge, are assumed to carry the same vector as that of the incoming edge. Consider the discrete polymatroid $\mathbb{D}(V_1,V_2,V_3,V_4,V_5)$ defined in Example \ref{ex10}. It can be verified that the network shown in Fig. \ref{mfnc_example} is (2,1,1;2)-discrete polymatroidal with respect to the discrete polymatroid $\mathbb{D}(V_1,V_2,V_3,V_4,V_5)$ with the network-discrete polymatroid mapping $f$ defined as follows: all the incoming and outgoing edges of node $i, i \in \{1,2,3,4,5\}$ are mapped on to the ground set element $i$ of the discrete polymatroid $\mathbb{D}(V_1,V_2,V_3,V_4,V_5).$

Lemma \ref{lemma5} below lists some of the properties of the network given in Fig. \ref{example2_construction}.

\begin{lemma}
\label{lemma5}
The network given in Fig. \ref{example2_construction} has the following properties:
\begin{enumerate}
\item
The network in Fig. \ref{example2_construction} does not admit any scalar or vector solution.

\item
The average coding capacity of the network in Fig. \ref{example2_construction} is 2/3. Hence, the solution provided in Fig. \ref{example2_construction} achieves the average coding capacity.
\item
The uniform coding capacity of the network in Fig. \ref{example2_construction} is 1/2. Hence the (2,1,1;2)-FNC solution provided in Fig. \ref{example2_construction} achieves an average rate of 2/3 which is strictly greater than the maximum average rate of 1/2 achievable using uniform FNC. 
\end{enumerate}
\begin{proof}
 1) To deliver message $x_3$ to node 12, the edge connecting nodes 5' and 5 needs to carry $x_3.$ In that case, message $x_2$ cannot be delivered to node 8, since the only path from node 2 which generates $x_2$ to node 8 contains the edge joining 5' and 5. Hence, the network in Fig. \ref{example2_construction} does not admit any scalar or vector solution.
  
2) To prove that the average coding capacity is 2/3, it needs to be shown that for all $(k_1,k_2,k_3;n)$-FNC solutions, $k_1+k_2+k_3 \leq 2n.$ First note that $k_1 \leq n.$ This follows from the fact that node 11 demands $x_1$ and there is only one path connecting the nodes 1 and 11. Hence, it can be assumed that the edges $1 \rightarrow 6$ and $1 \rightarrow 7$ carry $x_1.$ Since the nodes 6 and 7 demand $x_2$ and $x_3$ respectively, given the vectors carried in the edges $1 \rightarrow 7$ and $4' \rightarrow 4,$ one must be able to determine $x_1,x_2$ and $x_3.$ Hence, $k_1+k_2+k_3\leq 2n.$ Hence the average coding capacity is upper bounded by $2/3.$ Since the solution provided Fig. \ref{example2_construction} has an average achievable rate of 2/3, the average coding capacity is 2/3.
 
 3) For any $(k,k,k;n)$-FNC solution, $\frac{k}{n}$ cannot exceed $\frac{1}{2}.$ The reason is as follows: $k$ dimensions of the vector transmitted from 5' to 5 should carry $x_3$ and to ensure that node 8 gets $x_2,$ $n-k$ should be at least $k,$ i.e, $\frac{k}{n}\leq \frac{1}{2}.$ A uniform rate of $k/n=1/2$ can be achieved by choosing $x_{12}$ to be always zero in the FNC solution provided in Fig. \ref{example2_construction}.%The linear $(2,1,1;2)$-FNC solution shown in Fig. \ref{fnc_example} achieves an average rate of 2/3, which is greater than the maximum average rate of 1/2 achievable using a uniform FNC solution.  
\end{proof}
\end{lemma}
\section{Linear Index Coding and Discrete Polymatroid Representation}
In this section, we explore the connections between linear index coding and representable discrete polymatroids. In Section \ref{VI A}, it is shown that existence of a linear solution for an index coding problem is connected to the existence of a representable discrete polymatroid which satisfies certain conditions determined by the index coding problem considered. In Section \ref{VI B}, a construction of an index coding problem from a discrete polymatroid $\mathbb{D}$ is provided, which is a generalization of the construction from matroids provided in \cite{RoSpGe}. The constructed index coding problem admits a perfect linear index coding solution of dimension $n$ if and only if the discrete polymatroid $n\mathbb{D}$ is representable.
\subsection{Linear Index Coding and Discrete Polymatroid Representation}
\label{VI A}
The following theorem gives the necessary and sufficient condition in terms of discrete polymatroid representation, for the existence of a linear index code of length $c$ and dimension $n$ for an index coding problem $\mathcal{I}(X,\mathcal{R}).$
\begin{theorem}
\label{thm5}
A vector linear index code over $\mathbb{F}_q$ of length $c$ and dimension $n$ exists for an index coding problem $\mathcal{I}(X,\mathcal{R}),$ if and only if there exists a discrete polymatroid $\mathbb{D}$ representable over $\mathbb{F}_q$ on the ground set $\lceil m+1 \rfloor$ satisfying the following conditions:
\begin{itemize}
\item[(C1):]
$\rho(\{i\})=n, \forall i \in \lceil m \rfloor,$ $\rho(\lceil m \rfloor)=nm,$ $\rho(\{m+1\})=c$ and $rank(\mathbb{D})=nm.$
\item[(C2):]
$\forall (x_i,H) \in \mathcal{R},$ where $H=\{x_{j_1},x_{j_2},\dotso,x_{j_l}\},$\\ {\footnotesize$\rho(\{i\} \cup \{j_1,j_2,\dotso j_l\} \cup \{m+1\})=\rho(\{j_1,j_2,\dotso,j_l\}\cup \{m+1\}).$}
\end{itemize}
\begin{proof}
     
To prove the `if' part, assume that there exists a discrete polymatroid $\mathbb{D}$ of rank $nm$ representable over $\mathbb{F}_q$ which satisfies (C1) and (C2). Let $V_1,V_2,\dotso,V_m,V_{m+1}$ denote the vector subspaces over $\mathbb{F}_q$ which form a representation for $\mathbb{D}.$ From (C1), it follows that the vector subspaces $V_i, i \in \lceil m \rfloor,$ can be written as the column span of $mn \times n$ matrices $A_i$ over $\mathbb{F}_q,$ with $rank(A_i)=n$ and $rank([A_1 A_2 \dotso A_m])=mn.$ Also, the vector subspace $V_{m+1}$ can be written as the column span of a $mn \times c$ matrix $A_{m+1}$ over $\mathbb{F}_q$ which has a rank $c.$  Let $B$ denote the invertible $nm \times nm$ matrix given by $[A_1\; A_2\;\dotso\;A_m].$ Define $A'_i=B^{-1}A_i, i \in \lceil m+1\rfloor.$ The claim is that the map $f:\mathbb{F}^{nm}_q\rightarrow \mathbb{F}^{c}_{q}$ given by $f(y)=yA'_{m+1}$ forms an index code of length $c$ and dimension $n$ over $\mathbb{F}_q.$ Consider the receiver node $(x_i,H)\in \mathcal{R}$ where $H =\{x_{j_1},x_{j_2},\dotso,x_{j_l}\}.$ Let $b_H=[x_{j_1}\;x_{j_2}\;\dotso\;x_{j_l}].$ From (C2), it follows that the matrix $A_i$ can be written as $[A_{j_1}\;A_{j_2}\;\dotso  A_{j_l}\;A_{m+1}]M_i,$ where $M_i$ is of size $(n\vert H\vert+c) \times n.$ Hence, $A'_i=[A'_{j_1}\;A'_{j_2}\dotso A'_{j_l}\;A'_{m+1}]M_i.$ We have, $[b_H \; f(y)]=y [A'_{j_1}\;A'_{j_2}\dotso A'_{j_l}\;A'_{m+1}].$ The function $\psi_R$ defined as $\psi_R(H,f(y))=[b_H f(y)]M_i$ forms a valid decoding function at $R$ since $[b_H f(y)]M_i=y [A'_{j_1}\;A'_{j_2}\dotso A'_{j_l}\;A'_{m+1}]M_i=yA'_i=x_i.$

To prove the `only if' part, assume that there exists a vector linear index code $f$ of length $c$ and dimension $n$ for the index coding problem $\mathcal{I}(X,\mathcal{R}).$ Define $A_l$ to be the $nm \times n$ matrix with the $(i,j)^{th}$ entry being one for $i=(l-1)n+t, j= t,$ where $t \in \lceil n \rfloor$ and all other entries being zeros. The function $f$ can be written as $f(y)=y A_{m+1}$ where $ A_{m+1}$ is a matrix of size $nm \times c$ matrix over $\mathbb{F}_q.$ Define $V_i$ to be the column span of $A_i.$ It can be verified that the discrete polymatroid $\mathbb{D}(V_1,V_2,\dotso,V_{m+1})$ satisfies the condition (C1) and (C2).
\end{proof}
\end{theorem}
%\begin{example}
%Consider the index coding problem presented in \cite{RoSpGe}, with the message set $X=\{x_1,x_2,x_3\}$ and the set of receivers {\small$\mathcal{R}=\left\{\left(x_4,\left\{x_1\right\}\right),\left(x_3,\left\{x_2,x_4\right\}\right),\left(x_1,\left\{x_2,x_3\right\}\right),\left(x_2,\left\{x_1,x_3\right\}\right)\right\}.$} The function $f(X)=[x_1+x_2+x_3,x_1+x_4]$ forms a valid linear index code of length 2 and dimension 1 for this index coding problem over any field $\mathbb{F}_q.$ Let $A_i, i \in \lceil 4 \rfloor,$ denote the column vector of length 4 whose $i^{th}$ component  is 1 and all other components are zeros. Let $A_5=\begin{bmatrix}1 &1\\1 & 0\\1 & 0\\0 & 1\end{bmatrix}.$ Let $V_i, i \in \lceil 5 \rfloor$ denote the column span of $A_i$ over $\mathbb{F}_q.$ It can be verified that the discrete polymatroid $\mathbb{D}(V_1,V_2,V_3,V_4,V_5)$ on the ground set $\lceil 5 \rfloor$ has rank 4 and it satisfies the conditions (C1) and (C2) given in Theorem \ref{thm5}.   
%\end{example}
\subsection{Construction of an Index Coding Problem from a Discrete Polymatroid}
\label{VI B}
In \cite{RoSpGe}, a construction of an index coding problem $\mathcal{I}_{\mathbb{M}}(Z,\mathcal{R})$ from a matroid $\mathbb{M}$ was provided and it was shown that a perfect linear index coding solution of dimension $n$ over $\mathbb{F}_q$ exists for the index coding problem $\mathcal{I}_{\mathbb{M}}(Z,\mathcal{R})$ if and only if the matroid $\mathbb{M}$ has a multi-linear representation of dimension $n$ over $\mathbb{F}_q.$ This result implies a reduction from the problem of finding a multi-linear representation of dimension $n$ over $\mathbb{F}_q$ for a matroid $\mathbb{M}$ to the problem of finding a perfect linear solution of dimension $n$ over $\mathbb{F}_q$ for the index coding problem $\mathcal{I}_{\mathbb{M}}(Z,\mathcal{R}).$ 

In this subsection, we provide a construction of an index coding problem $\mathcal{I}_{\mathbb{D}}(Z,\mathcal{R})$ from a discrete polymatroid $\mathbb{D},$ which when specialized to the discrete polymatroid $\mathbb{D}(\mathbb{M}),$ $\mathbb{M}$ being a matroid, reduces to the construction given in \cite{RoSpGe}. We establish the connection between the existence of a perfect linear solution of dimension $n$ for $\mathcal{I}_{\mathbb{D}}(X,\mathcal{R})$ and the representability of the discrete polymatroid $n\mathbb{D}.$ Note that unlike the construction provided in Section \ref{IV B}, the construction provided in this subsection is applicable for any arbitrary discrete polymatroid.

The construction of the index coding problem $\mathcal{I}_{\mathbb{D}}(Z,\mathcal{R})$ from the discrete polymatroid $\mathbb{D}$ with  $rank(\mathbb{D})=k$ is provided below:
\begin{itemize}
\item[(i)]
The set of messages $Z=X \cup Y,$ where
\begin{itemize}
\item
 $X=\{x_1,x_2,\dotso,x_k\}$ and 
\item 
 $Y=\{y^1_1,y^2_1,\dotso,y^{\rho(\{1\})}_{1,}y^1_2,y^2_2,\dotso,y^{\rho(\{2\})}_{2,}$\\$~~~~~~~~~~~~~~~~~~~~~~~~~~~~~~\dotso,y^1_r,y^2_r,\dotso,y^{\rho(\{r\})}_r\}.$\\
 Let $\zeta_i=\{y^1_i,y^2_i,\dotso y^{\rho(\{i\})}_{i}\}.$
 \end{itemize}
 \item[(ii)]
 The set of receivers $\mathcal{R}=R_1 \cup R_2 \cup R_3,$ where the sets $R_1,R_2$ and $R_3$ are as defined below.
 \begin{itemize}
 \item[(a)]
 For a basis vector $b =\sum_{i \in \lceil r \rfloor}b_i\epsilon_{i,r} \in \mathcal{B}(\mathbb{D}),$ the set $S_1(b)$ is defined as 
{\small $$S_1(b)=\left\{\left(x_j,\bigcup_{l \in (b)_{>0}} \eta_l\right):j \in \lceil k \rfloor, \eta_l \subseteq \zeta_l, \vert \eta_l \vert =b_l\right\}.$$} 
Define 
$\displaystyle{R_1=\bigcup_{{b}\in \mathcal{B}(\mathbb{D})} S_1(b).}$ 
\item[(b)]
For a minimal excluded vector $c =\sum_{i \in \lceil r \rfloor}c_i\epsilon_{i,r} \in \mathcal{C}(\mathbb{D}),$ $j \in (c)_{>0},$ and $p \in \lceil \rho(\{j\})\rfloor,$ the set $S_2(c,j,p)$ is defined as 
{
\begin{align*}
&S_2(c,j,p)=\\
&\{(y_j^p,\Gamma_1 \cup \Gamma_2):\Gamma_1=\hspace{-.4 cm}\bigcup_{l \in (c)_{>0}\setminus \{j\}} \hspace{-.4 cm}\eta_l,\;\\& \eta_l \subseteq \zeta_l, \vert \eta_l \vert=c_l, \Gamma_2\subseteq \zeta_{j} \setminus \{y^p_j\}, \vert \Gamma_2\vert=c_j-1\}.
\end{align*}}
Define $\displaystyle{R_2=\bigcup_{c \in \mathcal{C}(\mathbb{D})}\bigcup_{j \in (c)_{>0}}\bigcup_{p \in \lceil \rho(\{j\})\rfloor}S_2(c,j,p)}.$
\item[(c)]
Define $R_3=\left\{\left(y_i^j,X\right):i \in \lceil r \rfloor, j \in \lceil \rho(\{i\})\rfloor\right\}.$ 
 \end{itemize}
 \end{itemize}

For the index coding problem $\mathcal{I}_{\mathbb{D}}(Z,\mathcal{R})$ defined above, $\mathcal{M}(\mathcal{I}_{\mathbb{D}}(Z,\mathcal{R}))=\sum_{i \in \lceil r \rfloor}\rho(\{i\}).$
 
 For a matroid $\mathbb{M},$ the index coding problem $\mathcal{I}_{\mathbb{D}(\mathbb{M})}(Z, \mathcal{R})$ reduces to the index coding problem $\mathcal{I}_{\mathbb{M}}(Z,\mathcal{R})$ provided in Section IV-B in \cite{RoSpGe}.
 
\begin{example}
\label{ex31}
Consider the discrete polymatroid provided in Example \ref{ex7}. We have $rank(\mathbb{D})=k=3.$ The index coding problem $\mathcal{I}_{\mathbb{D}}(Z,\mathcal{R})$ is as follows:
\begin{itemize}
\item[(i)]
The set of messages $Z=X \cup Y,$ where $X=\{x_1,x_2,x_3\}$ and $Y=\{y^1_1,y^2_1,y^1_2,y^2_2,y^1_3\}.$ 
\item[(ii)]
The set of receivers $\mathcal{R}=R_1 \cup R_2 \cup R_3$ where $R_1,R_2$ and $R_3$ are as given below:
\begin{itemize}
\item[(a)]
As mentioned in Example \ref{ex8}, the set of basis vectors\\
\mbox{$\mathcal{B}(\mathbb{D})=\{(1,1,1),(1,2,0),(2,0,1),(2,1,0))\}.$}
We have, 
{
\begin{align*}
&S_1\left(\left(1,1,1\right)\right)=\left\{\left(x_i,\left\{y_1^j,y_2^k,y_3^1\right\}\right):i\in \lceil 3 \rfloor\right.,\\
&\hspace{5.3 cm} \left. j,k \in \lceil 2 \rfloor\right\},\\
&S_1\left(\left(1,2,0\right)\right)=\left\{\left(x_i,\left\{y_1^j,y_2^1,y_2^2\right\}\right):i\in \lceil 3 \rfloor,\right.\\ 
&\left.\hspace{5.75 cm}j\in \lceil 2 \rfloor\right\},\\
 &S_1\left(\left(2,0,1\right)\right)=\left\{\left(x_i,\left\{y_1^1,y_1^2,y_3^1\right\}\right):i\in \lceil 3 \rfloor \right\},\\
 &S_1\left(\left(2,1,0\right)\right)=\left\{\left(x_i,\left\{y_1^1,y_1^2,y_2^j\right\}\right):i\in \lceil 3 \rfloor, \right.\\
 &\left.\hspace{5.75 cm}j\in \lceil 2 \rfloor\right\},\\
&R_1=S_1\left(\left(1,1,1\right)\right)\cup S_1\left(\left(1,2,0\right)\right) \cup S_1\left(\left(2,0,1\right)\right) \\
&\hspace{5.3cm}\cup S_1\left(\left(2,1,0\right)\right).
\end{align*}}
\item[(b)]
From Example \ref{ex19}, it follows that the minimal excluded vectors for $\mathbb{D}$ are \mbox{$c_1=(0,2,1),$} \mbox{$c_2=(2,1,1)$} and \mbox{$c_3=(2,2,0).$} We have,
{
\begin{align*}
&S_2(c_1,2,1)=\{(y^1_2,\{y^2_2,y^1_3\})\},\\
&S_2(c_1,2,2)=\{(y^2_2,\{y^1_2,y^1_3\})\},\\
&S_2(c_1,3,1)=\{(y^1_3,\{y^1_2,y^2_2\})\},\\
&S_2(c_2,1,1)=\{(y^1_1,\{y^2_1,y^i_2,y^1_3\}):i \in \lceil 2 \rfloor\},\\
&S_2(c_2,1,2)=\{(y^2_1,\{y^1_1,y^i_2,y^1_3\}):i \in \lceil 2 \rfloor\},\\
&S_2(c_2,2,1)=\{(y^1_2,\{y^1_1,y^2_1,y^1_3\})\},\\
&S_2(c_2,2,2)=\{(y^2_2,\{y^1_1,y^2_1,y^1_3\})\},\\
&S_2(c_2,3,1)=\{(y^1_3,\{y^1_1,y^2_1,y^i_2\}):i \in \lceil 2 \rfloor\},\\
&S_2(c_3,1,1)=\{(y^1_1,\{y^2_1,y^1_2,y^2_2\})\},\\
&S_2(c_3,1,2)=\{(y^2_1,\{y^1_1,y^1_2,y^2_2\})\},\\
&S_2(c_3,1,2)=\{(y^2_1,\{y^1_1,y^1_2,y^2_2\})\},\\
&S_2(c_3,2,1)=\{(y^1_2,\{y^1_1,y^2_1,y^2_2\})\},\\
&S_2(c_3,2,2)=\{(y^2_2,\{y^1_1,y^2_1,y^1_2\})\} \text{\;and}
\end{align*}
}
$\displaystyle{R_2=\bigcup_{c \in \{c_1,c_2,c_3\}}\bigcup_{j \in (c)_{>0}}\bigcup_{p \in \lceil \rho(\{j\})\rfloor}S_2(c,j,p)}.$
\item[(c)]
The set $R_3$ is given by,
\mbox{$\displaystyle R_3=\{(y^1_1,X),(y^2_1,X),(y^1_2,X),(y^2_2,X),(y^1_3,X)\}.$}
\end{itemize}
\end{itemize}
For the index coding problem constructed in this example, we have $\mathcal{M}(\mathcal{I}_{\mathbb{D}}(Z,\mathcal{R}))=5.$
\end{example} 
 
 In the following theorem, it is shown that existence of a perfect linear index coding solution of dimension $n$ over $\mathbb{F}_q$ for $\mathcal{I}_{\mathbb{D}}(Z, \mathcal{R})$ implies the existence of a representation for the discrete polymatroid $n\mathbb{D}$ over $\mathbb{F}_q.$ 
 \begin{theorem}
 \label{Thm_only_if}
 If a perfect linear index coding solution of dimension $n$ over $\mathbb{F}_q$ exists for the index coding problem  $\mathcal{I}_{\mathbb{D}}(Z, \mathcal{R}),$ then the discrete polymatroid $n\mathbb{D}$ is representable over $\mathbb{F}_q.$
\begin{proof} 
 Let $t=(k+\sum_{i =1}^{r}\rho(\{i\}))$ denote the number of messages in the index coding problem $\mathcal{I}_{\mathbb{D}}(Z, \mathcal{R}).$ If  a perfect linear index coding solution of dimension $n$ over $\mathbb{F}_q$ exists for the index coding problem  $\mathcal{I}_{\mathbb{D}}(Z, \mathcal{R})$ over $\mathbb{F}_q,$ then from Theorem \ref{thm5}, there exists a discrete polymatroid $\mathbb{D}'$ representable over $\mathbb{F}_q$ of rank $nt$ on the ground set $\lceil t+1\rfloor$ satisfying conditions (C1) and (C2). Let $V_1,V_2,\dotso,V_{t},V_{t+1}$ denote the vector subspaces over $\mathbb{F}_q$ which form a representation for $\mathbb{D}'.$  From (C1), it follows that $dim(V_i)=n, i \in \lceil t\rfloor$ and $dim(V_{t+1})=n\sum_{i=1}^r \rho(\{i\}).$ Let $A_i,i \in \lceil t\rfloor,$ denote an $nt \times n$ matrix whose columns span $V_i$ and let $A_{t+1}$ denote an $nt \times n(\sum_{i=1}^r \rho(\{i\}))$ matrix whose columns span $V_{t+1}.$ From (C1), it follows that  $rank([A_1 \; A_2 \; \dotso A_{t+1}])=nt.$ Since the matrix $B=[A_1 \; A_2 \; \dotso A_{t}]$ is invertible, it can be taken to be the identity matrix of size $ nt.$ Otherwise, define $A'_i=B^{-1}A_i, i \in \lceil t+1 \rfloor$ and vector subspaces given by the column spans of $A'_i$ will also form a representation for $\mathbb{D}'.$ 
 
Let $A_{t+1}=[C^T D^T]^T,$ where $C$ and $D$ are matrices of size $nk \times n \sum_{i=1}^{r}\rho(\{i\})$ and \\$n \sum_{i=1}^{r}\rho(\{i\})\times n \sum_{i=1}^{r}\rho(\{i\})$ respectively. The matrix $D$ has to be full rank, since (C2) needs to be satisfied for receivers $R \in R_3.$  We can assume $D$ to be identity matrix, otherwise we can define $A'_{t+1}=A_t D^{-1},$ so that the column spans of $A_{t+1}$ and $A'_{t+1}$ are the same. Let $C_i, i \in \lceil r \rfloor,$ denote the matrix obtained by taking only the $({n\sum_{j=1}^{i-1}\rho(\{j\})+1})^{th}$ to $({n \sum_{j=1}^{i}\rho(\{i\}}))^{th}$ columns of $C.$ Let $C_{i,j},j \in \lceil \rho(\{i\})\rfloor$ denote the $n k \times n$ matrix obtained by taking the ${(j-1)n+1}^{th}$ to ${jn}^{th}$ columns of $C_i.$

Let $V'_i$ denote the column span of $C_i$ and $V'_{i,j}$ denote the column span of $C_{i,j}.$ It is claimed that the vector subspaces $V'_i, i \in \lceil r \rfloor,$ form a representation for the discrete polymatroid $n\mathbb{D}.$ To prove the claim, it needs to be shown that for all $S \subseteq \lceil r \rfloor,$ $dim(\sum_{i \in S}V'_i)=n\rho(S).$

We have $\displaystyle{\rho(S)=\max_{b \in \mathbb{D}}\vert b(S)\vert}.$ For $S \subseteq \lceil r \rfloor,$ let $\displaystyle{b^S=\arg \max_{b \in \mathbb{D}}\vert b(S)\vert}.$ 
Let $b_i^S$ denote the $i^{th}$ component of $b^S.$ The vector $b^S$ should be a basis vector for $\mathbb{D},$ otherwise there should exist a basis vector ${\tilde{b}}^S$ of $\mathbb{D}$ for which $b^S<{\tilde{b}}^S$ and $\vert b^S(S) \vert \leq \vert {\tilde{b}}^S(S)\vert.$  
%From the fact that (C2) needs to be satisfied for the receivers  which belong to $S_1(b_S)$ 
Choose $b^S_i$ vector subspaces from the set $\mathcal{V}_i=\{V'_{i,j}:j \in \lceil \rho(\{i\})\rfloor\},$ denoted as ${V'}_{i,o_1},{V'}_{i,o_2},\dotso,{V'}_{i,o_{b^S_i}}$ for every $i\in \lceil r \rfloor.$  
Let $\tilde{V}_i=\sum_{j \in \lceil {b^S_i} \rfloor}V'_{i,o_j}.$ From the fact that (C2) needs to be satisfied for the receivers which belong to $S_1(b^S),$ it follows that  
$dim(\sum_{ i \in \lceil r \rfloor} \tilde{V}_i)=n \vert b^{S}\vert=n\:rank(\mathbb{D}).$  As a result, we have $dim(\sum_{ i \in S}\tilde{V}_i)=n \vert b^{S}(S)\vert.$ Since the vector subspace $\tilde{V}_i$ is a subspace of $V'_i,$ we have $dim(\sum_{i \in S}V'_i)\geq n\vert b^S(S)\vert.$ To complete the proof, it needs to be shown that $dim(\sum_{i \in S}V'_i)\leq n\vert b^S(S)\vert.$

Let $S=\{s_1,s_2\dotso,s_m\}\cup \{s_{m+1},s_{m+2},\dotso,s_l\},$ where $b^S_{s_i}<\rho(\{s_i\}),$ for $i \in \lceil m \rfloor$ and $b^S_{s_i}=\rho(s_i),$ for $i \in \{m+1,m+2,\dotso,l\}.$ Consider the vector $u=(b^S_{s_1}+1)\epsilon_{s_1,r}+\sum_{i \in S \setminus \{s_1\}} b^S_i \epsilon_{i,r}.$ The vector $u$ is an excluded vector. Otherwise, the choice of $b^S,$ $\displaystyle{b^S=\arg \max_{b \in \mathbb{D}}\vert b(S)\vert}$ is contradicted, since $\vert u(S) \vert=\vert b^S(S)\vert +1.$ Let $u_m$ be a minimum excluded vector for which $u_m\leq u.$ The ${s_1}^{th}$ component of $u_m$ has to be $b^S_{s_1}+1,$ otherwise $u_m$ satisfies $u_m <b^S$ and hence cannot be an excluded vector. The vector $u_m$ can be written as $(b^S_{s_1}+1)\epsilon_{s_1,r}+\sum_{i \in S \setminus \{s_1\}}c^S_i \epsilon_{i,r},$ where $c^S_i \leq b^S_i.$ From the fact that (C2) needs to be satisfied for the receivers which belong to the set $S_2(u_m,s_1,p),$ $\forall p \in \lceil \rho(\{s_1\})\rfloor \setminus \{o_1,o_2,\dotso, o_{b^{S}_{s_1}}\},$ it follows that, 
\begin{align*} 
 \sum_{p \in \lceil \rho(\{s_1\})\rfloor \setminus \{o_1,o_2,\dotso, o_{b^{S}_{s_1}}\}} V'_{s_1,p} &\subseteq\\
 &\hspace{-1 cm}\sum_{i \in (u_m)_{>0}\setminus \{s_1\}}\tilde{V}_i +\sum_{j \in \lceil b^{S}_{s_1}\rfloor} V'_{s_1,o_j}.
 \end{align*} 

From the above equation it follows that $\sum_{p \in \lceil \rho(\{s_1\})\rfloor}V'_{s_1,p} \subseteq \sum_{i \in (u_m)_{>0}} \tilde{V}_i \subseteq \sum_{i \in S} \tilde{V}_i.$ By a similar reasoning, $V'_{s_j} \subseteq \sum_{i \in S} \tilde{V}_i, \forall j \in \lceil m \rfloor.$ Since $b^S_{s_j}=\rho(\{s_j\}),$ for $j \in \{m+1,m+2,\dotso,l\},$ we have $V'_{s_j}=\tilde{V}_{s_j},$ for $j \in \{m+1,m+2,\dotso,l\}.$ From the above facts, we have $\sum_{i \in S}V'_i=\sum_{i \in \lceil l \rfloor} V'_{s_i} \subseteq \sum_{i \in S}\tilde{V}_i.$ Hence, $dim(\sum_{i \in S}V'_i)\leq dim(\sum_{ i \in S}\tilde{V}_i)=n \vert b^{S}(S)\vert.$ This completes the proof.
 
%Suppose, we have $dim(\sum_{i \in X}V'_i)> n\vert b^X(X)\vert.$ Then there should exist vector subspaces $\tilde{V}_i, i \in X,$ of $V'_i$ such that $dim(\tilde{V}_i)=nb^X_i+\theta_i,$ not all $\theta_i$'s are zeros and $dim(\sum_{i \in X}\tilde{V}_i)=dim(\sum_{i \in X}V'_i).$ Consider the vector $v=\sum_{i \in X} (nb^X_i+\theta_i) \epsilon_{i,r}.$ The vector $v$ cannot belong to $n\mathbb{D}$; otherwise the definition $\displaystyle{b_X=\arg \max_{b \in \mathbb{D}}\vert b(X)\vert}$ is contradicted. Since the vector $v$ is an excluded vector, there exists $u <v$ such that $u$ is a minimal excluded vector for $n\mathbb{D}.$ Choose $j\in X$ for which $\theta_j >0.$ The vector $\displaystyle{\sum_{i \in X\setminus \{j\}}nb^X_i\epsilon_{i,r}+(nb^X_j+1)\epsilon_{j,r}}$ is a minimal excluded vector for $n \mathbb{D}.$ As a result, the vector $\displaystyle{\sum_{i \in X\setminus \{j\}}b^X_i\epsilon_{i,r}+(b^X_j+1)\epsilon_{j,r}}$ is a minimal excluded vector for $\mathbb{D}.$ (C2) needs to be satisfied for the receiver nodes which belong to $R_2,$ from which it follows that $dim(\sum_{i \in X}V'_i)=n\rho(X)=\sum_{i \in X}b^i_X,$ which contradicts the assumption that $dim(\sum_{i \in X}V'_i)> \sum_{i \in X}b^i_X.$
 \end{proof}
 \end{theorem}
 
% Specializing Theorem \ref{Thm_only_if} for the case $n=1,$ we have the following corollary.
% \begin{corollary}
% If a scalar perfect linear solution  exists for the index coding problem  $\mathcal{I}_{\mathbb{D}}(Z, \mathcal{R})$ over $\mathbb{F}_q,$ then the discrete polymatroid $\mathbb{D}$ is representable over $\mathbb{F}_q.$
% \end{corollary}

For a basis vector $b \in \mathcal{B}(\mathbb{D}),$ let $b_i$ denote the $i^{th}$ component of $b.$ 
Define {\small $${N(\mathbb{D})=1+\max_{i \in \lceil r \rfloor} \sum_{b \in \mathcal{B}(\mathbb{D}):b_i>0}\left({\rho(\{i\})\choose {b_i-1}}\prod_{j \in (b)_{>0}\setminus\{i\}}{\rho(\{j\})\choose b_j}\right)}.$$}
% Note that $N(\mathbb{D})$ is equal to the number of receiver nodes which belong to the set $R_2$ in the index coding problem $\mathcal{I}_{\mathbb{D}}(Z,\mathcal{R}).$ 
 The following theorem shows that the converse of Theorem \ref{Thm_only_if} holds for fields of sufficiently large size.
 \begin{theorem}
 \label{Thm_if}
 If the discrete polymatroid $n\mathbb{D}$ is representable over $\mathbb{F}_q,$ then a perfect linear solution of dimension $n$ exists for the index coding problem $\mathcal{I}_{\mathbb{D}}(Z, \mathcal{R})$ over $\mathbb{F}_{q'},$ where $\mathbb{F}_q'$ is an extension field of $F_q$ with size $q'>N(\mathbb{D}).$
\begin{proof}
See Appendix \ref{App_thm_if}.
\end{proof}
 \end{theorem}
 
From Theorem \ref{Thm_if}, it follows that for $q> N(\mathbb{D}),$ if the discrete polymatroid $n\mathbb{D}$ is representable over $\mathbb{F}_q,$ then there exits a perfect linear solution of dimension $n$ for the index coding problem $\mathcal{I}(Z,\mathcal{R})$ over $\mathbb{F}_q.$
 Combining the results in Theorem \ref{Thm_only_if} and Theorem \ref{Thm_if}, we have the following theorem.
 \begin{theorem}
 \label{thm_if_only_if}
 For $q > N(\mathbb{D}),$ a perfect linear solution of dimension $n$ over $\mathbb{F}_q$ exists for the index coding problem $\mathcal{I}_{\mathbb{D}}(Z,\mathcal{R}),$ if and only if the discrete polymatroid $n\mathbb{D}$ is representable over $\mathbb{F}_q.$
 \end{theorem} 
 
When specialized to the discrete polymatroid $\mathbb{D}(\mathbb{M}),$ where $\mathbb{M}$ is a matroid, the statement of Theorem \ref{thm_if_only_if} reduces to the following statement:
For $q > N(\mathbb{D}(\mathbb{M})),$ a perfect linear solution of dimension $n$ over $\mathbb{F}_q$ exists for the index coding problem $\mathcal{I}_{\mathbb{D}(\mathbb{M})}(Z,\mathcal{R}),$ if and only if the matroid $\mathbb{M}$ has a multi-linear  representation of dimension $n$ over $\mathbb{F}_q.$
 Note that this is the same as the statement of Theorem 12 in \cite{RoSpGe}, with the additional restriction on the field size $q.$ As remarked in Remark 1 in the proof of Theorem \ref{Thm_if}, for the discrete polymatroid $\mathbb{D}(\mathbb{M}),$ this restriction on the field size is unnecessary and the converse of Theorem \ref{Thm_only_if} holds for all $\mathbb{F}_q.$ 
 
 It follows from Theorem \ref{thm_if_only_if} that a perfect linear solution of dimension $n$ exists over a sufficiently large field for the index coding problem $\mathcal{I}_{\mathbb{D}}(Z,\mathcal{R}),$ if and only if the discrete polymatroid $n\mathbb{D}$ is representable and it is stated as the following corollary.
 \begin{corollary}
 \label{cor1}
 A perfect linear solution of dimension $n$ exists for the index coding problem $\mathcal{I}_{\mathbb{D}}(Z,\mathcal{R}),$ if and only if the discrete polymatroid $n\mathbb{D}$ is representable.
 \end{corollary}
 
  Specializing Corollary \ref{cor1} for the case $n=1,$ we have the following corollary.
\begin{corollary}
A scalar perfect linear solution  exists for the index coding problem  $\mathcal{I}_{\mathbb{D}}(Z, \mathcal{R}),$ if and only if the discrete polymatroid $\mathbb{D}$ is representable over $\mathbb{F}_q.$
\end{corollary}

%The proof of Theorem \ref{Thm_if} gives an explicit procedure for obtaining a perfect linear solution of dimension $n$ for the index coding problem $\mathcal{I}_{\mathbb{D}}(Z,\mathcal{R}),$ from a representation for the discrete polymatroid $n\mathbb{D},$ which is summarized below.
%Let $A_i, i \in \lceil r \rfloor,$ denote the matrices over $\mathbb{F}_q$ whose columns span vector subspaces $V_i$ which form a representation for $n\mathbb{D}.$ Let $\Omega_i$ be full rank matrices of size $n \rho(\{i}) \times n \rho(\rho\{i})$ chosen in such a way that the condition (S1) below is satisfied. Let $G_i=A_i \Omega_i$ and let $G_i=[G_i(1)\;G_i(2)\dotso\;G_i(\rho(\{i\}))],$ where $G_i(j), j \in \lceil \rho(\{i\})\rfloor$ are $nk \times n$ matrices.
%\begin{itemize}
%\item[(S1)]
%Let $b=\sum_{i \in \lceil r}b_i \epsilon_{i,r}$ be a basis vector for $\mathbb{D}.$ The 
%\end{itemize}

Note that in Theorem \ref{Thm_if}, the condition that the field size $q'$ should be greater than $N(\mathbb{D})$ is only a sufficient condition. Even for a field size less than or equal to $N(\mathbb{D}),$  a perfect linear solution of dimension $n$ might exist for the index coding problem $\mathcal{I}_{\mathbb{D}}(Z, \mathcal{R}).$ This is illustrated in the following example.
\begin{example}
Consider the index coding problem $\mathcal{I}_{\mathbb{D}}(Z,\mathcal{R}),$ provided in Example \ref{ex31}. For this case, we have $N(\mathbb{D})=9.$  Even though the discrete polymatroid $\mathbb{D}$ has a representation over $\mathbb{F}_2,$ given in Example \ref{ex8}, it is shown in Lemma \ref{lemma_no_f2} below that the index coding problem $\mathcal{I}_{\mathbb{D}}(Z,\mathcal{R})$ does not admit a scalar perfect linear index code over $\mathbb{F}_2.$ This illustrates the fact that the converse of Theorem \ref{Thm_only_if} needs not hold when the field size is not sufficiently large. For a field of size greater than 9, a perfect linear solution of dimension $n$ is guaranteed to exist for $\mathcal{I}_{\mathbb{D}}(Z,\mathcal{R}),$ provided the discrete polymatroid $n\mathbb{D}$ is representable over that field. In this example, we provide a perfect linear solution of dimension 1 for $\mathcal{I}_{\mathbb{D}}(Z,\mathcal{R})$ over the finite field $\mathbb{F}_4=\{0,1,\alpha,1+\alpha\}$ of size 4, where $\alpha$ is a root of the irreducible polynomial $x^2+x+1=0$ over $\mathbb{F}_2.$ It can be verified that the function $f$ given by,
{
\begin{align*}
f(Z)&=\begin{bmatrix}y^1_1 & y^2_1&y^1_2 &y^2_2& y^1_3\end{bmatrix}\\
&\hspace{1.5 cm}+\begin{bmatrix}x_1&x_2&x_3\end{bmatrix}\underbrace{\begin{bmatrix}1 & 0 & 1 & 1 & 0\\0 & 1 & 1 & 1&0\\0 & 0 & 1+ \alpha &1&1\end{bmatrix}}_{A}
\end{align*}
}forms a scalar perfect linear index code for $\mathcal{I}_{\mathbb{D}}(Z,\mathcal{R})$ over $\mathbb{F}_4.$ Let $V_1$ denote the span of the first two columns of $A$ over $\mathbb{F}_4.$ Also, let $V_2$ denote the span of the third and fourth columns of $A,$ and let $V_3$ denote the span of  the last column of $A$ over $\mathbb{F}_4.$ The vector subspaces $V_1,$ $V_2$ and $V_3$ form a representation over $\mathbb{F}_4$ for the discrete polymatroid $\mathbb{D}.$

\begin{lemma}
\label{lemma_no_f2}
The index coding problem $\mathcal{I}_{\mathbb{D}}(Z,\mathcal{R})$ provided in Example \ref{ex31} does not admit a scalar perfect linear solution over $\mathbb{F}_2.$
\begin{proof}
On the contrary, assume that there exists a scalar perfect linear solution over $\mathbb{F}_2$ for $\mathcal{I}_{\mathbb{D}}(Z,\mathcal{R}).$ A scalar perfect linear solution exists for $\mathcal{I}_{\mathbb{D}}(Z,\mathcal{R})$ only if $\mathbb{D}$ is representable over $\mathbb{F}_2.$ Note that $\mathbb{D}$ is indeed representable over $\mathbb{F}_2$ and a representation for $\mathbb{D}$ over $\mathbb{F}_2$ has been provided in Example \ref{ex8}. Every scalar perfect linear solution for $\mathcal{I}_{\mathbb{D}}(Z,\mathcal{R})$ can be written as $f(Z)=[y^1_1\;y^2_1\;y^1_2\;y^2_2\;y^1_3]A+[x_1\;x_1\;x_3][G_1\;G_2\;G_3],$ where $A$ is a $5 \times 5$ over $\mathbb{F}_2,$  $G_1$ and $G_2$ are $3 \times 2$ matrices over $\mathbb{F}_2,$ and $G_3$ is a $3 \times 1$ matrix over $\mathbb{F}_2.$ In order to ensure the existence of decoding matrices for the receivers which belong to the set $R_3,$ $A$ needs to be full rank. Hence, without loss of generality, we can assume $A$ to be the identity matrix. Also, without loss of generality, the matrix $G_1$ can be assumed to be $\begin{bmatrix}1&0\\0&1\\0&0 \end{bmatrix}$ and the first column of $G_2$ can be assumed to be $ \begin{bmatrix}0\\0\\1\end{bmatrix}.$ The reason for this is that if the matrix $G$ obtained by the concatenation of $G_1$ and the first column of $G_2$ is not the identity matrix, taking $G^{-1}G_i$ to be $G'_i, i \in \lceil 3 \rfloor,$ the function $f'(Z)=[y^1_1\;y^2_1\;y^1_2\;y^2_2\;y^1_3]+[x_1\;x_2\;x_3][G'_1 G'_2 G'_3]$ forms a valid scalar perfect linear index code.  The second column of $G_2$ and the only column of $G_3$ need to be chosen. It is claimed that the only possibility for $G_3$ is $G_3=[1\;1\;1]^T.$ $G_3$ cannot be $[1\;0\;0],[0\;1\;0]$ and $[1\;1\;0],$ since $dim(V_1+V_3)=3.$ The only other possibilities for $G_3$ are $[0 \; 1 \; 1]^T,[1\;0\;1]^T,[0\;0\;1]^T$ and $[1\;1\;1]^T.$ If $G_3=[0 \; 1 \; 1]^T,$ it will not be possible to find a decoding function for the receiver nodes $(x_i,\{y^2_1,y^1_2,y^1_3\}), i \in \lceil 3 \rfloor.$ Similarly, if $G_3=[1\;0\;1]^T$ ($G_3=[0\;0\;1]^T$), it will not be possible to find decoding functions for the receiver nodes $(x_i,\{y^1_1,y^1_2,y^1_3\})$ ($(x_i,\{y^1_1,y^1_2,y^1_3\})$), where $i \in \lceil 3 \rfloor.$ Since $dim(V_2+V_3)=2$ and $dim(V_2)=2,$ the only possibilities for the second column of $G_2$ are $[1\;1\;0]^T$ and $[1\;1\;1]^T.$ If the second column of $G_2$ is equal to $[1\;1\;0]^T$ ($[1\;1\;1]^T),$ then it will not be possible to find decoding functions for the receiver nodes $(x_i,\{y^1_1,y^2_1,y^2_2\})$ ($(x_i,\{y^1_1,y^2_2,y^1_3\})$). This shows that there cannot exist a scalar perfect linear solution over $\mathbb{F}_2$ for $\mathcal{I}_{\mathbb{D}}(Z,\mathcal{R}).$
\end{proof}
\end{lemma}
\end{example}

\section{Other Possible Connections among Network Coding, Index Coding and Discrete Polymatroids}
In Section \ref{IV A}, a connection between existence of linear network coding solution for a network and representable discrete polymatroids was established. A similar connection between linear index coding and representable discrete polymatroids was established in Section \ref{VI A}. In this section, we explore  other possible connections among network coding, index coding and discrete polymatroids.

In \cite{RoSpGe}, a construction of index coding problem from a network coding problem was provided and it was shown that a linear solution to the network coding problem exists if and only if there exists a perfect linear solution for the index coding problem. This result was extended to non-linear network/index coding solutions in \cite{EfRoLa}. To establish connections between linear network coding and discrete polymatroid representability, one can translate the problem of linear solvability of a network to the problem of finding a perfect linear solution to an associated index coding problem using the results in \cite{RoSpGe,EfRoLa} and then use Theorem \ref{thm5} to find a connection with discrete polymatroids. Such a connection between linear network coding and discrete polymatroids is obtained in Section \ref{VII A}. Similarly, in Section \ref{VII B}, we obtain a connection between linear index coding and respresentable discrete polymatroids, using Theorem \ref{thm3} and the fact that index coding problem can be viewed as a special case of network coding problem. Also, it is shown that the results in Section \ref{VII A} and Section \ref{VII B} are equivalent to the ones in Theorem \ref{thm3} and Theorem \ref{thm5} respectively.

\subsection{Network coding to Discrete Polymatroids via Index Coding}
\label{VII A}
In this subsection, we restrict to vector linear network coding solutions, i.e., we do not consider solutions for which message vector lengths are different from the edge vector length\footnote{The reason for this restriction is that by definition in Section \ref{III B}, index coding problem assumes message vectors of equal length. In this subsection, connection between network coding and discrete polymatroid is obtained via index coding and a result from \cite{RoSpGe,EfRoLa}.}. First some notations are introduced and a result from \cite{RoSpGe,EfRoLa} is stated.

For a network coding problem with notations and terminologies as defined in Section \ref{III A}, let $\mathcal{V}_{>0}$ denote the set of vertices which demand at least one message, i.e., $\mathcal{V}_{>0}=\{v\in \mathcal{V}:\vert \delta(v)\vert>0\}.$ Also, let the set of edges be given by $\mathcal{E}=\{1,2,\dotso,\vert\mathcal{S}\vert,\vert\mathcal{S}\vert+1,\dotso,\vert\mathcal{E}\vert\},$ with $\mathcal{S}=\lceil \vert \mathcal{S}\vert \rfloor$ being the set of input edges. 

Consider the following index coding problem $\mathcal{I}(X,\mathcal{R})$ constructed from a network coding problem  using the procedure in \cite{EfRoLa}: 
\begin{itemize}
\item
The set of messages $X=\{x_1,x_2,\dotso,x_{\vert \mathcal{S} \vert}, \dotso, x_{\vert \mathcal{E}\vert}\}.$
\item
The set of receiver nodes $\mathcal{R}=\mathcal{R}_1\cup \mathcal{R}_2 \cup \mathcal{R}_3,$ 
\begin{itemize}
\item
$\mathcal{R}_1=\left\{\left(x_e,H_e\right);e\in \lceil \vert \mathcal{E}\setminus \mathcal{S}\vert\rfloor\right\},$ where $ H_e=\left\lbrace x_i : i\in In(head(e))\right\rbrace.$
\item
$\mathcal{R}_2=\bigcup_{v \in \mathcal{V}_{>0}}\{(x_i,H_v); x_i \in \delta(v)\},$ where $H_v=\{x_e:e\in In(v)\}.$
\item
$\mathcal{R}_3=\{(x_e,H);e \in \lceil \vert \mathcal{E}\setminus \mathcal{S}\vert\rfloor\},$ where $H=\{x_1,x_2,\dotso,x_{\vert \mathcal{S}\vert}\}.$ 
\end{itemize}
\end{itemize}

For the index coding problem defined above, we have $\mathcal{M}(\mathcal{I}(X,\mathcal{R}))=\vert \mathcal{E}\setminus \mathcal{S}\vert.$ From \cite{EfRoLa}, a perfect linear index coding solution of length $c=n\vert \mathcal{E}\setminus \mathcal{S}\vert$ exists for this index coding problem if and only if the network from which this was constructed admits a vector linear solution. Combining this result with the result in Theorem \ref{thm5}, we obtain the result in the following theorem.
 
\begin{theorem}
\label{thm9}
A network has a vector linear solution of dimension $n$ over $\mathbb{F}_q,$ if and only if there exists a  discrete polymatroid $\mathbb{D}$ on ground set $\lceil \vert \mathcal{E}\vert+1 \rfloor$ representable over $\mathbb{F}_q$ satisfying the following conditions: 
\begin{itemize}
\item[{\small (NID1)}]
{\footnotesize$\rho(\{i\})=n, \forall i \in \mathcal{E},$ $\rho(\{\vert \mathcal{E}\vert+1\})=n \vert \mathcal{E}\setminus \mathcal{S}\vert$,} {\footnotesize$\rho(\mathcal{E})=n\vert \mathcal{E}\vert$} and $rank(\mathbb{D})=n\vert \mathcal{E}\vert.$ 
\item[{\small (NID2)}]
For {\footnotesize$e \in \{\vert\mathcal{S}\vert+1,\vert\mathcal{S}\vert+2,\dotso,\vert \mathcal{E}\vert\},$}
{\footnotesize $$\rho(\{e\}\cup In(head(e))\cup\{\vert \mathcal{E}\vert+1\}) =\rho(In(head(e))\cup\{\vert \mathcal{E}\vert+1\}).$$}
 \item[{\small (NID3)}]
 For {\footnotesize$v \in \mathcal{V},$ $\rho(\delta(v)\cup In(v)\cup \{\vert \mathcal{E}\vert+1\})=\rho(In(v)\cup \{\vert \mathcal{E}\vert+1\}).$}
 \item[{\small (NID4)}]
For {\footnotesize$e \in \{\vert\mathcal{S}\vert+1,\vert\mathcal{S}\vert+2,\dotso,\vert \mathcal{E}\vert\},$\\ $\rho(\{e\}\cup \mathcal{S}\cup\{\vert\mathcal{E}\vert+1\})=\rho(\mathcal{S}\cup\{\vert\mathcal{E}\vert+1\}).$}
\end{itemize}
\end{theorem}

Note that the the discrete polymatroids which arise in Theorem \ref{thm3} and in Theorem \ref{thm9} are not the same. In Appendix \ref{appB}, it is shown that the results in Theorem \ref{thm3} and \ref{thm9} are equivalent, i.e., there exists a representable discrete polymatroid with respect to which a network is discrete polymatroidal if and only if there exists a representable discrete polymatroid satisfying the conditions in Theorem \ref{thm9}. 
\subsection{Index coding to Discrete Polymatroids via Network Coding}
\label{VII B}
One can obtain a connection between index coding and discrete polymatroids by posing the index coding problem as an equivalent network coding problem and then using the result in Theorem \ref{thm3}.

For the index coding problem defined in Section \ref{III B}, let $(x_i,H)\in \mathcal{R}, H=\{x_{j_1},x_{j_2},\dotso,x_{j_l}\}$ denote a receiver node.  
The problem of finding a linear solution to an index coding problem  of length $c$ and dimension $n$ is equivalent to finding a linear $(n,n,\dotso,n;c)$-FNC solution for the following network: The set of vertices is given by $\mathcal{V}=\{v_1,v_2,\dotso,v_m,v_{m+1},v_{m+2},v_{m+3},\dotso,v_{m+2+\vert \mathcal{R}\vert}\}.$ The first $m$ vertices $v_i, i \in \lceil m \rfloor$ are those vertices at which the $m$ messages are generated. The vertex $v_{m+1}$ has one incoming edge each from the vertices in the set $\{v_1,v_2,\dotso v_m\}$ and $v_{m+2}$ has a single incoming edge from $v_{m+1}.$ For $j \in \{m+3,\dotso,m+2+\vert \mathcal{R}\vert\},$ the node $v_j$ has incoming edges from vertices in the set $\{v_k:x_{k}\in H\}$ and demands $x_i.$

Let $\mathcal{S}=\lceil m \rfloor$ denote the $m$ source edges and let $e_{i,i'}$ denote an edge connecting vertices $v_i$ and $v_{i'}.$
From Theorem \ref{thm3}, the network thus defined above admits a linear $(n,n,\dotso,n;c)$-FNC solution if and only if it is $(n,n,\dotso,n;c)$-discrete polymatroidal with respect to a representable discrete polymatroid $\mathbb{D},$ i.e., there exists a function $f$ from the set of edges to the ground set $\lceil r \rfloor$ of $\mathbb{D}$ satisfying (DN1)--(DN4). From (DN1), since $f$ is one-to-one on the elements of $\mathcal{S},$ let $f(i)=i,$ for $i \in \lceil m \rfloor.$ From (DN2)--(DN4), it follows that the discrete polymatroid $\mathbb{D}$ should satisfy certain conditions which are stated in the following theorem:

\begin{theorem}
\label{thm_IC_NC_DPMD}
A vector linear index code over $\mathbb{F}_q$ of length $c$ and dimension $n$ exists for an index coding problem $\mathcal{I}(X,\mathcal{R}),$ if and only if there exists a discrete polymatroid $\mathbb{D}$ representable over $\mathbb{F}_q$ on the ground set $\lceil r \rfloor$ satisfying the following conditions:
\begin{itemize}
\item[{\small(IND1)}]
$\sum_{i \in \lceil m \rfloor} n\epsilon_{n,r} \in \mathbb{D}.$
\item[{\small(IND2)}]
$\rho(\{i\})=n, \forall i \in \lceil m \rfloor$ and $\max_{i\in \mathcal{E}\setminus \mathcal{S}} \rho(f(\{i\}))=c.$
\item[{\small(IND3)}]
$\forall (x_i,H) \in \mathcal{R},$ where $H=\{x_{j_1},x_{j_2},\dotso,x_{j_l}\},$\\ {\footnotesize$\rho(\{i\} \cup \{j_1,j_2,\dotso j_l\} \cup \{m+1\})=\rho(\{j_1,j_2,\dotso,j_l\}\cup \{m+1\}).$}
\item[{\small(IND4)}]
$\rho(f(\{e_{m+1,m+2}\})\cup \lceil m \rfloor)=\rho(\lceil m \rfloor).$
%\item[(IND5)]
%For $i\in \lceil m \rfloor,$ $\rho(f(\{e_{i,m+1}\}\cup \{i\})=\rho(\{i\}).$
%\item[(IND6)]
%For $i\in \{m+3,m+4,\dotso,m+\vert\mathcal{R}+2\},$ $\rho(f(\{e_{m+1,m+2}\}\cup f(\{i\}))=\rho(f(\{e_{m+1,m+2}\})).$
\end{itemize}  
\end{theorem}

Note that the discrete polymatroid which satisfies the conditions in Theorem \ref{thm_IC_NC_DPMD} need not be the same as the one which arises in Theorem \ref{thm5}. For example, the ground set of the discrete polymatroid  in Theorem \ref{thm5} has $m+1$ elements, whereas there is no such restriction on the one in Theorem \ref{thm_IC_NC_DPMD}. In Appendix \ref{appC}, it is shown that the results in Theorem \ref{thm5} and \ref{thm_IC_NC_DPMD} are equivalent, i.e., there exists a representable discrete polymatroid satisfying the conditions in Theorem \ref{thm5} if and only if there exists a representable discrete polymatroid satisfying the conditions in Theorem \ref{thm_IC_NC_DPMD}.  
\section{Discussion}
In this paper, the connections between linear network coding, linear index coding and representable discrete polymatroids were explored. The notion of a discrete polymatroidal network was introduced and it was shown that the existence of a linear solution for a network is connected to the network being discrete polymatroidal. Also, it was shown that a linear solution exists for an index coding problem if and only if there exists a representable discrete polymatroid satisfying certain conditions which are determined by the index coding problem considered. Also, constructions of networks and index coding problems from discrete polymatroids were provided, for which the existence of linear solutions depends on discrete polymatroid representability. This paper considers only representable discrete polymatroids. An interesting problem for future research is to investigate whether any connections exist between non-representable discrete polymatroids and non-linear network/index coding solutions.  
%The connection between the vector linear solvability of networks over a field $\mathbb{F}_q$ and the representation of discrete polymatroids was established. It was shown that for a network, a vector linear solution over a field $\mathbb{F}_q$ exists if and only if the network is discrete polymatroidal with respect to a representable discrete polymatroid. An algorithm to construct networks from discrete polymatroids was provided. Sample constructions of networks from representable discrete polymatroids which have vector linear solutions but no scalar linear solution over $\mathbb{F}_q$ were provided.   
%\end{description}

\begin{appendices}
\section{Proof of Theorem \ref{Thm_if}}
\label{App_thm_if}
Before proving Theorem \ref{Thm_if}, some useful lemmas are stated.

\begin{lemma}
\label{lemma_6}
If $b$ is a basis vector of a discrete polymatroid $\mathbb{D},$ then $nb$ is a basis vector of the discrete polymatroid $n \mathbb{D}.$
\begin{proof}
Since $b \in \mathcal{B}(\mathbb{D}),$ we have $\vert b(X)\vert \leq \rho(X), \forall X \subseteq \lceil r \rfloor.$ Hence, we have $\vert(nb)(X)\vert=n \vert b(X)\vert\leq n \rho(X)=\rho^{n \mathbb{D}}(X), \forall X \subseteq \lceil r \rfloor.$ Hence, it follows that $nb \in n\mathbb{D}.$ To complete the proof, it needs to be shown that there does not exist $u \in  n\mathbb{D}$ for which $u>b.$ On the contrary, assume that such a $u$ exists. Then, we have, $\vert u \vert > \vert n b\vert =n \vert b \vert =n \:rank(\mathbb{D})=rank(n\mathbb{D}),$ which means that $u \notin n \mathbb{D},$ a contradiction.
\end{proof}
\end{lemma}

\begin{lemma}
\label{lemma_7}
Consider a representable discrete polymatroid $\mathbb{D},$ with vector subspaces $V_1,V_2,\dotso V_r$ forming a representation for $\mathbb{D}.$ Let $b$ be a basis vector vector of $\mathbb{D}$ and let $b_i$ denote the $i^{th}$ component of $b.$ There exists vector subspaces $V'_i$ of $V_i, i \in (b)_{>0},$ such that $dim(V'_i)=b_i$ and $dim(\sum_{i \in (b)_{>0}} V'_i)=rank(\mathbb{D}).$ 
\begin{proof}
Follows from Lemma 6.3 in \cite{FaMaPa}.
\end{proof}
\end{lemma}

Now we proceed to give the proof of Theorem \ref{Thm_if}.
\subsubsection*{{PROOF OF THEOREM \ref{Thm_if}}}
 %\begin{proof}
 Assume that the vector subspaces $V_i, i \in \lceil r \rfloor,$ form a representation for the discrete polymatroid $n\mathbb{D}$ over $\mathbb{F}_q.$ Let $A_i, i \in \lceil r \rfloor,$ denote a matrix over $\mathbb{F}_q$ of size $n k \times n \rho(\{i\})$ whose columns span $V_i.$ Let $A'_i=A_i \Gamma_i,$ where $\Gamma_i$ is a matrix of size $n \rho(\{i\}) \times n \rho(\{i\}),$ whose entries are indeterminates. Let $A'_i(j), j \in \lceil \rho(\{i\} \rfloor),$ denote the submatrix of $A'_i$ of size $nk \times n$ obtained by taking only the ${(j-1)n+1}^{th}$ to ${jn}^{th}$ columns of $A'_i.$ 
 
Let $b$ be a basis vector of $\mathbb{D}.$ Let $b_i$ denote the $i^{th}$ element of $b.$ Let us define a set of polynomials with the entries of the matrices $\Gamma_i, i \in \lceil r \rfloor$ as the indeterminates as follows: Choose $b_i$ integers from the set $\lceil \rho(\{i\})\rfloor,\forall i \in \lceil r \rfloor,$ denoted by $l_1^i,l_2^i,\dotso,l_{b_i}^i.$ Consider the polynomial which is the determinant of the $nk \times nk$ matrix obtained by the concatenation of all the matrices $A'_i(l_j^i),$ where $j \in \lceil b_i \rfloor$ and $i \in \lceil r \rfloor.$ Let $\mathcal{P}(b)$ denote the set of all polynomials obtainable using the procedure mentioned above, for a fixed basis vector $b.$

Suppose we want to find an assignment for the indeterminates in the matrices $\Gamma_i, i \in \lceil r \rfloor,$ from a field $\mathbb{F}_{q'},$ such that the following conditions are satisfied:\\
(i) the determinant of all the matrices $\Gamma_i$ evaluate to non-zero values and\\
(ii) for all the basis vectors $b \in \mathcal{B}(\mathbb{D}),$ all the polynomials which belong to the set $\mathcal{P}(b)$ evaluate to non-zero values. 

The claim is that from a extension field $\mathbb{F}_{q'}$ of $\mathbb{F}_q$ of size greater than $N(\mathbb{D}),$  it is possible to find an assignment for the indeterminates such that the above two conditions are satisfied.

\begin{remark}
If the discrete polymatroid $\mathbb{D}$ is of the form $\mathbb{D}(\mathbb{M}),$ where $\mathbb{M}$ is a matroid,  assigning $\Gamma_i$'s to be identity matrices, the two conditions given above are satisfied. There is no need to look for an extension field whose size is greater than $q$ for this case.  
\end{remark}

Towards proving the claim, we first show that all the polynomials which belong to the set $\mathcal{P}(b)$ are non-zero polynomials, for all $ b \in \mathcal{B}(\mathbb{D}).$ To show this, it is enough to show that there exists an assignment of values for the indeterminates for each one of the polynomials which belong to $\mathcal{P}(b),$ possibly different for different polynomials, such that the polynomials evaluate to non-zero values in $\mathbb{F}_{q'}.$  

From Lemma \ref{lemma_6}, it follows that for $b \in \mathbb{D},$ $nb \in \mathcal{B}(n \mathbb{D}).$   Since $nb \in \mathcal{B}(n\mathbb{D}),$ from Lemma \ref{lemma_7} it follows that there exists vector subspaces $V'_i$ of $V_i,i \in (b)_{>0},$ of dimension $n b_i$ such that $dim(\sum_{i \in (b)_{>0}}V'_i)=n k.$  Let $B_i$ denote a matrix whose columns span $V'_i.$ Since the columns of the matrix $A_i$ form a basis for $V_i$ and $V'_i$ is a subspace of $V_i,$ $B_i$ can be written as $A_i \Lambda_i,$ where $\Lambda_i$ is an $n \rho(\{i\})\times n b_i$ matrix over $\mathbb{F}_q.$ The determinant of the $nk \times nk $ matrix obtained by the concatenation of the matrices $B_i, i \in (b)_{>0}$ is non-zero. A polynomial which belongs to $\mathcal{P}(b)$ is nothing but the determinant of a $n k \times nk$ matrix obtained by the concatenation of matrices obtained multiplying the matrix $A_i$ by $n b_i$ columns of $\Gamma_i,$ for every $i \in (b)_{>0}.$ Assigning the $n b_i$ columns of $\Gamma_i$ to be the columns of $B_i,$ the polynomials which belong to $\mathcal{P}(b)$ evaluate to non-zero values and hence they are non-zero polynomials. 

To find an assignment for for the indeterminates in $\Gamma_i, i \in \lceil r \rfloor$ such that the two conditions (i) and (ii) are satisfied, it suffices to find an assignment for the indeterminates such that the following polynomial evaluates to a non-zero value:
$$P(\Gamma_1,\Gamma_2,\dotso,\Gamma_r)=\left(\prod_{i \in \lceil r \rfloor} det(\Gamma_i)\right)\left( \prod_{b \in \mathcal{B}(\mathbb{D})} \prod_{p \in \mathcal{P}(b)}p\right).$$
If the field size $q'$ is greater than the degree of the above the polynomial in every indeterminate, then an assignment for the indeterminates form $\mathbb{F}_{q'}$ for which the above polynomial evaluates to a non-zero value is guaranteed to exist (follows from Lemma 19.27 in Chapter 19, \cite{Ye}). 

Consider an indeterminate which is an entry of the matrix $\Gamma_i,$ which is denoted by $\gamma_i.$ For a basis vector $b$ for which $b_i>0,$ there are $\left({\rho(\{i\})\choose {b_i-1}}\prod_{j \in (b)_{>0}\setminus\{i\}}{\rho(\{j\})\choose b_j}\right)$ polynomials in $\mathcal{P}(b)$ which involve $\gamma_i$ and in each one of these polynomials, the degree of $\gamma_i$ is one. Also, the degree of the polynomial $det(\Gamma_i)$ in $\gamma_i$ is one. Hence, the degree of the polynomial $P(\Gamma_1,\Gamma_2,\dotso,\Gamma_r)$ in $\gamma_i$ is $1+\left({\rho(\{i\})\choose {b_i-1}}\prod_{j \in (b)_{>0}\setminus\{i\}}{\rho(\{j\})\choose b_j}\right).$ Maximizing over all $i \in \lceil r \rfloor,$ it follows that for $q' > N(\mathbb{D})$ there exists an assignment for $\Gamma_i, i \in \lceil r \rfloor$ for which the polynomial $P(\Gamma_1,\Gamma_2,\dotso,\Gamma_r)$ evaluates to a non-zero value. Let $\Omega_i,i\in \lceil r \rfloor,$ denote one such assignment. Let $G_i=A_i \Omega_i.$ Note that $G_i$ has a rank $n \rho(\{i\})$ and the columns of $G_i$ span $V_i.$ Let $\theta=[x_1 \; x_2 \dotso x_k].$ 
Define the function $f$ as,

{\small
\begin{align*}
f(Z)&\triangleq[\tau^1_1\;\tau^2_1\dotso \tau^{\rho(\{1\})}_1\; \tau^1_2\;\tau^2_2\dotso \tau^{\rho(\{2\})}_1\dotso \tau^1_r\;\tau^1_r\dotso \tau^{\rho(\{r\})}_r]\\
&\hspace{0 cm}=[y^1_1\;y^2_1\dotso y^{\rho(\{1\})}_1\; y^1_2\;y^2_2\dotso y^{\rho(\{2\})}_1\dotso y^1_r\;y^1_r\dotso y^{\rho(\{r\})}_r]\\
&\hspace{5.5 cm}+\theta[G_1 \; G_2\dotso G_r].
\end{align*}
}

Let $G_i=[G_i(1)\;G_i(2)\dotso\;G_i(\rho(\{i\}))],$ where $G_i(j), j \in \lceil \rho(\{i\})\rfloor$ are $nk \times n$ matrices.
It is shown below that $f$ forms a perfect linear index coding solution of dimension $n$ over $\mathbb{F}_q$ for the index coding problem $\mathcal{I}_{\mathbb{D}}(Z, \mathcal{R}).$

For a receiver node $R=(y^j_i,X)$ which belongs to $R_3,$ the function $\Psi_R(f(Z),X)=\tau^j_i-\theta G_i(j)$ forms a valid decoding function. 

Consider a receiver node $\left(x_j,\bigcup_{l \in (b)_{>0}} \eta_l\right)$  which belongs to the set $S_1(b),$ where $j \in \lceil k \rfloor, \eta_l \subseteq \zeta_l, \vert \eta_l \vert =b_l,$ and $\zeta_l=\{y^1_l,y^2_l,\dotso y^{\rho(\{l\})}_{l}\}.$ Consider the matrix $M$ of size $nk \times nk$ obtained by the concatenation of the matrices $G_l(t),$ where $l \in (b)_{>0}$ and $t$ is such that $y^t_l \in \eta_l.$ By virtue of the choice of $G_i$'s, the matrix $M$ is full rank. Let $\chi$ denote the vector obtained by the concatenation of the vectors which belong to $\bigcup_{l \in (b)_{>0}} \eta_l.$ Let $\omega$ denote the concatenation of $\tau_l^t$'s for which $l \in (b)_{>0}$ and $t$ is such that $y_l^t \in \eta_l.$ The vector $\theta$ is given by $(\omega-\chi)M^{-1}.$ Hence, decoding functions exist for receivers which belong to $R_1.$

Let $c$ be a minimal excluded vector for $\mathbb{D}$ and let $c_l$ denote the $l^{th}$ component of $c.$ Consider a receiver node $ (y_j^p,\Gamma_1 \cup \Gamma_2)$ which belongs to $S_2(c,j,p),$ where $j \in (c)_{>0},$ $p \in \lceil \rho(\{j\})\rfloor,$ $\Gamma_1=\bigcup_{l \in (c)_{>0}\setminus \{j\}} \eta_l,\; \eta_l \subseteq \zeta_l, \vert \eta_l \vert=c_l, \Gamma_2\subseteq \zeta_{j} \setminus \{y^p_j\}, \vert \Gamma_2\vert=c_j-1.$ Let $M'$ denote the concatenation of the matrices $G_l(t),$ where $l \in (c)_{>0}$ and $t$ is such that $y^t_l \in \Gamma_1 \cup \Gamma_2.$ It is claimed that $rank([M'\;G_j(p)])=rank(M').$ Since $c$ is a minimal excluded vector, the vector $u=\sum_{i \in (c)_{>0}} c_i \epsilon_{i,r} + (c_j-1)\epsilon_{i,r}$ belongs to $\mathbb{D}.$ Hence, there exists a basis vector $b \in \mathcal{B}(\mathbb{D})$ for which $u\leq  b.$ Note that $b_j=c_j-1,$ since, if $b_j >c_j-1,$ then $c<b$ and hence $b$ and $c$ respectively cannot be simultaneously basis and excluded vectors. Define the set of matrices $\mathcal{G}_j=\{G_j(o) :y ^o_j \in \Gamma_2\}$ and for $i \in (c)_{>0}\setminus\{j\},$ define $\mathcal{G}_i=\{G_i(o) :y ^o_i \in \Gamma_1\}.$ Note that the matrix $M'$ is the concatenation of the matrices which belong to the sets  $\mathcal{G}_i, i \in (c)_{>0}.$ For $i \in (b)_{>0},$ define $\mathcal{G}'_i$ to be a set of matrices which is a subset of size $(b)_{>0}-(c)_{>0}$ of the set $ \{G_i(o) : o \in \lceil \rho(\{i\})\rfloor\}\setminus \mathcal{G}_i.$ Note that $\mathcal{G}'_j$ is the null set. Let $M''$ denote the $nk \times nk$ matrix obtained by the concatenation of matrices which belong to $\mathcal{G}_i$ and $\mathcal{G}'_i, i \in (b)_{>0}.$ The choice of the matrices $G_i$'s ensures that $M''$ is of full rank equal to $nk.$ Since $M'$ is a submatrix of $M''$ of size $nk \times n (\vert c \vert -1),$ $M'$ should be of rank $n (\vert c \vert -1).$ Note the the vector subspace $V_i$ is the column span of $G_i.$ Since, the vector subspaces $V_i, i \in \lceil r \rfloor,$ form a representation of $n\mathbb{D},$ the rank of the matrix $\tilde{M}$ obtained by the concatenation of the matrices $G_l,l \in (c)_{>0}$ should equal $n\rho((c)_{>0}),$ which is equal to $n (\vert c\vert -1).$ Since $[M'\;\;G_j(p)]$ is a submatrix of $\tilde{M},$ we have,
 $$n (\vert c \vert -1)=rank(M')\leq rank([M'\;G_j(p)])\leq rank(\tilde{M}).$$

Since $rank(\tilde{M})=n (\vert c \vert -1),$ we have $rank([M'\;G_j(p)])=rank(M')$ and the matrix $G_j(p)$ can be written as $M'W,$ where $W$ is of size $n(\vert c \vert-1)\times n.$ Let $\tau'$ denote the concatenation of the vectors which belong to the set $\{\tau_i^o: y_i^o \in \Gamma_1\cup \Gamma_2\}$ and let $y'$ denote the concatenation of the vectors which belong to the set $\Gamma_1\cup \Gamma_2.$ We have, $y^p_j=\tau^p_j-\theta G_j(p)=\tau^p_j-(\theta M')W=\tau^p_j-(\tau'-y')W.$ Hence, decoding functions exist for the receivers which belong to $R_2.$ This completes the proof of Theorem \ref{Thm_if}.
\section{Equivalence of Theorem \ref{thm3} and Theorem \ref{thm9}}
\label{appB}
\subsection*{From Theorem \ref{thm9} to Theorem \ref{thm3}}
Let $V_i,i\in \lceil \vert \mathcal{E}\vert+1\rfloor,$  denote the vector spaces which form a representation of the discrete polymatroid $\mathbb{D}$ in Theorem \ref{thm9}, with $V_i$ being the column span of a matrix $A_i.$ For $i \in  \mathcal{E},$ 
since $\rho(\{i\})=n$ and $rank(\mathbb{D})=n\vert \mathcal{E}\vert,$ $A_i$ is of size $n \vert \mathcal{E}\vert \rfloor \times n.$ Since $\rho(\{\vert \mathcal{E}\vert +1\})=n\vert \mathcal{E}\setminus \mathcal{S}\vert,$ $A_{\vert \mathcal{E}\vert+1}$ is of size $n \vert \mathcal{E}\vert \times n \vert \mathcal{E}\setminus \mathcal{S}\vert.$ Let $B$ denote the $n \vert \mathcal{E}\vert \times n \vert \mathcal{E}\vert$ matrix $[A_1 A_2 \dotso A_{\vert \mathcal{E}\vert}].$ Since $\rho(\mathcal{E})=n\vert \mathcal{E}\vert,$ $B$ is invertible and can be assumed to identity (If $B$ is not identity, one can define verctor spaces $V'_i, i \in \lceil \vert \mathcal{E}\vert+1\rfloor$ to be column span of $A'_i=B^{-1}A_i$ which will also form a representation of $\mathbb{D}.$

Also, one can assume the lower $n \vert \mathcal{E} \setminus \mathcal{S}\vert \times n \vert \mathcal{E} \setminus \mathcal{S}\vert$ sub-matrix of $A_{\vert \mathcal{E}\vert+1}$ to be an identity matrix. The reason for this is as follows: Define $B_1=[A_1 A_2 \dotso A_{\vert \mathcal{S}\vert}], B_2=[A_{\vert \mathcal{S}\vert+1} A_{\vert \mathcal{S}\vert+2}\dotso A_{\vert \mathcal{E}\setminus \mathcal{S}\vert}].$ Also, let $\Gamma_1$ and $\Gamma_2$ respectively denote the upper $n \vert \mathcal{S}\vert \times n \vert \mathcal{E} \setminus \mathcal{S}\vert$ submatrix and lower  $n \vert  \mathcal{E} \setminus \mathcal{S}\vert \times n \vert \mathcal{E} \setminus \mathcal{S}\vert$ submatrix of $A_{\vert \mathcal{E}\vert+1}$ respectively. From (NID4), it follows that $rank([B_1 \;A_{\vert \mathcal{E}\vert+1}])=rank([B_1\; B_2\; A_{\vert \mathcal{E}\vert+1}])=n \vert \mathcal{E}\vert.$ Let $I_a$ denote the identity matrix of order $a$ and $0_{a \times b}$ denote the all-zero matrix of size $a \times b.$ Since the matrix $[B_1 \;A_{\vert \mathcal{E}\vert+1}]= 
\left[\begin{matrix} I_{n \vert \mathcal{S}\vert} &\Gamma_1\\0_{n \vert \mathcal{E}\setminus \mathcal{S}\vert \times n \vert \mathcal{S}\vert} &\Gamma_2\end{matrix}\right]$ is full rank, the lower $n \vert \mathcal{E}\setminus \mathcal{S}\vert$ rows should have a rank $n \vert \mathcal{E}\setminus \mathcal{S}\vert$ and hence $\Gamma_2$ is an invertible matrix. Post-multiplying by an invertible matrix $\Gamma_2^{-1}$ does not change the column span of $A_{\vert \mathcal{E}\vert+1}.$ Hence, we can assume $\Gamma_2$ to be an identity matrix.

Let $C$ denote the matrix $[B\; A_{\vert \mathcal{E}\vert+1}].$ The matrix $C$ is of the form given below:

$$C=\left[
\underbrace{\begin{matrix}
I_{n \vert \mathcal{S}\vert} &0_{n \vert \mathcal{S}\vert\times n \vert \mathcal{E}\setminus \mathcal{S}\vert}\\0_{n \vert \mathcal{E}\setminus \mathcal{S} \times n \vert \mathcal{S}\vert} & I_{n \vert \mathcal{E}\setminus\mathcal{S}\vert}
\end{matrix}}_{B} \underbrace{\begin{matrix}
\Gamma_1\\I_{n \vert \mathcal{E}\setminus \mathcal{S}\vert}
\end{matrix}}_{A_{\vert \mathcal{E}\vert+1}}\right].$$ 

 Let $\Theta=[I_{n \vert \mathcal{S}\vert} \; \Gamma_1].$ Let $V_i'', i \in \mathcal{E},$ denote the span of ${(i-1)n+1}^{\text{th}}$ to ${in+1}^{\text{th}}$ columns of $\Theta.$ The claim is that the discrete polymatroid $\mathbb{D}'=\mathbb{D}(V''_1,V''_2,\dotso V''_{\vert \mathcal{E}\vert})$ with rank function $\rho'$ satisfies the condition in Theorem \ref{thm3}.

Assuming that the edges of the network are numbered as in Section \ref{VII A}, define the network-discrete polymatroid mapping $f$ to be $f(\{i\})=i, i \in  \mathcal{E}.$ Clearly $f$ is one-to-one on the elements of $\mathcal{S}$ and hence (DN1) is satisfied. 

To show that (DN2) is satisfied, it needs to be shown that the vector $u=\sum_{i \in  \mathcal{S}}n \epsilon_{i,\vert \mathcal{E}\vert}$ is in $\mathbb{D'}.$ For $X\subseteq   \mathcal{E},$ we have $\vert u(X)\vert=n \vert X'\vert,$ where $X'=X \cap \mathcal{S}.$ We have $\rho'(X)\geq \rho'(X')=dim(\sum_{i \in X'}V''_i)=n \vert X'\vert.$ Hence, from the definition of a discrete polymatroid, it follows that $u \in \mathbb{D}'$ and (DN2) is satisfied.

For all $i \in \lceil \vert \mathcal{S}\vert\rfloor,$ we have $dim(V''_i)=n$ and hence $\rho'(\{i\})=n.$ Since each one of the vector subspaces $V''_i$ is a span of $n$ columns, we have $\max_{i \in \mathcal{E}\setminus \mathcal{S}} \rho'(\{i\})=n.$ Hence (DN3) is satisfied.

Let $\Gamma'_i,i \in \mathcal{E},$ denote a $n \vert \mathcal{E}\vert \times n$ matrix whose $(a,b)^{\text{th}}$ entry is the $(a,(i-1)n+b)^{\text{th}}$ entry of $\Gamma_1$ for $a \in \lceil \vert \mathcal{E}\setminus \mathcal{S}\vert \rfloor$ and zero otherwise. 
From this definition, we have $A_{\vert \mathcal{E}+1\vert}=[(\Gamma'_1+A_{\vert \mathcal{S}\vert+1})\;(\Gamma'_2+A_{\vert \mathcal{S}\vert+2})\dotso\;(\Gamma'_{\vert\mathcal{E}\setminus \mathcal{S}\vert}+A_{\vert \mathcal{E}\vert})].$

For $v \in \mathcal{V}$ and $e\in Out(v)\setminus \delta(v),$ from (NID2) it follows that $A_e$ can be written as a linear combination of the columns of the matrices $A_i, i \in In(v)$ and $A_{\vert \mathcal{E}\vert+1}.$ In other words, for appropriate choices of weight matrices whose columns represent the linear combinations, $A_e$ can be written as,

{\footnotesize
\begin{align}
\nonumber
A_e=&\sum_{i\in In(v)}A_i W_i+A_{\vert \mathcal{E}\vert+1}W'_{\vert \mathcal{E}\vert+1}\\
\nonumber
=&\sum_{i \in In(v)\cap \mathcal{S}}A_i W_i+\sum_{j\in In(v)\cap \mathcal{E}\setminus \mathcal{S}} A_j W_j +\sum_{j\in In(v)\cap \mathcal{E}\setminus \mathcal{S}} (\Gamma'_j+A_j) W'_j \\
\label{Ae_expansion}
&+
\sum_{k\in \mathcal{E}\setminus \mathcal{S}\setminus In(v)\setminus \{e\}}(\Gamma'_k+A_k)W_k'+(\Gamma'_e+A_e)W'_e. 
\end{align}}

The columns of the matrices $A_i,i \in \lceil \vert \mathcal{E}\setminus \mathcal{S}\vert \rfloor,$ form a set of basis vectors for an $n \vert \mathcal{E}\vert \times n\vert \mathcal{E} \vert$ vector space over $\mathbb{F}_q.$ The matrices $\Gamma'_j$ can be written as a linear combination of the columns of the matrices $A_j, j \in  \mathcal{S}.$
Hence, from  \eqref{Ae_expansion}, it follows that $W'_e$ is an identity matrix, $W'_k=0, \forall k\in \mathcal{E}\setminus \mathcal{S}$ and $W'_j=-W_j, \forall j \in  In(v)\cap \mathcal{E}\setminus \mathcal{S}.$ Hence, we have $\Gamma'_e=-\sum_{i \in In(v)\cap \mathcal{S}}A_i W_i -\sum_{j \in  In(v)\cap \mathcal{E}\setminus \mathcal{S}}\Gamma'_j W'_j.$ Hence (DN3) is satisfied for $e\in out(v)\setminus\delta(v).$

For $e\in \delta(v),$ similar to \eqref{Ae_expansion}, from (NID3), $A_e$ can be written as,

{\footnotesize
\begin{align}
\nonumber
A_e=&\sum_{i \in In(v)\cap \mathcal{S}\setminus \{e\}}A_i W_i+\sum_{j\in In(v)\cap \mathcal{E}\setminus \mathcal{S}} A_j W_j \\
\nonumber
&+\sum_{j\in In(v)\cap \mathcal{E}\setminus \mathcal{S}} (\Gamma'_j+A_j) W'_j +
\sum_{k\in \mathcal{E}\setminus \mathcal{S}\setminus In(v)}(\Gamma'_k+A_k)W_k'. 
\end{align}}

From the above equation, it follows $W'_k=0, \forall k\in \mathcal{E}\setminus \mathcal{S}$ and $W'_j=-W_j, \forall j \in  In(v)\cap \mathcal{E}\setminus \mathcal{S}.$ Hence, we have, $A_e=\sum_{i \in In(v)\cap \mathcal{S}\setminus \{e\}}A_i W_i+\sum_{j\in In(v)\cap \mathcal{E}\setminus \mathcal{S}} \Gamma'_j W'_j$ and (DN3) is satisfied for $e \in \delta(v).$
\subsection*{From Theorem \ref{thm3} to Theorem \ref{thm9}}
Let $\mathbb{D}'$ denote the discrete polymatroid in Theorem \ref{thm3} with respect to which the network considered is discrete polymatroidal.As shown in the proof of Theorem \ref{thm3}, one can assume $\lceil r\rfloor$ to be the image of the mapping $f$ and $rank(D')=n \vert \mathcal{S}\vert.$ Let $V_i, i \in \lceil r \rfloor$ denote the vector subspaces which form a representation of $\mathbb{D}$ with $V_i$ being the column span of a matrix $A_i$  of size $n \vert \mathcal{S}\vert\times n.$

Let $B_i,i\in \mathcal{E},$ denote the $n \vert \mathcal{E}\vert$ matrix obtained by taking the $(i-1)n+1^{th}$ to $in^{th}$ columns of the $n  \vert \mathcal{E}\vert \times n \vert \mathcal{E}\vert$ identity matrix. Let $B_{\vert \mathcal{E}\vert+1}$ denote the matrix $\left[\begin{matrix}A_{f(\{\vert \mathcal{S}\vert+1\}}& A_{f(\{\vert \mathcal{S}\vert+2\}}& \dotso &A_{\vert \mathcal{E}\vert}\\&I_{n \vert \mathcal{E}\setminus \mathcal{S}\vert\times n \vert \mathcal{E}\setminus \mathcal{S}\vert}&&\end{matrix}\right].$ Let $V'_i, i \in \vert\mathcal{E}\vert+1,$ denote the column span of $B_i.$ It is claimed that the discrete polymatroid $\mathbb{D}''=\mathbb{D}(V'_1,V'_2,\dotso,V'_{\vert\mathcal{E}\vert+1})$ satisfies the conditions in Theorem \ref{thm9}. (NID1) follow and (NID4) directly from the definition of $\mathbb{D}''.$ 

Let $A'_i$ denote the matrix $\left[\begin{matrix}A_i\\0_{n \vert \mathcal{S}\vert\times n}\end{matrix}\right].$
For $e \in \mathcal{E}\setminus \mathcal{S},$ from (DN2), it follows that $A'_{f(\{e\}}=\sum_{i \in In(head(e)\setminus \mathcal{S}}A'_{f(\{i\})}W_i+\sum_{j \in In(head(e))\cap \mathcal{S}}B_jW'_j.$ The matrix $B_e$ can be written as follows:

\begin{align*}
\nonumber
B_e=&(A'_{f(\{e\})}+B_e)-A'_{f(\{e\})}\\
=&(A'_{f(\{e\})}+B_e)-\sum_{i \in In(head(e)\setminus \mathcal{S}}(A'_{f(\{i\})}+B_i)W_i\\
&+\sum_{i \in In(head(e)\setminus \mathcal{S}}B_iW_i+\sum_{j \in In(head(e))\cap \mathcal{S}}B_jW'_j.
\end{align*}

From the above equation, it follows that $B_e$ can be written as a linear combination of the columns of the matrices $B_i, i \in In(head(e))$ and $B_{\vert \mathcal{E}\vert+1}.$ Hence (NID2) is satisfied.

For $v\in \mathcal{v},$ from (DN2) $B_{\delta(v)}$ can be written as 

\begin{align*}
B_{\delta(v)}=&\sum_{i \in In(v)\setminus \mathcal{S}}A'_{f(\{i\})}W_i+\sum_{j \in In(v)\cap \mathcal{S}}B_jW'_j\\
=&\sum_{i \in In(v)\setminus \mathcal{S}}(A'_{f(\{i\})}+B_i)W_i+\sum_{j \in In(v)\cap \mathcal{S}}B_jW'_j\\
&-\sum_{i \in In(v)\setminus \mathcal{S}}B_iW_i.
\end{align*}

From the above equation, it follows that $B_{\delta(v)}$ can be written as a linear combination of the columns of the matrices $B_i, i \in In(v)$ and $B_{\vert \mathcal{E}\vert+1}.$ Hence (NID3) is satisfied.

This completes the proof of equivalence of Theorem \ref{thm3} and Theorem \ref{thm9}.
\section{Equivalence of Theorem \ref{thm5} and Theorem \ref{thm_IC_NC_DPMD}}
\subsection*{From Theorem \ref{thm_IC_NC_DPMD} to Theorem \ref{thm5}}
Let $\mathbb{D}$ denote the discrete polymatroid satisfying the conditions in Theorem \ref{thm_IC_NC_DPMD}.
Consider the following two cases.\\
\textit{Case 1:} $\rho(f(\{e_{m+1,m+2}\}))\notin \lceil m \rfloor$\\
Without loss of generality, assume $f(\{e_{m+1,m+2}\})=m+1.$ 
Define a new discrete polymatroid $\mathbb{D}'$ on ground set $\lceil m+1 \rfloor$ with rank function $\rho':2^{\lceil m+1 \rfloor} \rightarrow \mathbb{Z}_{\geq 0}$ given by $\rho'(A)=\rho(A).$ The fact that $\mathbb{D}'$ is a discrete polymatroid follows directly from the fact that $\mathbb{D}$ is a discrete polymatroid. 

\textit{Proof of (C1):}
From (IND2), we have $\rho'(\{i\})=n, \forall i\in \lceil m \rfloor.$ From (IND1), since $\sum_{i \in \lceil m \rfloor}n \epsilon_{i,r} \in \mathbb{D},$ we have $\rho'(\lceil m \rfloor)= \rho(\lceil m \rfloor)\leq nm.$ Since $\mathbb{D}'$ is a discrete polymatroid, from (R3) it follows that,

\begin{align*}
\rho'(\lceil m \rfloor)&\leq \rho'(\{m\})+\rho'(\lceil m-1\rfloor)\\
&\leq \rho'(\{m\})+\rho'(\{m-1\})+\rho'(\lceil m-2\rfloor)\\
\dotso &\leq \sum_{i \in \lceil m \rfloor} \rho'(\{i\})=nm.
\end{align*}

Hence, we have $\rho'(\lceil m \rfloor)=nm.$
From (IND2), it follows that $max_{i\in\lceil m+1 \rfloor}\rho'(\{i\})=\max\{n,\rho'(\{m+1\})\}=c.$ Hence, we have $\rho'(\{m+1\})=c.$
Also, we have, $rank(\mathbb{D'})=\rho'(\lceil m+1 \rfloor\})=\rho'(\{\lceil m \rfloor\})=nm$ (follows from (IND4)). 

Condition (C2) in Theorem \ref{thm5} follows directly from condition (IND3) in Theorem \ref{thm_IC_NC_DPMD}.

\textit{Case 2:} $\rho(f(\{e_{m+1,m+2}\}))\in \lceil m \rfloor$\\
Let $V_i, i \in \lceil r \rfloor,$ denote the vector subspaces which form a representation of $\mathbb{D}.$ Following exactly the same steps as in \textit{Case 1}, it can be shown that the discrete polymatroid $\mathbb{D}(V_1,V_2,\dotso V_m,V_f(\{e_{m+1,m+2}\}))$ satisfies the conditions in Theorem \ref{thm5}.
\subsection*{From Theorem \ref{thm5} to Theorem \ref{thm_IC_NC_DPMD}}
Let $\mathbb{D}$ denote the discrete polymatroid on ground set $\lceil m+1 \rfloor$ satisfying the conditions in Theorem \ref{thm5}. It will be shown that $\mathbb{D}$ satisfies all the conditions in Theorem \ref{thm_IC_NC_DPMD} with $f(e_{m+1,m+2})=m+1.$

(IND2) and (IND3) follow directly from (C1) and (C2) respectively. Since $\rho(\lceil m \rfloor)=rank(\mathbb{D})=nm,$ we have,
$$nm=\rho(\lceil m \rfloor)\leq \rho(\lceil m+1 \rfloor)\leq rank(\mathbb{D})=nm.$$ Hence, we have $\rho(\lceil m+1 \rfloor)=nm$ and (IND4) is satisfied.

Consider a set $A'\subseteq \lceil m \rfloor$ and 
$B'=\lceil m \rfloor \setminus A.$ 
From (R3), it can be shown that $\rho(B')\leq\sum_{i \in B'}\rho(\{i\})=n\vert B' \vert=n(m-\vert A' \vert).$ Also from (R3) we have, $\rho(\{A'\})+\rho(\{B'\})\geq \rho(\lceil m\rfloor)=nm.$ Hence, we have $\rho(A')\geq mn-\rho(B') \geq mn-(mn-n\vert A' \vert)=n \vert A'\vert.$ 

Define $x=\sum_{i \in \lceil m \rfloor}n \epsilon_{i,m+1}.$ For $x$ to be a member of $\mathbb{D},$ it should satisfy $\vert x(A) \vert \leq \rho(A),\forall A \subseteq \lceil m+1 \rfloor.$ Define $A'=A \cap \lceil m \rfloor.$ We have $\vert x(A)\vert=n \vert A' \vert.$ Also, we have, $\rho(A)=\rho(A' \cup (A \cap \{m+1\}))\geq \rho(A')\geq n \vert A' \vert=\vert x(A)\vert.$ Hence, (IND1) is satisfied. 

This completes the proof of equivalence of Theorem \ref{thm5} and Theorem \ref{thm_IC_NC_DPMD}.
\label{appC}
\end{appendices}
\section*{Acknowledgment}
This work was supported partly by the Science and Engineering Research Board (SERB) of Department of Science and Technology (DST), Government of India, through J.C. Bose National Fellowship to B. Sundar Rajan.
   
\end{document}